\documentclass[12pt,a4paper]{article}
\usepackage{amsfonts}
 
\usepackage[T1]{fontenc}

\usepackage{graphics}
\usepackage{epsfig}
\usepackage{enumerate}
\usepackage{hyperref}
\usepackage{ulem}
\usepackage{mathtools}
\mathtoolsset{showonlyrefs}
\usepackage{amsmath}
\usepackage{graphicx}
\usepackage{enumerate}
\usepackage{algorithm}
\usepackage{algpseudocode}
\usepackage{amssymb}
\usepackage{color}
\usepackage{amsthm}
\usepackage{physics}
\usepackage{caption}
\usepackage{subcaption}
\usepackage{interval}
\usepackage{multirow}
\intervalconfig{
soft open fences
}
\usepackage{tikz}
\usetikzlibrary{quantikz}
\usepackage{svg}
\usepackage{authblk}

\usepackage[left=0.5in,right=0.5in,top=0.5in,bottom=0.75in]{geometry}
\newcommand{\appref}[1]{\hyperref[#1] {{Appendix~\ref*{#1}}}}
\newcommand{\be}{\begin{eqnarray} \begin{aligned}}
\newcommand{\ee}{\end{aligned} \end{eqnarray}}
\newcommand{\benn}{\begin{eqnarray*} \begin{aligned}}
\newcommand{\eenn}{\end{aligned} \end{eqnarray*}}
\newcommand*{\textfrac}[2]{{{#1}/{#2}}}

\newcommand*{\cA}{\mathcal{A}} 
\newcommand*{\cB}{\mathcal{B}}
\newcommand*{\cC}{\mathcal{C}}
\newcommand*{\cE}{\mathcal{E}}
\newcommand*{\cF}{\mathcal{F}}
\newcommand*{\cG}{\mathcal{G}}

\newcommand*{\cL}{\mathcal{L}}
\newcommand*{\cM}{\mathcal{M}}
\newcommand*{\cN}{\mathcal{N}}
\newcommand*{\cD}{\mathcal{D}}

\newcommand*{\cR}{\mathcal{R}}

\newcommand*{\cP}{\mathcal{P}}

\newcommand*{\cV}{\mathcal{V}}

\newcommand*{\cX}{\mathcal{X}}
\newcommand*{\cY}{\mathcal{Y}}
\newcommand*{\cZ}{\mathcal{Z}}

\newcommand*{\supp}{\mathrm{supp}}

\newcommand*{\fr}[2]{\frac{#1}{#2}}

\newcommand{\bc}{\begin{center}}
\newcommand{\ec}{\end{center}}

\newtheorem{theorem}{Theorem}[section]
\newtheorem{lemma}[theorem]{Lemma}
\newtheorem{proposition}[theorem]{Proposition}

\newtheorem{definition}[theorem]{Definition}

\newtheorem{corollary}[theorem]{Corollary}

\usepackage{amsfonts}

\def\01{\{0,1\}}
\newcommand{\ceil}[1]{\lceil{#1}\rceil}

\newcommand*{\neigh}{\mathrm{neigh}}
\newcommand*{\syn}{\mathsf{Syn}}

\newcommand*{\Vext}{V^{\mathrm{ext}}}
\newcommand*{\Vint}{V^{\mathrm{int}}}
\newcommand*{\Zcirc}{Z_{\circ}}

\newcommand*{\Zpathext}{Z_{\mathrm{ext}}}

\newcommand*{\Ctilde}{\widetilde{\mathcal{C}}}
\newcommand*{\Ctildedual}{\widetilde{\mathcal{C}}^*}

\newcommand*{\Cprime}{\widetilde{\mathcal{C}}} 
\newcommand*{\Cprimedual}{\widetilde{\mathcal{C}}^*} 

\newcommand*{\Esc}{\mathsf{E}_{\mathrm{sc}}}
\newcommand*{\Escdual}{\mathsf{E}_{\mathrm{sc}^*}}
\newcommand*{\Tsc}{\mathsf{T}_{\mathrm{sc}}}
\newcommand*{\Tscdual}{\mathsf{T}_{\mathrm{sc}^*}}
\newcommand*{\Vsc}{\mathsf{V}_{\mathrm{sc}}}
\newcommand*{\Vdangsc}{\mathsf{V}'_{\mathrm{sc}}}
\newcommand*{\Vdangscdual}{\mathsf{V}'_{\mathrm{sc}^*}}
\newcommand*{\Vscdual}{\mathsf{V}_{\mathrm{sc}^*}}

\newcommand*{\Eclsc}{\mathsf{E}_{\mathrm{cl},\mathrm{sc}}}
\newcommand*{\Eclscdual}{\mathsf{E}_{\mathrm{cl},\mathrm{sc}^*}}
\newcommand*{\Tclsc}{\mathsf{T}_{\mathrm{cl},\mathrm{sc}}}
\newcommand*{\Tclscdual}{\mathsf{T}_{\mathrm{cl},\mathrm{sc}^*}}
\newcommand*{\Vclsc}{\mathsf{V}_{\mathrm{cl},\mathrm{sc}}}

\newcommand*{\Vclscdual}{\mathsf{V}_{\mathrm{cl},\mathrm{sc}^*}}

\newcommand*{\Teven}{\mathsf{T}_{\mathrm{even}}}
\newcommand*{\Veven}{\mathsf{V}_{\mathrm{even}}}
\newcommand*{\Vcirceven}{\mathsf{V}^{\circ}_{\mathrm{even}}}
\newcommand*{\Vdangeven}{\mathsf{V}'_{\mathrm{even}}}
\newcommand*{\Eeven}{\mathsf{E}_{\mathrm{even}}}
\newcommand*{\Xcl}{\cX}
\newcommand*{\Adec}{{\Ecldec \cap \cA}}
\newcommand*{\adec}{{a \cap \Ecldec}}

\newcommand*{\Todd}{\mathsf{T}_{\mathrm{odd}}}
\newcommand*{\Vodd}{\mathsf{V}_{\mathrm{odd}}}
\newcommand*{\Vcircodd}{\mathsf{V}^{\circ}_{\mathrm{odd}}}
\newcommand*{\Vdangodd}{\mathsf{V}'_{\mathrm{odd}}}
\newcommand*{\Eodd}{\mathsf{E}_{\mathrm{odd}}}
\newcommand*{\Zcl}{\cZ}
\newcommand*{\Adecdual}{{\mathcal{A}_{\mathrm{dec}^{*}} }}

\newcommand*{\Tgl}{\mathsf{T}_{\mathrm{gl}}}
\newcommand*{\Vgl}{\mathsf{V}_{\mathrm{gl}}}
\newcommand*{\Egl}{\mathsf{E}_{\mathrm{gl}}}

\newcommand*{\Tgldual}{\mathsf{T}^*_{\mathrm{gl}}}
\newcommand*{\Vgldual}{\mathsf{V}^*_{\mathrm{gl}}}
\newcommand*{\Egldual}{\mathsf{E}^*_{\mathrm{gl}}}

\newcommand*{\Tdec}{\mathsf{T}_{\mathrm{dec}}}
\newcommand*{\Vdecint}{\mathsf{V}^{\mathrm{int}}_{\mathrm{dec}}}
\newcommand*{\Vdecext}{\mathsf{V}^{\mathrm{ext}}_{\mathrm{dec}}}
\newcommand*{\Vdec}{\mathsf{V}_{\mathrm{dec}}}
\newcommand*{\Edec}{\mathsf{E}_{\mathrm{dec}}}

\newcommand*{\Tdecdual}{\mathsf{T}_{\mathrm{dec}^*}}
\newcommand*{\Vdecdualint}{\mathsf{V}^{\mathrm{int}}_{\mathrm{dec}^{*}}}
\newcommand*{\Vdecdualext}{\mathsf{V}^{\mathrm{ext}}_{\mathrm{dec}^{*}}}
\newcommand*{\Vdecdual}{\mathsf{V}_{\mathrm{dec}^*}}
\newcommand*{\Edecdual}{\mathsf{E}_{\mathrm{dec}^*}}

\newcommand*{\Tcldec}{\mathsf{T}_{\mathrm{cl},\mathrm{dec}}}
\newcommand*{\Vcldecint}{\mathsf{V}^{\mathrm{int}}_{\mathrm{cl},\mathrm{dec}}}
\newcommand*{\Vcldecext}{\mathsf{V}^{\mathrm{ext}}_{\mathrm{cl},\mathrm{dec}}}
\newcommand*{\Vcldecextleft}{\mathsf{V}^{\mathrm{ext},\mathrm{left}}_{\mathrm{cl},\mathrm{dec}}}
\newcommand*{\Vcldecextright}{\mathsf{V}^{\mathrm{ext},\mathrm{right}}_{\mathrm{cl},\mathrm{dec}}}
\newcommand*{\Vextleft}[1]{\mathsf{V}^{\mathrm{ext},\mathrm{left}}_{\mathrm{cl,dec}}(#1)}
\newcommand*{\Vcldec}{\mathsf{V}_{\mathrm{cl},\mathrm{dec}}}
\newcommand*{\Ecldec}{\mathsf{E}_{\mathrm{cl},\mathrm{dec}}}

\newcommand*{\Tcldecdual}{\mathsf{T}_{\mathrm{cl},\mathrm{dec}^*}}
\newcommand*{\Vcldecdualint}{\mathsf{V}^{\mathrm{int}}_{\mathrm{cl},\mathrm{dec}^*}}
\newcommand*{\Vcldecdualext}{\mathsf{V}^{\mathrm{ext}}_{\mathrm{cl},\mathrm{dec}^*}}
\newcommand*{\Vcldecdualextbottom}{\mathsf{V}^{\mathrm{ext},\mathrm{bottom}}_{\mathrm{cl},\mathrm{dec}^*}}
\newcommand*{\Vcldecdualexttop}{\mathsf{V}^{\mathrm{ext},\mathrm{top}}_{\mathrm{cl},\mathrm{dec}^*}}
\newcommand*{\Vcldecdual}{\mathsf{V}_{\mathrm{cl},\mathrm{dec}^*}}
\newcommand*{\Vextbottom}[1]{\mathsf{V}^{\mathrm{ext},\mathrm{bottom}}_{\mathrm{cl},\mathrm{dec}^{*}}(#1)}
\newcommand*{\Ecldecdual}{\mathsf{E}_{\mathrm{cl},\mathrm{dec}^*}}
\newcommand*{\cLXcl}{\cL_{\mathrm{cl},X}}
\newcommand*{\cLZcl}{\cL^*_{\mathrm{cl},Z}}

\newcommand*{\mmatch}{\mathsf{MinMatch}}
\newcommand*{\incident}{\mathrm{Inci}}

\newcommand{\ztwoinner}[2]{\ensuremath{\left\langle\hspace{-0.7ex}\left\langle{#1},{#2}\right\rangle\hspace{-0.7ex}\right\rangle}}
\newcommand*{\res}{\mathsf{res}}
\definecolor{lightblue}{RGB}{109, 194, 247}
\definecolor{bulklogical}{RGB}{26,11,11}
\definecolor{bulk}{RGB}{217, 212, 212}
\definecolor{lightred}{RGB}{255, 148, 148}
\definecolor{mygreen}{RGB}{0, 100, 0}
\begin{document}

\title{Long-range  data transmission in \\
a fault-tolerant quantum bus architecture}
\author[1,2]{Shin Ho Choe}
\author[1,2]{Robert König}
\affil[1]{Department of Mathematics, Technical University of Munich}
\affil[2]{Munich Center for Quantum Science and Technology (MCQST), Munich, Germany}
\maketitle

\begin{abstract}
We propose a scheme for fault-tolerant  long-range entanglement generation at the ends of a rectangular array of qubits of length~$R$ and a square cross section of size~$d\times d$ with~$d=O(\log R)$.  
Up to an efficiently computable  Pauli correction, the scheme generates a maximally entangled state of two qubits using a depth-$6$  circuit consisting of nearest-neighbor Clifford gates and local measurements only. Compared with  existing fault-tolerance schemes for quantum communication, the protocol is distinguished by its low latency:  starting from a product state, the entangled state is prepared in a time~$O(t_{\textrm{local}})$ determined only by the local gate and measurement operation time~$t_{\textrm{local}}$. Furthermore, the requirements on local repeater stations are minimal: Each repeater uses only~$\Theta(\log^2 R)$~qubits  with 
a lifetime of order~$O(t_{\textrm{local}})$.  
 We prove a converse bound $\Omega(\log R)$ on the number of qubits per repeater among all low-latency schemes for fault-tolerant  quantum communication over distance~$R$. Furthermore, all operations 
within a repeater are local when the qubits are arranged in a square lattice.

The noise-resilience of our scheme relies on the fault-tolerance properties of the underlying cluster state. We give a full error analysis, establishing a fault-tolerance threshold against   general (circuit-level) local stochastic noise affecting preparation, entangling operations and measurements. This includes, in particular, errors  correlated in time and space. Our conservative analytical estimates are surprisingly optimistic, suggesting that the scheme is  suited for long-range entanglement generation both in and between near-term quantum computing devices. 
\end{abstract}

\tableofcontents

\section{Introduction}
The ability of operating on spatially separated qubits has been recognized early on as a key requirement for a working quantum computer~\cite{DiVincenzocriteria,preskillproconquantum}. Long-range quantum communication is also a fundamental primitive at the core of quantum information-theoretic~\cite{Gisineta07,Kimble08,wehnerelkousshanson} and cryptographic~\cite{BB84,GisinTittleZbinden02}
protocols in quantum networks, as well as e.g., distributed quantum computation~\cite{ciracekerthuelgamacchiavello,MeterMunroNemotoKohei08,bealsetal13} or distributed sensing~\cite{komar14,Eldredge}. With increasingly powerful capabilities of present-day devices that can act as local nodes in such a network, finding appropriate protocols for entanglement sharing and quantum communication  over faulty networks remains an active and central theme of present-day research.

The nature of the best protocol to be used depends heavily
on the available hardware and its noise characteristics. Here we consider a monolithic architecture for point-to-point communication in a quasi-1D array of qubits. In terms of more standard terminology in quantum communication, this scheme can be interpreted as consisting of a line of repeater stations. Our proposal  addresses the specific constraints on near-term devices by optimizing the following aspects:
\begin{enumerate}[(i)]
\item minimal number of qubits per repeater station:
As a function of the total communication distance~$R$, every repeater  only needs to manipulate $O(\log^2 R)$ number of qubits. This is in line with the best  currently known proposals for long-range quantum communication. We also provide a converse bound $\Omega(\log R)$ on the number of qubits per repeater for low-latency protocols as considered in this work. 
\item short coherence times: The required coherence time of each qubit  is independent of the overall communication distance~$R$: It only needs to be of the order of the local  operation time~$t_{\textrm{local}}$.  This is the maximal time
needed to perform a single- or two-qubit gate, a single-qubit measurement, or an entangling gate between two neighboring repeaters. 
\item simplicity of involved operations: Our scheme only involves one- and two-qubit gates between nearest-neighbor qubits in the architecture. In particular, assuming the repeater stations are connected appropriately, all operations at each repeater  are geometrically local when the qubits are arranged on a 2D~array. This makes the scheme particularly attractive for, e.g., planar chip-like devices where the connectivity to neighboring devices is established through the third dimension.

Because of the simplicity of the scheme, we can explicitly specify and analyze all involved operations without resorting to fault-tolerant  gadgets (e.g., for encoding and recovery) that are often
treated in a black-box manner. In particular, this means that we can provide a stringent error analysis for fully general circuit-level noise. 

    \item resilience to general errors: we provide a rigorous proof of the robustness of our scheme in the form of a fault-tolerance threshold theorem  allowing for general local stochastic noise. This models not only the commonly studied situation where qubits are depolarized according to identically and independently distributed errors (i.e., product channels), but also captures situations where errors are correlated both in space and time. We assume all involved operations (state preparation, gates and measurements)  to be non-ideal, i.e., affected by noise.

\item minimal latency:  Our protocol minimizes the ``time to entanglement''~$T_{\textrm{ent}}$, the amount of time it takes to establish long-range entanglement from a product state of all qubits. The long-range entangled state is  available after a  time of order~$T_{\textrm{ent}}=O(t_{\textrm{local}})$. 
 In contrast, the quantity~$T_{\textrm{ent}}$ scales at least linearly with the total communication distance~$R$  in several existing protocols based, e.g., on transmitting encoded information successively between neighboring repeaters.

  We note that some of these protocols 
    reach the scaling~$O(t_{\textrm{local}})$ per generated Bell pair asymptotically  when operated continuously, i.e., 
    this is an achievable  rate of communication. In contrast, the quantity~$T_{\textrm{ent}}$ measures the latency when starting to run the scheme, and may be particularly relevant in near-term applications in scenarios where continuous operation is challenging.
\end{enumerate}
We note that our analysis treats Pauli errors and erasure errors on the same footing (similar to e.g., Refs.~\cite{Murali98})  leading to a fault-tolerance threshold for situations where the noise strength of both of these types of noise is of the same order of magitude. Future work may examine using modified decoders~\cite{delfossenickerson21} exploiting the high loss tolerance of topological codes~\cite{barretstace}; indeed, we expect these schemes to be resilient to loss errors with a probability of occurrence closer to the percolation threshold of the square lattice, i.e., much higher than the thresholds we establish for local stochastic Pauli noise.

Because of the simplicity of the involved operations, our scheme is a candidate  architecture for transferring quantum information  between  as well as within quantum devices. That is, it can act as a data bus which can exchange information between various components as first studied in~\cite{brennensongwilliams}. Because it prepares a Bell pair up to an (efficiently computable) Pauli error, it is compatible with e.g., magic-state injection methods to achieve universality: corrections can be propagated ``in software'', i.e., by commuting Pauli operators through a stabilizer circuit without the need for physically applying a correction before proceeding with the computation. 

In the area of quantum computing architectures, our protocol could be especially beneficial in settings where gates or measurements need to performed between spatially distant qubits. Such a need arises for example when using a quantum low-density parity check (LDPC) code whose stabilizer generators are not spatially local. With our scheme, we can  design a quasi-2D architecture that allows for the application of joint measurements on any subset of size $\ell=O(1)$ of a total of~$n$ qubits, see Fig.~\ref{fig:busarchitecture2D}.  Each such $\ell$-qubit measurement is executed with constant fidelity in a constant amount of time (up to Pauli correction whose classical postprocessing can typically be deferred to a later stage). The architecture requires $O(n \cdot \mathsf{polylog}(n))$ qubits in total. When combined with the recently discovered good quantum LDPC codes that have parameters close to optimal~\cite{panteleevkalachev}, this logarithmic overhead unfortunately  only leads to an encoding rate that decays as an inverse polylogarithmic function. In the near term, however, this still beats  e.g.,  surface-code based approaches~\cite{defolsseiyerpoulin,Litinski2019gameofsurfacecodes} that have an inverse polynomial encoding rate, although the latter have the advantage of being strictly 2D and not quasi-2D as our scheme.
\begin{figure}
    \centering
    \includegraphics[width=15cm]{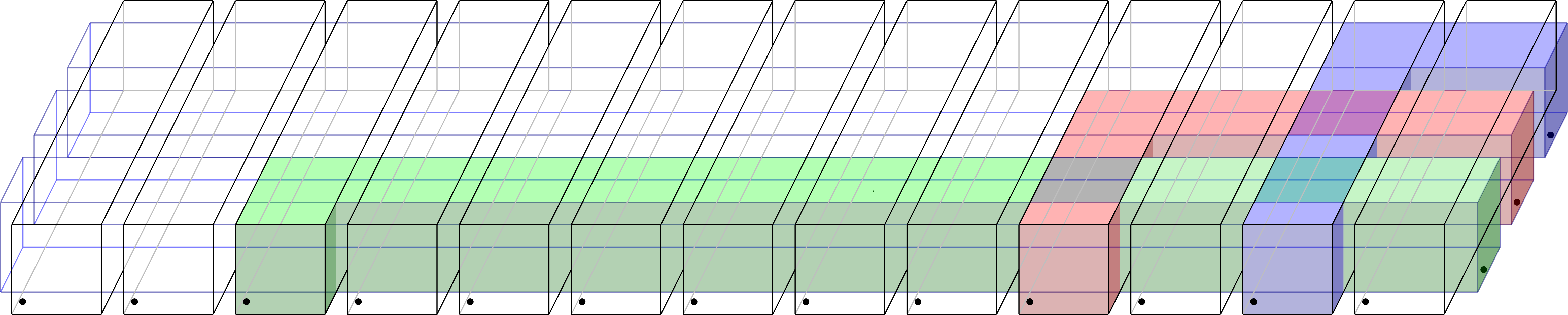}
    \caption{
    A quasi-2D network of quantum buses. This architecture allows to perform a joint measurement on any subset of $\ell$~qubits (illustrated here for~$\ell=3$) with constant fidelity by ``busing" the corresponding qubits to the same location before the measurement, and subsequently transferring the qubits  back.
    The transfer is achieved by teleporting the qubits using the long-range entanglement established with our protocol on each bus.\label{fig:busarchitecture2D}}
\end{figure}

\subsubsection*{Prior work on entanglement generation from cluster states}
 Our work builds on the fault-tolerance properties of the 3D~cluster state. In seminal work~\cite{raussendorfLongrangeQuantumEntanglement2005}, Raussendorf, Bravyi and Harrington proved that  for temperatures below some threshold, the thermal state of the cluster state Hamiltonian on a simple 3D~cubic lattice has localizable entanglement: applying single-qubit Pauli measurements
 to each ``bulk'' qubit  results -- up to Pauli corrections determined by the measurement outcomes  -- in an encoded Bell state
of 2D~surface codes lying on two opposing faces of the cube.  Since the initial thermal state is the cluster state with independent  and identically distributed Pauli-$Z$-errors on each qubit, and the cluster state can be created in constant depth with local operations, this provides a scheme for long-range (encoded) entanglement generation by a constant-depth circuit with local gates. While the assumed (physically motivated) noise model is somewhat restricted, the authors of~\cite{raussendorfLongrangeQuantumEntanglement2005} also point out that more general errors on the initial cluster state such as independent single-qubit depolarizing noise can be tolerated when using ideal (non-faulty) operations at the two surfaces.

Entanglement generation in a cubic lattice can indeed be achieved under more general noise:  As shown in~\cite{bravyiQuantumAdvantageNoisy2020}, the Raussendorf-Bravyi-Harrington scheme produces a logical  Bell pair even under general local stochastic noise below some threshold. Local stochastic errors model general circuit-level Pauli noise that can include correlations in time and space. As a special case, this encompasses 
independent and identically distributed (i.i.d.) Pauli noise on the initial cluster state, as well as e.g., measurement errors. We discuss this error model  and  the exact statement of~\cite{bravyiQuantumAdvantageNoisy2020} below. It shows that under local stochastic noise, it is possible to generate an encoded Bell pair which is corrupted by a residual local stochastic error. The latter can be  corrected by surface-code decoding.

The scheme of~\cite{raussendorfLongrangeQuantumEntanglement2005}  produces a surface-code encoded Bell pair, i.e., the resulting entanglement lives in a logical $2$-qubit subspace of a many-qubit system. In applications such as the one considered in~\cite{bravyiQuantumAdvantageNoisy2020}, this is a desirable feature: the encoded information can be operated on without decoding, thus  preserving fault-tolerance properties. In a communication scenario, however, it may be more desirable to establish entanglement between two physical qubits instead. A scheme for achieving this kind of localization of the entanglement was  proposed in~\cite{perseguersFidelityThresholdLongrange2010} by modifying the approach of~\cite{raussendorfLongrangeQuantumEntanglement2005}: Here a different single-qubit measurement pattern is used at the two (``surface code'') boundaries. This results in entanglement between two distinguished qubits in each surface. For the problem of establishing a constant-fidelity Bell pair in a cube of size~$R\times R\times R$, the work~\cite{perseguersFidelityThresholdLongrange2010} provides analytical and numerical threshold estimates in the case of i.i.d.~Pauli-$X$ and Pauli-$Z$ errors in the cluster state. (It was argued in~\cite{perseguersFidelityThresholdLongrange2010} that this error model captures a situation where initially, independent noisy Bell states are shared along each edge of the cubic lattice, and all subsequent (single-qubit) operations are ideal, i.e., noise-free.) In particular, it is shown that there is a threshold error rate such that constant-fidelity two-qubit entanglement (i.e., independent of~$R$) is established for any error rate below that threshold.

To achieve long-distance communication, one can use cluster states on lattices of the form~$d\times d\times R$ with $R\gg d$, i.e., lattices that  are elongated in the direction where quantum information is to be transferred. In this work, we denote a lattice of such form by $\cC[d \times d \times R]$, see Section~\ref{sec:entanglementgenerationnoisy} for a detailed definition. The corresponding setup is quasi-1D, with $O(d^2)$~qubits arranged on a square on each ``slice'' along this direction. This idea has been explored numerically in the case of i.i.d.~Pauli-$Z$ noise in~\cite{raussendorfLongrangeQuantumEntanglement2005}. For sufficiently small error rates, a fidelity of the encoded (error corrected) scaling as
$F\sim \exp\left(-R\kappa_1\exp(-d\kappa_2)\right)$ for some constants $\kappa_1,\kappa_2>0$ was observed. These simulations were carried out for period boundary conditions. Assuming the surface codes have similar behavior as toric codes away from the threshold, this implies that using~$d=\Omega(\log R)$ (i.e., a polylogarithmic number~$n=O(d^2)$ of qubits per repeater station) is sufficient to achieve constant fidelity over a distance~$R$. In a similar vein, the results of~\cite{perseguersFidelityThresholdLongrange2010} are stated for cubes (i.e., lattices of the form~$R\times R\times R$), but the author suggests using lattices of the form~$d\times d\times R$ with $d=\Theta(\log R)$.  Here we establish rigorous bounds on the fault-tolerance threshold for this scenario for a more general noise model.

\subsubsection*{Surface codes for communication and as topological quantum memories}
Our work is closely related to~\cite{perseguersFidelityThresholdLongrange2010}: indeed, the cluster state and measurement pattern used in our scheme is identical to that of~\cite{perseguersFidelityThresholdLongrange2010}.  As we explain below, this measurement strategy and the associated correction operation can be interpreted as a fault-tolerant single-shot decoding procedure for the surface code.
 In fact, the corresponding  measurement pattern for the surface code has also previously been used in~\cite{Litinski2019magicstate} to realize a fault-tolerant $T$-measurement. A similar  triangular pattern has been applied in Refs.~\cite{mazurekLongdistanceQuantumCommunication2014,lodygaSimpleSchemeEncoding2015} for encoding, i.e., state preparation, for the surface code, and more specifically, for injecting  a magic state into the surface code~\cite{Li_2015}.

Surface codes have found many uses  in  information-processing and communication settings. While some of these proposals are superficially similar to our work, they typically have different objectives. In particular, they often involve operations over an extensive amount of time, preventing realization by a constant-depth circuit:  Consider  first the fault-tolerant encoding/decoding scheme for the 2D surface code studied in~\cite{lodygaSimpleSchemeEncoding2015}. This scheme establishes entanglement in time, i.e., gives a complete prescription (including preparation and readout) of how the surface code can be used to preserve quantum information. Similar to the seminal  work~\cite{dennisTopologicalQuantumMemory2002} showing that the surface code can be used as a quantum memory  even with noisy syndrome measurements, this leads to a matching problem on a 3D~lattice in spacetime. This is conceptually different from our procedure: Our decoding procedure involves a ``spacelike'' 3D~lattice, and the entire process is achieved in a constant  time. (Ref.~\cite{lodygaSimpleSchemeEncoding2015} also gives a lower bound on the achievable fidelity in terms of ``circuit volume'' when using concatenated codes. The non-locality and complexity of the involved operations is mentioned as a challenge.)

The proposal of~\cite{FowleretalSurfaceCodeCommunication} also relies on surface codes: Several corresponding code patches are placed  next to each other along the communication direction and an encoded qubit is teleported through this architecture. Similar to a topological quantum memory with noisy syndrome measurements, the scheme relies on interaction and measurement patterns
that are repeated a number of times scaling linearly with the code distance.

\subsubsection*{Local stochastic noise and Clifford circuits}
We consider quantum circuits subject to local stochastic noise. This is a general notion of errors introduced in~\cite{gottesmanlocalstochastic}
to model situations where the noise is ``locally decaying'' but otherwise arbitrary. In particular, this model includes
not only i.i.d.~noise on each qubit, but also adversarial settings where errors can have unspecified correlations both in space and time (i.e., at different circuit locations).

Here we follow the formalism used in~\cite{fawzietalFOCS2018}. When implementing an ideal circuit under local stochastic noise,
all qubits experience  a random Pauli error  at each time step (gate layer) of the circuit. The only assumption we make is that high-weight errors are exponentially suppressed as quantified by a parameter~$p\in [0,1]$ called the noise strength. To be specific, consider a quantum circuit using~$n$ qubits.
A Pauli error~$E$ is a  random variables on the $n$-qubit Pauli group. Its support~$\supp(E)\subseteq [n]$ is the (random) subset of qubits it acts on non-trivially, i.e., with either~$X$, $Y$, or~$Z$. We call a Pauli error~$E$ a local stochastic error with strength~$p\in [0,1]$, written~$E\sim\cN(p)$, if  any subset of~$k$ qubits is affected by~$E$ with probability at most~$p^k$, i.e., if
\begin{align}\label{eq:introlocalstochasticnoisecondition}
  \Pr\left[F\subseteq\supp(E)\right]&\leq p^{|F|}\qquad\textrm{ for any subset }\qquad F\subseteq [n]\ .
\end{align}

Consider an ideal depth-$d$ circuit that starts in a product state~$\ket{0^n}:=\ket{0}^{\otimes n}$ of $n$~qubits, applies a
unitary~
\begin{align}
    U=U_d\cdots U_1
\end{align}
consisting of~$d$ gate layers~$U_j$ (each of which is a depth-$1$ circuit consisting of one- and two-qubit gates) and subsequently measures a subset of the qubits in the computational basis. We call a noisy execution of this circuit with parameter~$p$
the same kind of process, but with~$U$ replaced by
\begin{align}
U_{\textrm{noisy}}&=E_{d+1}E_{d}U_d\cdots E_2 U_2 E_1 U_1E_0\  ,\label{eq:noisycircuitrealizationex} 
\end{align}
where each error~
\begin{align}
  E_j\sim \cN(p)\qquad\textrm{ for  }\qquad j\in \{0,\ldots,d+1\} \label{eq:localerrorassumptiondem}
\end{align}
is a local stochastic error with strength~$p$. This models both preparation errors (described by~$E_0$), gate execution errors (described by~$E_j$ for $j=1,\ldots,d$) and measurement errors, i.e., bit flips in the measurement outcome (described by~$E_{d+1}$). Importantly, the random variables~$\{E_j\}_{j=0}^{d+1}$ need not be independent, capturing situations where errors are correlated in time. Indeed, condition~\eqref{eq:localerrorassumptiondem} merely imposes constraints on the marginal distributions of the errors~$\{E_j\}_{j=0}^{d+1}$.

Put succintly, the goal of fault-tolerance is to design an ideal circuit in such a way that any noisy implementation of the circuit  still realizes the desired functionality, assuming that
the error strength~$p$ is below some threshold. The noise model considered here makes the corresponding analysis  tractable when working with Clifford circuits: Any local stochastic error can be commuted forward or backward (in time) through  such a circuit while preserving the property of being a local stochastic error. In addition, products of local stochastic errors are themselves local stochastic errors, meaning that errors can be accumulated, i.e., combined into a single local stochastic error. These manipulations change the error strength. More precisely, it is possible to establish conservative upper bounds on the error strength of the resulting error, see~\cite[Lemma~11]{bravyiQuantumAdvantageNoisy2020} for a formal statement of this calculus of local stochastic errors. Importantly, however, these manipulations do not change the nature of the errors: This is  because of the minimal assumptions used in the definition of local stochastic noise.

We note that the practice of commuting Pauli errors forward or backward through a circuit composed of Clifford gates is common throughout the literature. However, in some  earlier work on quantum repeaters, 
the fact that such operations may introduce correlations between errors  when starting from an i.i.d.~error model is sometimes neglected. The concept of local stochastic errors and the corresponding calculus elegantly sidesteps this difficulty: For any Clifford circuit
of constant depth~$d$, it is sufficient to establish resilience against a single local stochastic error at any time step~$t\in\{0,\ldots,d+1\}$ in the circuit, rather than considering the $d+1$-tuple~$(E_0,\ldots,E_{d+1})$: Indeed, any noisy circuit~\eqref{eq:noisycircuitrealizationex}  behaves identically to a circuit of the form
\begin{align}
U_{\textrm{noisy}}&=U_d\cdots U_{t+1} E U_t\cdots U_1 \  ,\label{eq:noisycircuitrealizationexcommuted} 
\end{align}
with a certain local stochastic error~$E\sim\cN(p^{O(1)})$.  A corresponding  upper bound on the error strength follows by applying commutation relations and combining errors according to rules such as those given in~\cite[Lemma~11]{bravyiQuantumAdvantageNoisy2020}. 

\subsubsection*{Single-shot encoded Bell state preparation and surface code readout}
To exemplify this line of reasoning, consider the (ideal) circuit studied~\cite{bravyiQuantumAdvantageNoisy2020}. Starting from a product state~$\ket{0^n}$, it applies a
depth-$6$ circuit~$W$ with nearest-neighbor Clifford gates  to prepare the cluster state~$W\ket{0^n}$ on a lattice~$\cC[d\times d\times d]$. Subsequently, a subset~$\cA$ of qubits is measured in the computational basis, yielding an outcome (string~$s$). Up to a Pauli correction~$\mathsf{Rec}(s)$ determined by the measurement outcome~$s$, this circuit prepares logical Bell state~$\overline{\Phi}$ encoded in two surface codes on the two boundaries of the lattice, that is, we have 
\begin{align}
\mathsf{Rec}(s)(\bra{s}_{\cA}\otimes I_{\cA^c})W\ket{0^n}&\propto\ket{\overline{\Phi}}\ .
\end{align}
Now consider a noisy execution of this circuit. There are three sources of error: Errors in the preparation of the initial state~$\ket{0^n}$, errors during the execution of the circuit~$W$, and errors affecting the single-qubit measurements.  The former two types of errors can be commuted forward in time past the last gate layer of~$W$.  Measurement errors (where an error on a single qubit measurement corresponds to a Pauli-$X$-error before application of the ideal measurement) can be commuted backwards in time. Combining the commuted versions of these errors, it thus suffices to consider the case where  a single local stochastic error~$E$ is applied after the execution of~$W$. All other operations may be assumed to be ideal, i.e., noise-free: The overall process is equivalent to  measuring (using an error-free measurement) a noisy cluster state~$EW\ket{0^n}$, corrupted by a local stochastic error~$E$.

What does a noisy execution of the circuit achieve? It turns out that the procedure still prepares the logical Bell state~$\overline{\Phi}$ up to the known Pauli correction~$\mathsf{Rec}(s)$ (depending on the measurement result~$s$) and an additional residual error~$\mathsf{Rep}(E)$ (depending on the local stochastic error~$E$). The main technical result established in~\cite{bravyiQuantumAdvantageNoisy2020} for this ``single-shot preparation procedure'' shows that this residual error itself is local stochastic, and its error strength can be related to the error strength~$p$ of the original local stochastic error~$E\sim\cN(p)$. Let us state the corresponding result since it motivates our construction. 

\begin{theorem}[Single-shot encoded Bell state preparation~\cite{bravyiQuantumAdvantageNoisy2020}]\label{thm:singleshotbellstateprep}
Suppose $d\geq 4$. Then there is a threshold probability~$p_0$ such that the following holds: Let $\cC[d\times d\times d]$ be the cubic lattice of linear size~$2d-1$ with one qubit per site. Let  $E\sim\cN(p)$ be a local stochastic error with error strength~$p<p_0$. 
Then there is a local stochastic error~$\mathsf{Rep}(E)$ with 
\begin{align}
\mathsf{Rep}(E) \sim \cN(11p^{1/128})\ \label{eq:problRepE}
\end{align}
such that
\begin{align}
\mathsf{Rep}(E)\mathsf{Rec}(s)(\bra{s}_\cA\otimes I_{\cA^c})EW\ket{0^n}\propto \overline{\Phi}_{\cA^{c}}\ 
\end{align}
with certainty, where the probability is taken over the choice of~$E$ and the syndrome (measurement result)~$s$.
\end{theorem}
We note that Theorem~\ref{thm:singleshotbellstateprep}  primarily constitutes   a proof-of-principle, giving a precise formulation of the idea that the final state is the encoded Bell state up to a locally decaying error. Both the threshold error strength obtained, as well as  the exponent in~\eqref{eq:problRepE} are likely to be more conservative bounds than necessary in reality, and are expected to be too demanding for experimentally relevant parameters ranges. Nevertheless, Theorem~\ref{thm:singleshotbellstateprep} shows that the single-shot encoded Bell state preparation procedure can be combined with suitable surface code protocols that are resilient to local stochastic errors. Indeed, this is the way the cluster state is used in the application of~\cite{bravyiQuantumAdvantageNoisy2020} and follow-up work~\cite{faulttolerantinteractiveadvantage,secondprotocolnoisy}.

One such noise-resilient surface code protocol is the single-shot decoding protocol discussed in~\cite{bravyiQuantumAdvantageNoisy2020}.  It shows that measurement of the logical Pauli-$\overline{Z}$ operator in the surface code can be achieved fault-tolerantly by local operations even in the presence of local stochastic noise. Let us paraphrase the corresponding statement here.
\begin{theorem}[Single-shot logical~$\overline{Z}$-measurement~\cite{bravyiQuantumAdvantageNoisy2020}]
\label{thm:singleshotlogicalZsurface}
Consider a surface code of distance~$d\geq 7$
with~$n=2d^2-2d+1$ qubits.  
Let $E\sim \cN(q)$ be a local stochastic error with strength~$q\leq 0.01$. Then there is a function~$\mathsf{Dec}:\{0,1\}^n\rightarrow\{0,1\}$ such  that the following holds. Let $\overline{\Psi}$ be an arbitrary encoded state. Assume we measure every qubit of the state~$E\overline{\Psi}$ in the computational basis obtaining the outcome~$x\in \{0,1\}^n$. 
Then
\begin{align}
\left|
\Pr_{E,x}\left[
(-1)^{\mathsf{Dec}(x)}=1
\right]- p_{\overline{Z}}(\overline{\Psi})
\right| &\leq 3 \exp(-0.2d)
\end{align}
 where 
 \begin{align}
 p_{\overline{Z}}(\overline{\Psi})&:=\frac{1}{2}\left(1+\bra{\overline{\Psi}}\overline{Z}\ket{\overline{\Psi}}\right)\ .
 \end{align}
\end{theorem}
\noindent (We note that without loss of generality, the error~$E$ can be assumed to be consisting of Pauli-$X$ operators only.)

Since~$p_{\overline{Z}}(\overline{\Psi})$ is the probability of observing the outcome~$+1$ when measuring the logical~$\overline{Z}$-operator, Theorem~\ref{thm:singleshotlogicalZsurface}  provides a way of fault-tolerantly performing a measurement of this operator. While we do not  directly use this result, it is part of the motivation for our single-shot decoding procedure and its analysis.

\subsubsection*{Prior work on quantum repeaters}
A comparative study of different approaches to long distance quantum communication is provided in~\cite{Murali98} where the authors discuss three generations of quantum repeater protocols:

  The first generation of quantum repeaters  is based on heralded entanglement generation between neighboring stations or repeaters (where entanglement is generated using post-selection on suitable measurement outcomes) combined with heralded entanglement purification (where multiple imperfect Bell pairs are used to probabilistically produce higher-fidelity Bell pairs)~\cite{qrgen1.1,qrgen1.2}. These steps are repeated recursively at different distance scales. First-generation protocols of this kind
require at a number of qubits polynomial in the total communication distance~$R$ at each repeater station. The rate (i.e., number of Bell pairs established per channel use) is an inverse polynomial in~$R$, and the total time per Bell pair is given determined by~$\max\{R/c,t_0\}$, where~$t_0$ is the gate time and~$c$ the speed of light (corresponding to classical data transfer).

Second-generation protocols make use of heralded entanglement generation to counteract loss errors, and use  quantum error-correcting codes to deal with operation errors: long-range encoded entanglement is generated by (encoded) entanglement swapping~\cite{qrgen2.1,qrgen2.2,qrgen2.3}. This leads to a reduction in classical two-way communication previously required in the entanglement purification step. The communication rate is improved to an inverse polylogarithmic function of~$R$, whereas the overall time-scale still scales as~$\max\{R/c,t_0\}$.

  Protocols of the third generation use quantum error correcting codes   to address not only operation errors, but also qubit loss~\cite{FowleretalSurfaceCodeCommunication,muralidharanUltrafast,muralidharanErasureErrors,munroQCwithoutthenecessityofquantummemories}. This family of protocols
   proceeds by communicating encoded quantum information: Repeaters apply a recovery operation to each codeblock before retransmitting the entire block to the next repeater. 
   Because this family of protocols can be realized using one-way communication only, it requires little time per transmitted qubit: The time taken per transmitted qubit is  only limited by~$O(t_0)$, where~$t_0$ is the local gate operation time.     This advantage comes at the cost of requiring relatively high-fidelity quantum gates to realize the error correction operation at each repeater. The achieved rate in this generation of protocols is an inverse polylogarithmic function of~$R$.

Our work is a close relative to these third-generation protocols: indeed, one interpretation  of the corresponding architecture is that it encodes logical information into a surface code at each repeater station. We emphasize, however, that this only motivates the scheme similar to the way e.g., cluster-state measurement-based quantum computation is related to braiding surface code anyons. For example, the scheme does not involve intermediate encoding/recovery operations associated with these surface codes.

Purely error-correction based schemes such as those of the third generation have a loss-tolerance limited by 50\% because of the no-cloning theorem.  To sidestep this limitation, the more recent work~\cite{zwergerLongrangeBigQuantumdata2018} combines error correction with hashing-based entanglement purification. In this scheme, intermediate error correction operations are realized based on stabilizer (resource) states prepared at each repeater. The proposed scheme achieves a constant rate of qubit transmission. Here error correction is based on measurement-based gate teleportation using special (graph)  states encoding the recovery map at each repeater location.
While such schemes are ultimately desirable, the generation of these resource states is arguably more challenging to realize with near-term devices than the protocol discussed here.

\section{Our contribution: Main ideas and results}

\subsubsection*{A threshold theorem for long-range entanglement generation}
In this work, we present a protocol for generating a long-range entangled two-qubit Bell state in a fault-tolerant manner in the presence of local stochastic noise. It is based on a 3D cluster state on a lattice of the form~$d\times d\times R$, see Fig.~\ref{fig:longclusterstate}.

\begin{figure}[htbp]
  \centering
  \includegraphics[width=0.9\textwidth]{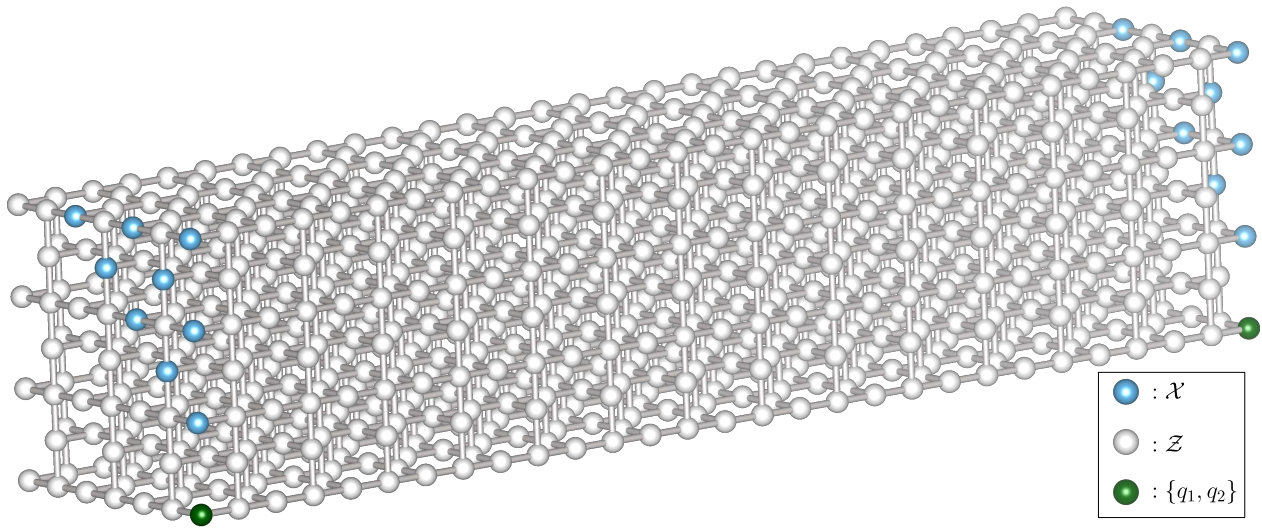}
  \caption{A cluster state on the lattice~$\cC[d\times d \times R]$ with $R \gg d$. Each sphere represents a qubit. Also indicated is the measurement pattern used in our protocol, consisting of single-qubit Pauli-$X$ and Pauli-$Z$ measurement. Entanglement is established between qubits~$q_1$ and $q_2$.}
  \label{fig:longclusterstate}
\end{figure}

To state our main result, we use the following notations: We let~$\cC$ denote the set of qubits  on the cluster state lattice~$\cC[d\times d\times R]$. For a subset~$\Omega\subseteq \cC$, we write
$\ket{0^\Omega}=\bigotimes_{c\in\Omega}\ket{0}_c$ for the all-zero basis state on the qubits~$\Omega$. Similarly, we write $H(\Omega):=\prod_{c\in\Omega} H_c$ for the product of single-qubit Hadamard operators on each qubit of~$\Omega$. When measuring a subset~$\Omega$ in the computational basis, we denote the measurement result by~$s\in \{0,1\}^{\Omega}$, where  $s=(s_c)_{c\in\Omega}$ consists of one measurement outcome bit~$s_c$ associated with each qubit~$c\in\cC$. The standard two-qubit Bell state is denoted
\begin{align}
    \ket{\Phi}=\frac{1}{\sqrt{2}} \left(\ket{00}+\ket{11}\right)\ ,
\end{align}
and the two-qubit Bell basis~$\{\Phi_{(\alpha,\beta)}\}_{(\alpha,\beta)\in\{0,1\}^2}$ consists of the Bell states
\begin{align}
    \Phi_{(\alpha,\beta)}&=(I\otimes Z^{\alpha}X^{\beta})\Phi\ .\label{eq:defbellstate}
\end{align}
Our main result is the following fault-tolerance threshold theorem:
\begin{theorem}[Long-range entanglement generation]\label{thm:mainentanglement}
Let $d\in\mathbb{N}$ with $d\geq 3$ and  $R \in \mathbb{N}$ be odd with $R \geq 3$.
Assume that 
\begin{align}
  R\leq  \frac{1}{d} \left(\frac{1}{10\sqrt{p}}\right)^{d-2} \ .\label{eq:Rupperboundmscaling}
\end{align}
There is a partition 
\begin{align}
    \cC=\cX\cup\cZ \cup \{q_1,q_2\}
\end{align}
of the set of qubits~$\cC=\cC[d\times d\times R]$ into disjoint sets~$\cX$, $\cZ$ and a set~$\{q_1,q_2\}$ consisting of two qubits located at opposite ends of~$\cC$, with coordinates~$u_3=1$ and $u_3=R$, respectively, as well as two  efficiently computable functions
\begin{align}
\begin{matrix}
\alpha: & \{0,1\}^{\cC\backslash \{q_1,q_2\}} & \rightarrow &\{0,1\}\\
\beta: & \{0,1\}^{\cC\backslash \{q_1,q_2\}} & \rightarrow &\{0,1\}
\end{matrix}
\end{align}
such that the following holds. Consider the following protocol:
\begin{enumerate}[(i)]
\item\label{it:clusterstateprepstepmainprotocol}
Prepare the cluster state~$W_{\cC}\ket{0^\cC}$
by applying the local  depth-$6$ Clifford circuit~$W_{\cC}$ 
to qubits associated with the cluster state lattice~$\cC=\cC[d\times d\times R]$.
\item
Measure each qubit except the two  distinguished qubits~$q_1,q_2$ according to the following pattern:
\begin{enumerate}[(a)]
\item
Each qubit~$c\in\cX$ is measured in the Hadamard basis
\item
Each qubit~$c\in \cZ$ is measured in the computational basis.
\end{enumerate}
Let $s\in \{0,1\}^{\cC\backslash \{q_1,q_2\}}$ be the (collection)  of measurement outcomes.
\end{enumerate} Assume that the state after step~\eqref{it:clusterstateprepstepmainprotocol} is corrupted by a local stochastic error~$E\sim\cN(p)$ of strength
\begin{align}
p\leq \frac{1}{5006}\approx 2\times 10^{-4}\ .
\end{align}
Let
\begin{align}
\psi_E(s)\propto (I_{q_1,q_2}\otimes \bra{s}H(\cX)) EW_\cC\ket{0^\cC}
\end{align}
be the post-measurement state on qubits~$\{q_1,q_2\}$. Then  the post-measurement state on qubits $q_1,q_2$ is the Bell basis state~$\Phi_{(\alpha(s), \beta(s))}$ 
 with probability at  least
\begin{align}
    \Pr_{E,s}\left[\psi_E(s)\propto \Phi_{(\alpha(s),\beta(s))}\right]&\geq 1- 5006p\ .
    \end{align}
\end{theorem}

Recall that the cluster state lattice~$\cC[d\times d\times R]$
has $m=2d^2-2d+1$~qubits 
in each ``sheet'' (with fixed coordinate~$u_3$) along the communication direction. This is the number of qubits that our protocol uses at each repeater station. 
Eq.~\eqref{eq:Rupperboundmscaling} thus implies that, to produce constant-fidelity two-qubit entanglement (under constant-strength noise), our protocol uses~$d=\Theta(\log R)$ or  \begin{align}
  m&=\Theta\left( \log^2 R\right)\ . \label{eq:polylogscalingexplicitprotocol}
\end{align}
 That is, the number~$m$ of qubits per repeater station is polylogarithmic in the communication distance~$R$ for our protocol. 
 We give a converse bound, see Theorem~\ref{thm:nogotheoremmain} below.

The functions~$\alpha$ and $\beta$ involve solving minimum matching problems on certain sublattices of~$\cC$. With Edmond's algorithm~\cite{Edmonds1965,lawler2001,gabow1976,wangThresholdErrorRates2009}, 
they can be evaluated in a time~$O(d^6 \cdot R^3)$. In the context of surface code decoding, it has been shown that parallel processing leads to a polylogarithmic runtime instead~\cite{duclos-cianciFastDecodersTopological2010, devittClassicalProcessingRequirements2010}. We expect that similar approaches are applicable here.

Our scheme  is essentially identical to the procedure of~\cite{perseguersFidelityThresholdLongrange2010}, but constructed in a modular manner: We introduce a new primitive we call single-shot decoding of surface codes. Combining this with the Raussendorf, Bravyi and Harrington-scheme~\cite{raussendorfLongrangeQuantumEntanglement2005}, i.e., applying single-shot decoding to each of the two surface codes yields our scheme. Our work thus gives a reinterpretation of the scheme~\cite{perseguersFidelityThresholdLongrange2010} along with a more general error analysis.

\subsubsection*{Optimality: A converse to the threshold theorem}
We study to which extent the protocol provided by Theorem~\ref{thm:mainentanglement} is optimal among protocols  with low latency, i.e., where the entanglement is created after a constant amount of time (modulo local correction operations).   
Specifically, we consider protocols that are given by constant-depth quantum circuits, and ask if the relationship between the number~$m$ of qubits at each repeater station and the communication distance~$R$ 
can be better than polylogarithmic (cf.~\eqref{eq:polylogscalingexplicitprotocol}). We find that this is not the case, i.e., the  scaling of our protocol is essentially optimal.

The class of protocols we consider are given by constant-depth circuits. A depth-$\Delta$ circuit for entanglement generation  \begin{enumerate}[(i)]
\item
starts with a product state among all qubits, where each repeater (site) has $m$~qubits,
\item
applies a depth-$\Delta$ circuit consisting of single-qubit measurements, single-qubit gates and two-qubit gates between qubits at a repeater (site) and/or between nearest-neighbor sites. Constant-depth here means that~$\Delta$ is a constant (independent of~$R$).
\item
returns a state on two qubits~$\{q_1,q_2\}$ located at the two ends of the communication line at distance~$R$. This state is supposed to be the  Bell state~$\Phi$  up to a local (i.e., tensor product) correction that can be computed from the measurement outcomes.
\end{enumerate}
We allow for adaptivity, i.e., unitaries can be classically controlled by measurement outcomes obtained during the execution of the circuit. Our  main no-go result applies to all such protocols irrespective of whether or not the classical computations determining these adaptive actions are efficient.

Our protocol is of this kind (for qubits arranged on the cluster lattice~$\cC[d\times d\times R]$). Also of this kind are protocols proposed in the seminal work~\cite{acinciraclewenstein} on ``entanglement percolation''. The work~\cite{acinciraclewenstein} considers a unitary error model where  the 
resource state has a non-maximally entangled Bell state on each link. In addition to proposing protocols for this setting for 2D networks, it is shown in~\cite{acinciraclewenstein} that for a 1D~chain of qubits, long-distance entanglement cannot be created by any single-shot protocol: Assuming any constant error per link, the fidelity of the resulting Bell state decays exponentially with distance.

We prove a converse to our ``achievability'' Theorem~\ref{thm:mainentanglement}.  It generalizes the statement of~\cite{acinciraclewenstein} to a setting where instead of a  single (pair) of qubits per site (along the communication line), each site can have a polylogarithmic number of qubits. In addition, we allow constant-depth circuits instead of local measurements only. We find a local stochastic error of constant strength which leads any such protocol to fail if the communication distance~$R$ is too large compared to the number~$m$ of bits at each repeater station. That is, we show the following, see Corollary~\ref{cor:separabilityconverse} for a detailed version of this bound.

\begin{theorem}[Converse for low latency schemes]\label{thm:nogotheoremmain}
Let~$\pi$ be a protocol for long-range entanglement generation over distance~$R$ which 
\begin{enumerate}[(i)]
\item
uses $m$~qubits per repeater,
\item
is realized by a constant-depth (adaptive) circuit 
\item
is resilient to arbitrary local-stochastic noise of strength~$p$ below a constant threshold~$p_0$, i.e., produces a constant-fidelity Bell pair.
\end{enumerate}
Then we must have 
\begin{align}
    m&=\Omega\left(\log R\right)\ .\label{eq:mlowerboundgeneral}
\end{align}
\end{theorem}
The proof of  Theorem~\ref{thm:nogotheoremmain} replaces the analysis of a low-latency protocol~$\pi$ by that of an entanglement-assisted  protocol~$\pi'$, and  follows from the fact that for any $p\in (0,1)$,  the qubit depolarizing channel $\cE_p(\rho)=(1-p)\rho+p \tr(\rho)\frac{I}{2}$ 
 is a convex combination of an entanglement-breaking channel~\cite{ebchannelshorodecki} and the identity channel (see Section~\ref{sec:converseboundlowlatency} for details). It implies that local stochastic noise with any constant noise strength prevents entanglement generation over distances greater than exponential in the number of qubits at each site. In particular, it shows that the protocol discussed in Theorem~\ref{thm:mainentanglement} is essentially optimal 
 among all low-latency protocols:  Comparing~\eqref{eq:mlowerboundgeneral} with~\eqref{eq:polylogscalingexplicitprotocol}, we conclude that
the polylogarithmic scaling of the number~$m$ of qubits per site in the protocol of Theorem~\ref{thm:mainentanglement} can at best be improved to a logarithmic scaling in~$R$. 

Closing the log-factor gap between  the converse Theorem~\ref{thm:nogotheoremmain} and achievability  Theorem~\ref{thm:mainentanglement} appears to require a new approach, however: We can show that the analysis of Theorem~\ref{thm:mainentanglement} leading to the scaling~\eqref{eq:polylogscalingexplicitprotocol} is tight. In fact, the scaling is necessary for any procedure which is based on the cluster state and uses the same syndrome information using  functions $\alpha$ and $\beta$ as in our protocol: We show the lower bound
\begin{align}
    m=\Omega(\log^2 R)
\end{align}
for any such protocol (see Theorem~\ref{thm:mainconversebound} for a detailed statement). This result is obtained by a detailed analysis of the associated decoding problem, see Section~\ref{sec:converseboundcluster}. It applies, in particular, to natural variants of our scheme that are obtained, e.g., by replacing Edmond's minimal matching algorithm by other procedures, including heuristic algorithms.
This indicates that improving  over our protocol will require a different resource state or measurement pattern.

\subsubsection*{Combining cluster-state generation and single-shot decoding}
Let us briefly discuss the main ideas underlying our work, giving a high-level overview focused on the two main building blocks:
\begin{enumerate}[(a)]
\item\label{it:clusterstategenerationprocedure}
  The generation of a cluster state on an elongated lattice~$\cC[d\times d\times R]$  and subsequent single-qubit measurement of ``bulk'' qubits located at sites not belonging to the two boundaries with third coordinate equal to~$x_3=1$ and $x_3=R$, respectively. Up to local unitaries, this part of the procedure is identical to
  the original  scheme of~\cite{raussendorfLongrangeQuantumEntanglement2005}. As shown in~\cite{bravyiQuantumAdvantageNoisy2020} (see Theorem~\ref{thm:singleshotbellstateprep} above), this generates a surface-code encoded Bell pair up to a (computable) Pauli correction and residual local stochastic noise.
\item\label{it:singleshotdecodingprocedure}
  A novel single-shot decoding procedure for transferring surface-code encoded quantum information onto a single physical qubit. We show that this procedure is resilient to local stochastic noise: There is  a threshold $p^{\textrm{dec}}_0$ such that if the initial encoded state~$\overline{\Psi}$ is corrupted by a local stochastic $E\sim \cN(p)$ with noise strength $p\leq p^{\textrm{dec}}_0$, then the fidelity of the resulting single-qubit state with the logical state~$\overline{\Psi}$ is exponentially close to~$1$ as a function of the code distance~$d$. We refer to Theorem~\ref{thm:singleshotdecoding} below for details.
    \end{enumerate}

From a conceptual point of view, it is clear that combining~\eqref{it:clusterstategenerationprocedure} with~\eqref{it:singleshotdecodingprocedure} results in a procedure for generating long-range entanglement: Theorem~\ref{thm:singleshotbellstateprep} guarantees that 
the state after the first part of the procedure is essentially an encoded Bell state  corrupted by a local stochastic error. Applying the single-shot decoding procedure to each of the two surface codes then results in a high-fidelity two-qubit Bell state. Indeed, our result (Theorem~\ref{thm:singleshotdecoding} below) ensures that the local stochastic error is dealt with properly in the single-shot decoding phase.

While our scheme for long-range entanglement generation is indeed a direct combination of steps~\eqref{it:clusterstategenerationprocedure} with~\eqref{it:singleshotdecodingprocedure}, our analysis of the resulting fidelity is not simply a combination of the analytical results stated above for each module~\eqref{it:clusterstategenerationprocedure} and~\eqref{it:singleshotdecodingprocedure}. This is because a naive combination of these analytical building blocks only provides a proof-of-principle: A threshold result derived in this way does not yield parameters of practical relevance. Instead, we establish a significantly tighter threshold estimate by a direct analysis of the achieved fidelity for the entire (combined) process.

\subsubsection*{Fault-tolerant single-shot decoding for surface codes}
\begin{figure}
  \centering
    \includegraphics[width=.6\textwidth]{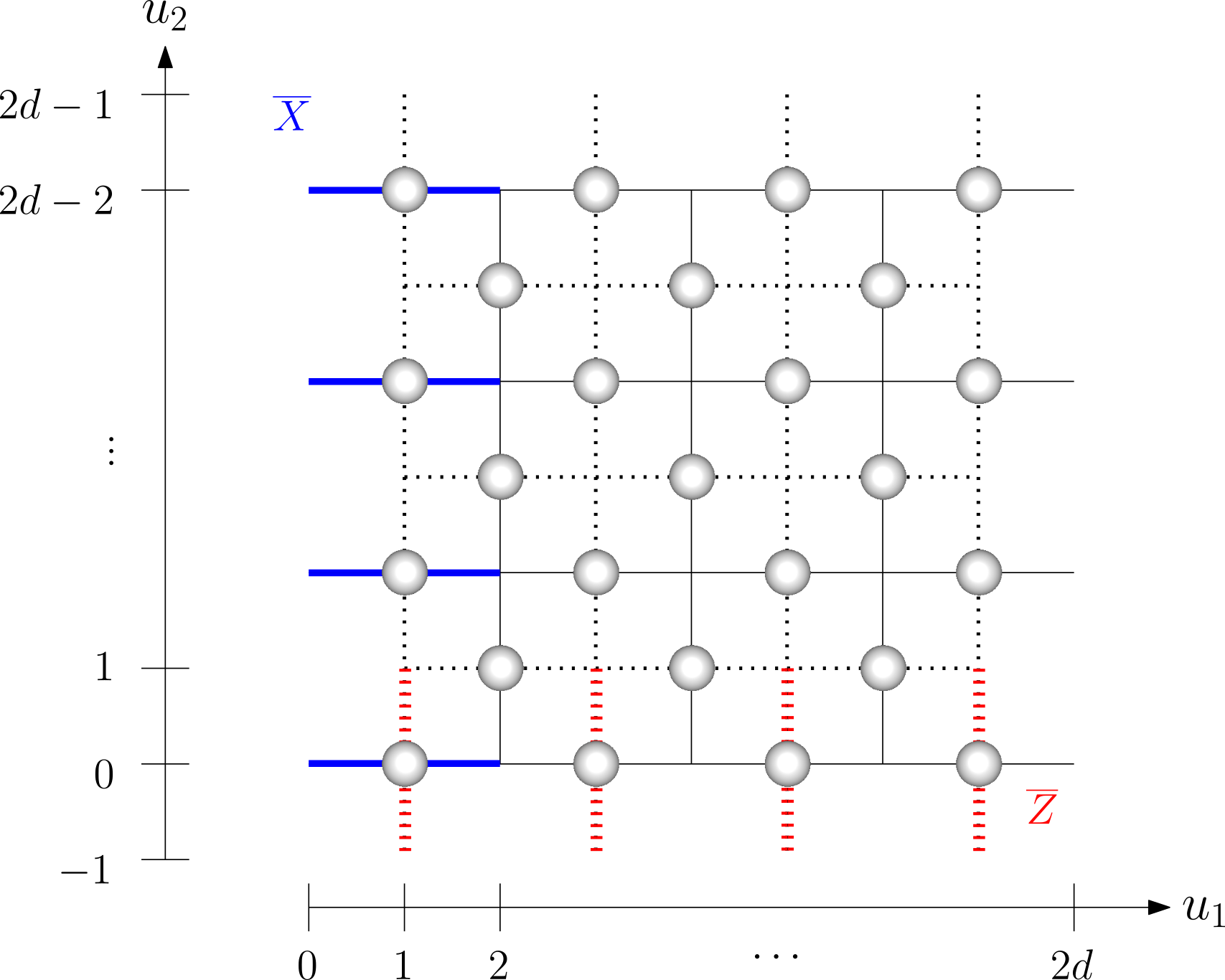}
    \caption{The surface code of distance $d$ with the lattice (solid) and its dual (dotted). Each sphere represents a qubit. Logical string operators $\overline{X}$ (blue, solid edges) and $\overline{Z}$ (red, dotted edges) are illustrated.}
    \label{fig:surfacecodelogical}
\end{figure}

In the following, 
we consider a single qubit encoded into a surface code using an isometric encoding map $V:\mathbb{C}^2\rightarrow(\mathbb{C}^{2})^{\otimes n}$ that maps the single-qubit Pauli operators to the ``canonical'' logical string operators 
$\overline{X}=VXV^\dagger$ and ~$\overline{Z}=VZV^\dagger$ of the surface code, see Fig.~\ref{fig:surfacecodelogical}.  Let us write $\overline{\Psi}:=V\Psi$ for the encoded version of a physical qubit~$\Psi\in\mathbb{C}^2$. Our main  building block is a decoding procedure for the surface code that can be summarized as follows. It is illustrated in Fig.~\ref{fig:surfacecodemeasurementpattern}. 

\begin{figure}[h]
    \centering
    \includegraphics[width=0.7\textwidth]{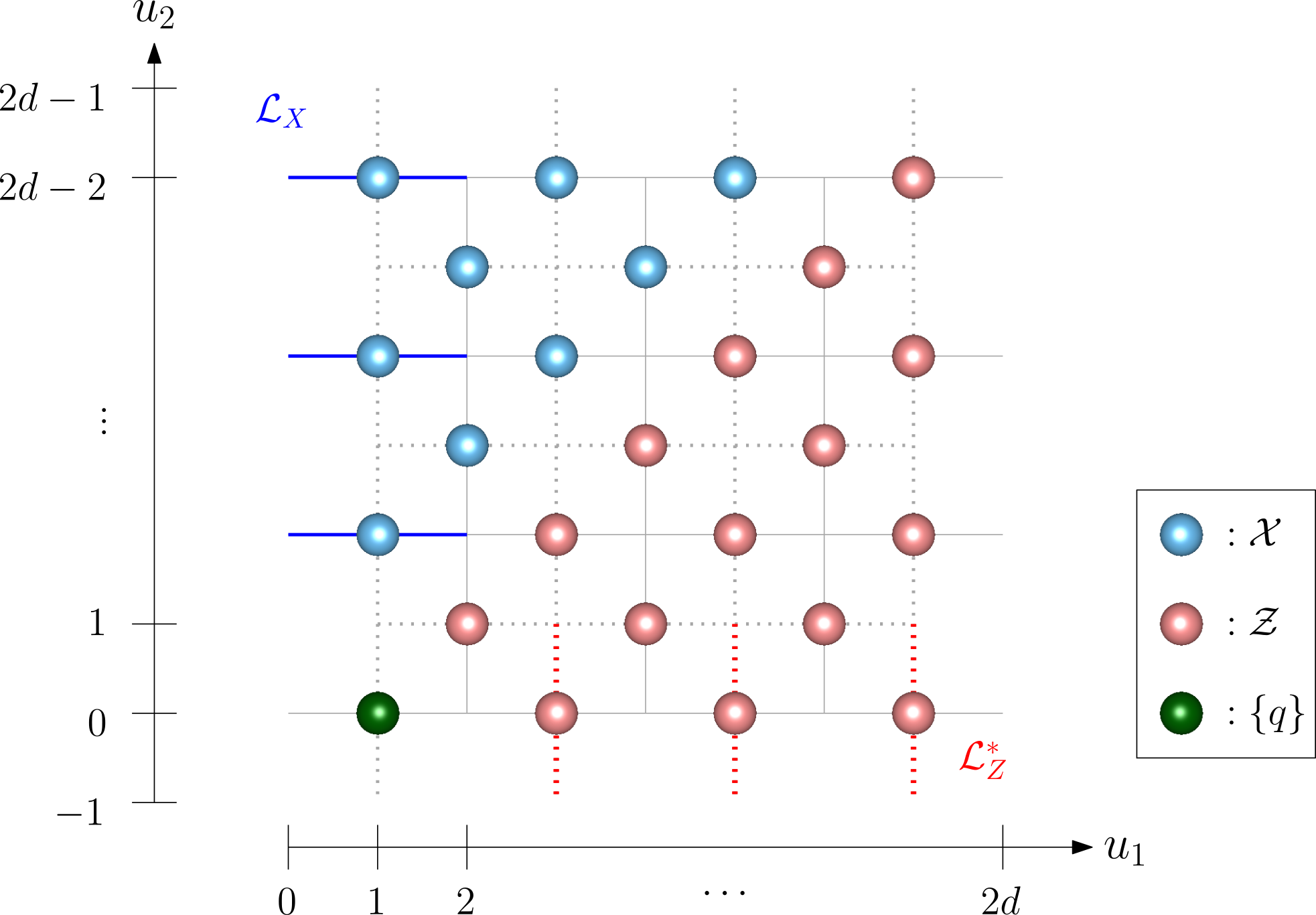}
    \caption{The measurement pattern for single-shot surface code decoding. The recovery sets $\cL_X \subset \cX$ (blue, solid edges) and $\cL_Z^{*} \subset \cZ$ (red, dotted edges) for fault-tolerant measurement of $\overline{X}$ and $\overline{Z}$ are illustrated.}
    \label{fig:surfacecodemeasurementpattern}
\end{figure}

\begin{theorem}[Single-shot surface code decoding]\label{thm:singleshotdecoding} Consider a distance-$d$ surface code with~$d\geq 2$ consisting of~$n=2d^2-2d+1$ qubits. 
Then there is a partition~$[n]=\cX\cup\cZ\cup \{q\}$
of the set of qubits into disjoint sets~$\cX$, $\cZ$ and a set~$\{q\}$ consisting of a single qubit~$q$, and two  efficiently computable functions 
\begin{align}
\begin{matrix}
\alpha: & \{0,1\}^{\cC\backslash \{q\}} & \rightarrow &\{0,1\}\\
\beta: & \{0,1\}^{\cC\backslash \{q\}} & \rightarrow &\{0,1\}
\end{matrix}
\end{align}
such that the following holds. Consider the following protocol applied to a surface-code encoded logical state~$\overline{\Psi}:=V\Psi$ corrupted by 
a local stochastic error~$E\sim\cN(p)$ of strengh
\begin{align}
p \leq \frac{1}{144}\approx 7\times 10^{-3} \ ,
\end{align}
where  $\Psi\in\mathbb{C}^2$ is arbitrary:
\begin{enumerate}[(a)]
\item
Each qubit~$c\in\cX$ is measured in the Hadamard basis, and
\item
each qubit~$c\in \cZ$ is measured in the computational basis.
\end{enumerate}
Let $s\in \{0,1\}^{\cC\backslash \{q\}}$ be the (collection of) measurement outcomes, and let 
\begin{align}
    \psi_E(s)\propto (I_{q}\otimes \bra{s}H(\cX))EV\ket{\Psi}
\end{align}
be the corresponding  post-measurement state on qubit~$q$. Then
\begin{align}
    \Pr_{E,s}\left[Z_q^{\alpha(s)}X_q^{\beta(s)}\psi_E(s)\propto \Psi\right]\geq 1-94p \ .
\end{align}
\end{theorem}
This shows that the logical state encoded in the surface code can be transferred to a single qubit in a fault-tolerant manner, up to a Pauli correction which is determined by the measurement outcomes.  We refer to Theorem~\ref{thm:mainsingleshot} below for a more detailed statement.  Similarly as before,  the functions~$\alpha,\beta$ are defined in terms of minimum matchings, and can be computed in time~$O(n^3)$ in this case.

 We note that  Theorem~\ref{thm:singleshotdecoding} immediately yields a protocol for fault-tolerantly measuring the logical Pauli-$\overline{Z}$ operator, thus subsuming Theorem~\ref{thm:singleshotlogicalZsurface}. The corresponding measurement protocol  proceeds by  applying a computational basis measurement on the post-measurement state~$\psi_E(s)$  on qubit~$q$ with outcome~$m$, and outputting the corrected bit~$m\oplus \beta(s)$.  In a similar manner, one can realize a fault-tolerant logical $T$-measurement as considered in~\cite{Litinski2019magicstate}.  Beyond measuring~$\overline{Z}$ or~$\overline{T}$, Theorem~\ref{thm:singleshotdecoding} also permits to execute any logical measurement on the encoded information. More generally, as the information ends up being localized on qubit~$q$, it can be acted on in subsequent quantum-information-processing protocols.

The protocol provided in Theorem~\ref{thm:singleshotdecoding} bears some similarity to teleportation: encoded information is essentially teleported to a single physical qubit. To illustrate the basic 
underlying idea, consider a CSS-stabilizer code encoding a single qubit, with a property that logical $\overline{X}$- and $\overline{Z}$-Pauli operators can be chosen in a way that they overlap only in a single qubit~$q$. As a concrete example, consider
the stabilizer code with generators
\begin{align}
  S_1 &=X_1X_3\\
  S_2 &=Z_1Z_2Z_3\ .
\end{align}
It encodes  a single logical qubit into three physical qubits. We can choose logical operators as
\begin{align}
  \overline{X}&=X_1X_2\\
  \overline{Z}&=Z_1Z_3\ ,
  \end{align}
  which has the desired property (with $q$ being the first qubit). Suppose~$\overline{\Psi}\in(\mathbb{C}^2)^{\otimes 3}$ is an encoded state associated with~$\Psi\in\mathbb{C}^2$. Assume we measure qubit~$2$ in the Hadamard basis, obtaining outcome~$x\in\{0,1\}$, and measure qubit~$3$ in the computational basis getting the outcome~$z\in \{0,1\}$. This corresponds to measuring the observables~$X_2$ and~$Z_3$. Then it is easy to check that the post-measurement state~$\varphi(x,z)\in\mathbb{C}^2$ on qubit~$1$  has expectation values
  \begin{align}
  \bra{\varphi(x,z)} X\ket{\varphi(x,z)}&=(-1)^x \bra{\overline{\Psi}}\overline{X}\ket{\overline{\Psi}}\\
  \bra{\varphi(x,z)} Z\ket{\varphi(x,z)}&=(-1)^z \bra{\overline{\Psi}}\overline{Z}\ket{\overline{\Psi}}\ .
    \end{align}
 It follows immediately that the state~$Z^xX^z\ket{\varphi(x,z)}$ has identical expectation values as the logical state~$\overline{\Psi}$, and hence must be proportional to~$\Psi$. That is, the logical state~$\overline{\Psi}$ is decoded, i.e., transferred to the first physical qubit. We note, however, that the process described here is not fault-tolerant. In particular, the measurement of the observables~$X_2$ and $Z_3$ is not resilient to errors and may lead to an incorrect ``correction'' operation~$Z^xX^z$.

 Here we argue that for the surface code, an analogous process can in fact be made fault-tolerant. The idea is to use logical operators that again only overlap in a single qubit~$q$, i.e.,
 \begin{align}
 \overline{X}&=X_q X(\cL_X)\\
 \overline{Z}&=Z_q Z(\cL_Z^{*})\ ,
 \end{align}
 where~$\cL_X$ and $\cL_Z^{*}$ are disjoint subsets of qubits, see Fig.~\ref{fig:surfacecodemeasurementpattern}. The goal, then, is to fault-tolerantly measure the observables~$X(\cL_X)$ and $Z(\cL_Z^{*})$. Such a procedure can be obtained by modifying the protocol of Theorem~\ref{thm:singleshotlogicalZsurface}. The latter gives a fault-tolerant way of measuring~$\overline{Z}$ by measuring  all qubits in the computational basis, and, by considering the dual lattice, a way of measuring the logical operator~$\overline{X}$ by measuring all qubits in the  Hadamard basis. 
 To simulate the measurement of both~$X(\cL_X)$ and $Z(\cL_Z^{*})$ at the same time, we use a  pattern of single-qubit measurements that simultaneously measures a subset of qubits in the computational basis, and a subset of qubits in the Hadamard basis. The measurement pattern is ``triangular'' as shown in Fig.~\ref{fig:surfacecodemeasurementpattern}. We show that this procedure has the properties stated in Theorem~\ref{thm:singleshotdecoding}.

\subsubsection*{Outline}
In Section~\ref{sec:singleshotdecodingprotocol}, we introduce our single-shot decoding protocol for the surface code accompanied with background material. In Section~\ref{sec:boundmatchingmain}, we provide a combinatorial framework for analyzing  matching problems
 arising in the analysis of both  our surface code decoding protocol and the  entanglement generation protocol introduced in Section~\ref{sec:entanglementgenerationnoisy}. Using the framework, we prove in Section~\ref{sec:singleshotdecode} that our surface code decoding protocol is robust against local stochastic noise.

Then in Section~\ref{sec:entanglementgenerationnoisy}, we introduce our long-range entanglement generation protocol.  We prove in Section~\ref{sec:resiliencenoiselocalstochastic} that this protocol is resilient against local stochastic noise of strength below a certain threshold: A lower bound on the achievable distance of entanglement generation is presented in terms of the noise strength and the number of qubits. Finally, in Section~\ref{sec:converse}, we present a converse result for low latency entanglement generation schemes.

\section{Single-shot surface code decoding\label{sec:singleshotdecodingprotocol}}
In this section, we define our single-shot decoding protocol for the surface code: We introduce all relevant definitions to apply this protocol in practice.  We defer the analysis of the resilience of this protocol against local stochastic noise to subsequent sections: In Section~\ref{sec:boundmatchingmain}, we will first provide general bounds on certain matching problems on graphs. These will subsequently be applied in Section~\ref{sec:singleshotdecode} to establish a fault-tolerance threshold theorem for the single-shot decoding protocol introduced here.

\subsection{Definition of surface code\label{sec:surfacecodelattice} }
We consider a distance-$d$-surface code with
a total of $n=2d^2-2d+1$ qubits associated with $d^2$~horizontal edges and $(d-1)^2$~vertical edges as illustrated in Fig.~\ref{fig:surfacecodegraph}.  The code has smooth boundaries at the top and bottom, and rough boundaries on the left and right. In more detail, consider a square lattice in~$\mathbb{R}^2$ consisting of {\em sites}
\begin{align}
\Cprime&=\left\{(u_1,u_2)\in\mathbb{Z}^2\ |\ 0\leq u_1\leq 2d\textrm{ and } 0\leq u_2\leq 2d-2\right\}\ .
\end{align}
We  place qubits at sites~$(u_1,u_2)\in\Cprime$ with the property that one coordinate~$u_j$ is even, whereas the other coordinate~$u_k$ is odd, i.e., qubits are located at the sites
\begin{align}
\cC=\Cprime \setminus\{(o,o),(e,e)\}=\{(o,e)\}\cup \{(e,o)\}\subset\Cprime .
\end{align}
Here and below we write $e$ (respectively~$o$) for an even (respectively odd) integer and use the convention that e.g., $\{(o,o)\}\subset\Cprime$ is the subset of all pairs $(u_1,u_2)\in \Cprime$ with both $u_1$ and $u_2$ even.  Let
\begin{align}
d(u,v):=\sum_{j=1}^2 \abs{u_j-v_j}\ 
\end{align}
denote Manhattan distance between~$u$ and $v$. Given a site~$u\in \Cprime$, the set of nearest  neighbors of~$u$ is defined as 
\begin{align}
\neigh(u) = \left\{v \in \mathcal{C} \mid d(u,v) = 1\right\}\ 
\label{eq:nearestneighborset}
\end{align}
i.e., the  nearest neighbors of~$u$ are associated with qubits at Manhattan distance~$1$ from~$u$.

The  surface code graph~$T_{sc}=(\Vsc,\Esc)$ is  defined as
\begin{align}
\Vsc&=\{(e,e)\}\subset \Cprime\\
\Esc&=\left\{\{u,v\}\in \Vsc\times \Vsc\ |\ (u_1,v_1)\not\in\{(0,0),(2d,2d)\}, d(u,v)=2\right\}\ .
\end{align}
In other words, edges of~$\Tsc$ are associated with pairs~$(u,v)$ separated by a distance~$d(u,v)=2$, with the exception of pairs~$(u,v)$ that lie on the same vertical at $u_1=v_1=0$ or $u_1=v_1=2d$. The corresponding vertices
\begin{align}
  \Vdangsc &=\{(0,e)\}\cup \{(2d,e)\} \subset \Vsc
\end{align}
are degree-$1$ vertices in~$\Tsc$, i.e., endpoints of ``dangling'' edges at one of the ``rough'' boundaries of~$\Tsc$: they are connected only to their right respectively left nearest neighbor. We can therefore express the set of edges of~$\Tsc$ succintly as
\begin{align}
\Esc&=\left\{\{u,v\}\in \Vsc\times \Vsc\ |\ d(u,v)=2\textrm{ and }  (u,v)\not\in\Vdangsc\times\Vdangsc\right\}\ .
  \end{align}
In the following, it will be convenient
to label an edge~$\{u,v\}\in \Esc$ by its midpoint~$(u+v)/2\in \cC$.

We will extensively use a ``dual graph'' $\Tscdual=(\Vscdual,\Escdual)$. To define this, let us first consider the set
\begin{align}
  \Cprimedual&=\left\{(u_1,u_2)\in\mathbb{Z}^2\ |\ 1\leq u_1\leq 2d-1\textrm{ and } -1\leq u_2\leq 2d-1\right\}\
\end{align}
of dual sites. We note that the location of qubits can equivalently be specified as
\begin{align}
\cC^{*}&=\{(o,e)\}\cup \{(e,o)\}\subset\Cprime^{*}\ .
\end{align}
i.e., $\cC = \cC^{*}$.
The dual surface code graph~$\Tscdual=(\Vscdual,\Escdual)$
is defined by 
\begin{align}
\Vscdual&=\{(o,o)\}\subset \Cprime_*\\
\Escdual&=\left\{\{u,v\}\in \Vscdual\times \Vscdual\ |\ (u_2,v_2)\not\in\{(-1,-1),(2d-1,2d-1)\} \textrm{ and }d(u,v)=2\right\}\ .
\end{align}
Here degree-$1$-vertices, i.e., those belonging to
\begin{align}
  \Vdangscdual &=\{(o,-1)\}\cup \{(o,2d-1)\} \subset \Vscdual\ ,
\end{align}
are located at the top- and bottom ``rough'' boundaries.
Again, we sometimes refer to an edge~$\{u,v\}\in \Escdual$ by its midpoint~$(u+v)/2$.

The surface code associated with the surface code graphs~$\Tsc,\Tscdual$ is defined as follow: For every vertex~$u\in \Vsc$ with degree~$\deg_{\Tsc}(u)\geq 3$, i.e.,
every vertex~$u\in \Vsc\setminus\Vdangsc$ which is not the endpoint of a dangling edge~$u$, there is a stabilizer generator~$A_u:=\prod_{v\in \neigh(u)} X_v$ consisting of Pauli-$X$ operators only. Pauli-$Z$-stabilizer generators are defined similarly based on the dual graph~$\Tscdual$: For every $u\in \Vsc^*$, there is a stabilizer generator~$B_u:=\prod_{v\in\neigh(u)}Z_v$ consisting of Pauli-$Z$ operators only. 

\begin{figure}
\centering 
\begin{subfigure}[b]{0.45\textwidth}
   \centering 
   \includegraphics[width=.9\linewidth]{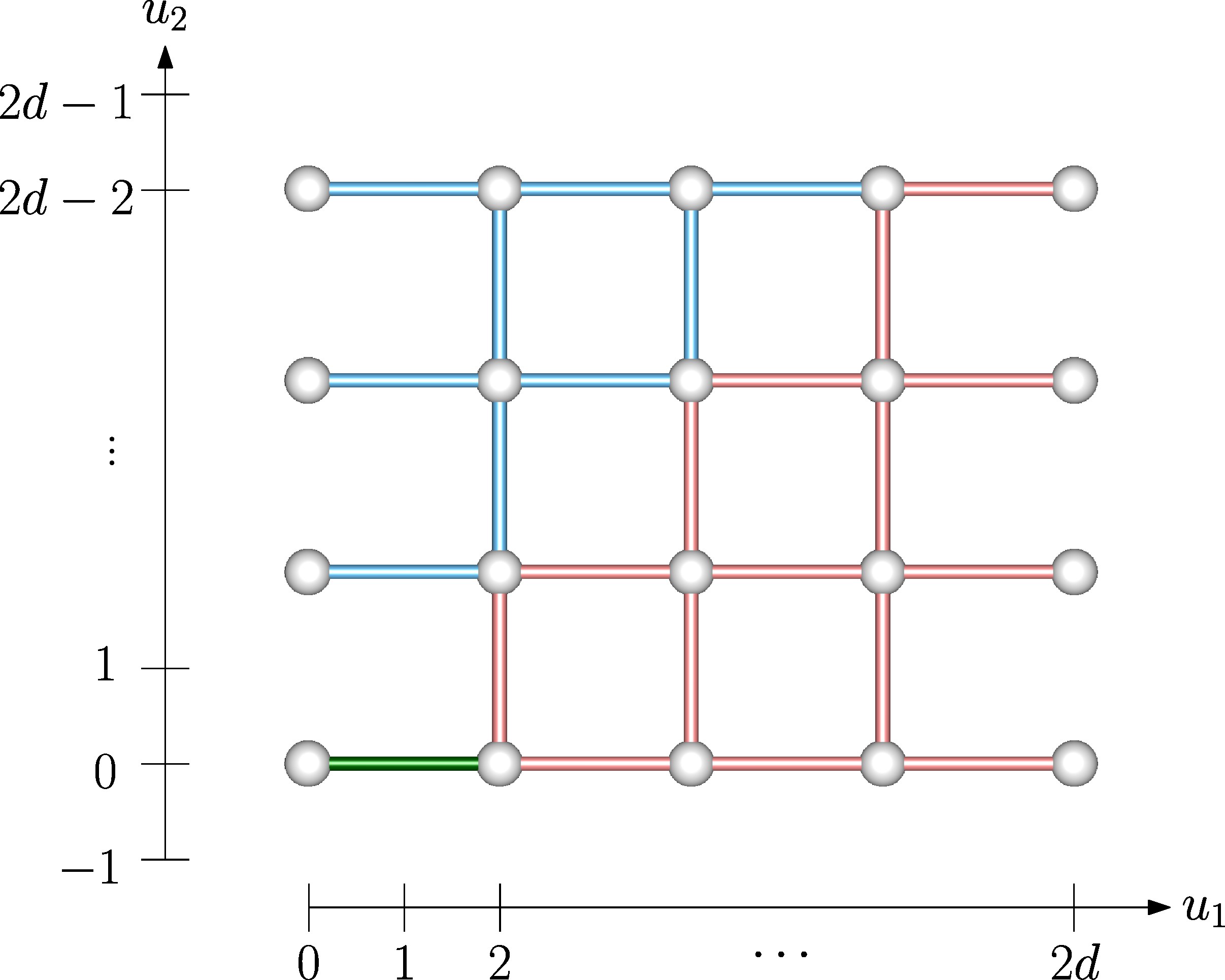}
    \caption{
    \begin{tabular}{ |c||c|c| } 
     \hline
     \multicolumn{3}{|c|}{ $\Tsc=(\Vsc, \Esc)$} \\
     \hline\hline
     subsets of $\Esc$ & color & measurement\\ 
     \hline\hline
     $\cX$ &\textcolor{lightblue}{ \rule{1.5cm}{0.3cm}} & $X$-basis \\
     \hline 
     $\cZ$ &\textcolor{lightred}{ \rule{1.5cm}{0.3cm}} & $Z$-basis \\
     \hline 
     $\{q\}$ & \textcolor{mygreen}{ \rule{1.5cm}{0.3cm}} & - \\
     \hline 
    \end{tabular}
    }\label{fig:Tscmeasurement}
\end{subfigure}
\hfill  \quad 
\begin{subfigure}[b]{0.45\textwidth}
   \centering 
   \includegraphics[width=.9\linewidth]{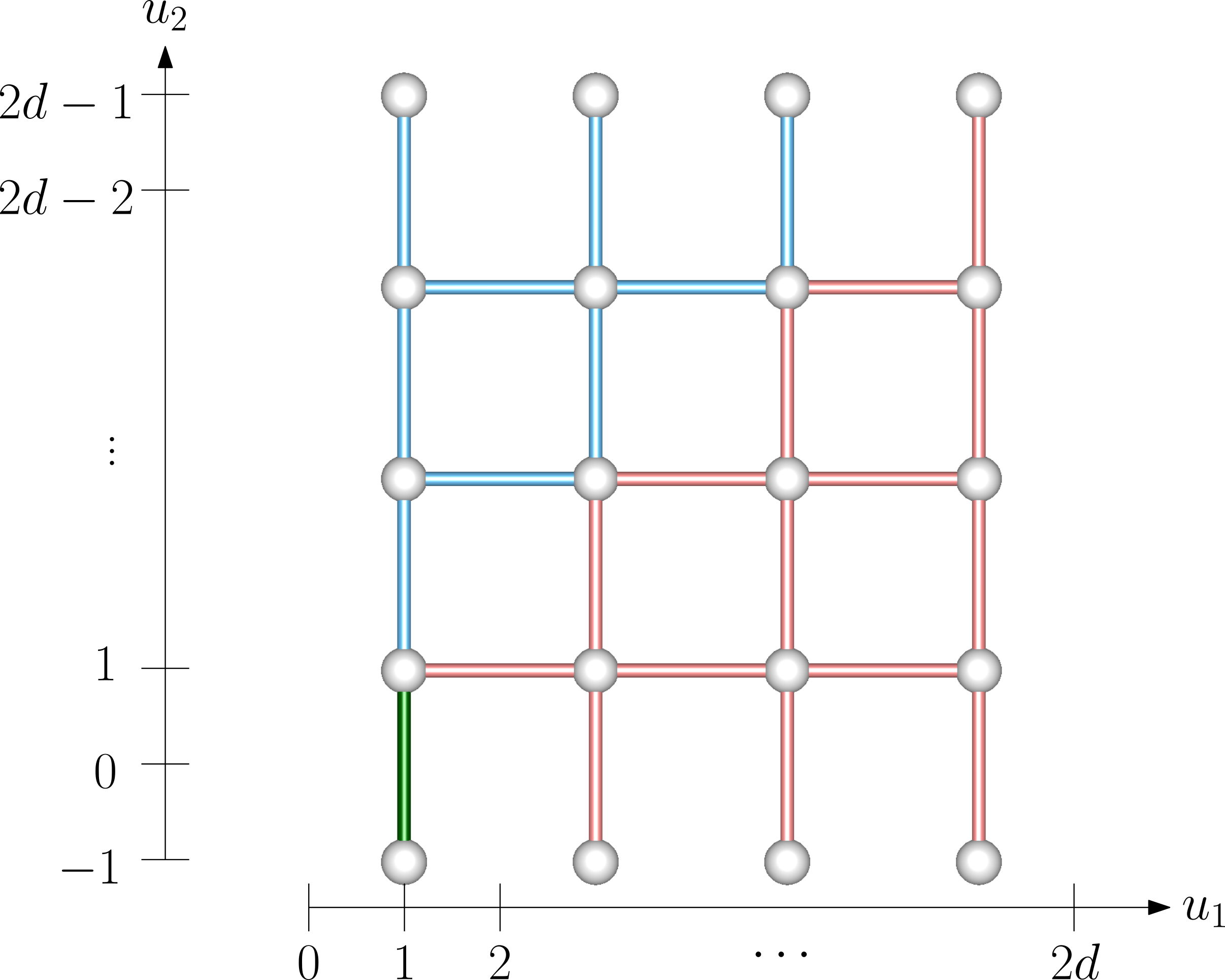}
    \caption{
    \begin{tabular}{ |c||c|c| } 
     \hline
     \multicolumn{3}{|c|}{ $\Tscdual=(\Vscdual, \Escdual)$} \\
     \hline\hline
     subsets of $\Escdual$ & color & measurement\\ 
     \hline\hline
     $\cX$ &\textcolor{lightblue}{ \rule{1.5cm}{0.3cm}} & $X$-basis \\
     \hline 
     $\cZ$ &\textcolor{lightred}{ \rule{1.5cm}{0.3cm}} & $Z$-basis \\
     \hline 
     $\{q\}$ & \textcolor{mygreen}{ \rule{1.5cm}{0.3cm}} & - \\
     \hline 
    \end{tabular}
    }\label{fig:Tscdualmeasurement}
\end{subfigure}
\caption{
The surface code graph~$\Tsc$, and the dual graph~$\Tscdual$. All qubits are situated at the midpoints of the edges.
The protocol uses a measurement pattern where qubits belonging to~$\cX$ (midpoints of blue edges) are measured in the Pauli-$X$ basis, qubits belonging to~$\cZ$ (midpoints of red edges) are measured in the Pauli-$Z$ basis, and qubit~$q$ (midpoints of green edges) is not measured.
\label{fig:surfacecodegraph}}
\end{figure}

\subsection{Graphs with internal vertices: boundary map and minimal matchings}\label{sec:graphswithinternalvertices}
Our decoder is based on an algorithm~$\mmatch$ which finds certain matchings of a graph with internal and external vertices.  Let us first discuss some basic notions.

We  consider  pairs $(G=(V,E),V^\mathrm{int})$ where~$G$ is a graph and~$V^\mathrm{int}\subset V$ a subset of vertices we call internal vertices. The complement~$V^\mathrm{ext}:=V\setminus V^\mathrm{int}$ will be referred to as the set of external vertices.  

Define a boundary map
$\partial:2^{E} \rightarrow 2^{V^\mathrm{int}}$ that takes a subset~$\cE\subset E$ of edges 
to a subset~$\partial \cE\subset V^\mathrm{int}$ of internal vertices: An interval vertex~$v\in V^\mathrm{int}$ belongs to~$\partial\cE$ if and  only if the number of edges~$e\in\cE$ incident on~$v$ is odd. 
 The map~$\partial$ can be understood as the ``usual'' $\mathbb{Z}_2$-boundary map taking $1$-chains  (sets of edges) on the graph~$G$ to $0$-chains (sets of vertices)  on~$G$, followed by a projection onto the subset of vertices~$V^\mathrm{int}$. We observe that, interpreted as a function $\partial:\mathbb{Z}_2^{E}\rightarrow\mathbb{Z}_2^{V^\mathrm{int}}$, the boundary map~$\partial$ is $\mathbb{Z}_2$-linear. 
 
Let $\cV\subset V^\mathrm{int}$ be a subset of internal vertices. A subset~$\cE\subset E$ of edges is  called a matching of~$\cV$ if~$\partial\cE=\cV$. 
A matching~$\cE$ of~$\cV$ is called minimal if it its size~$|\cE|$ is minimal among all matchings of~$\cV$.

We note that for every graph $G=(V,E)$ with interval vertices $\Vint$, there is a deterministic algorithm~$\mmatch$ which outputs a minimum matching~$\mmatch(\cV)$ of the given subset~$\cV \subset \Vint$ with runtime  bounded by $O(\abs{V}^3)$~\cite{lawler2001,gabow1976,wangThresholdErrorRates2009}. This is obtained by suitably adapting Edmond's matching algorithm \cite{Edmonds1965}.

In Section~\ref{sec:boundmatchingmain}, we will provide a general combinatorial analysis of matching on graphs with internal vertices. The ``decoding'' graphs of interest will be introduced in Section~\ref{sec:decodingsubgraphssingle}.

\subsection{Measurement pattern used in protocol\label{sec:measurementpatternsingle}}
Here we define the triangular measurement pattern used in our single-shot decoding protocol.  Let 
\begin{align}
q&:=(1,0)\in \cC 
\end{align}
be the qubit at the bottom left corner. The measurement pattern is then defined as follows (see Fig.~\ref{fig:surfacecodegraph}):
Every qubit belonging to
\begin{align}
\cX&:=\left\{(u_1,u_2)\in \cC\setminus\{q\}\ |\ u_2\geq u_1+1\right\}
\end{align}
is measured in the Pauli-$X$-basis. Every qubit belonging to
\begin{align}
\cZ&:=\left\{(u_1,u_2)\in \cC\setminus\{q\}\ |\ u_2< u_1+1\right\}
\end{align}
is measured in the Pauli-$Z$ basis. The qubit at~$q$ is not measured.

\subsection{Recovery sets\label{sec:recoverysetsurface}}

We will consider a special subset~$\cL_X\subset \Esc$ of edges, see Fig.~\ref{fig:Tdecmarked}. The set $\cL_X$ is located on a vertical strip and consists of edges associated with the qubits
\begin{align}
\cL_X&=\{ (1,u_2)\ |\ u_2\geq 2, u_2\textrm{ even }\}\ .\label{eq:cLXdefinition}
\end{align}
(Here we use the convention that an edge $\{u,v\}$ is represented by its midpoint; it is, however, more natural to think of elements of~$\cL_X$ as edges of~$\Esc$.)

Similarly, we consider a ``dual'' set~$\cL_Z^*\subset\Edecdual$ located on a horizontal strip defined by their midpoints
\begin{align}
  \cL_Z^*&=\left\{(u_1,0)\ |\ u_2\geq 3, u_2\textrm{ odd }\right\}\ .\label{eq:cLZdefinition}
\end{align}
We will argue below that our protocol will need to estimate parities of certain strings restricted to~$\cL_X$ and $\cL_Z^*$, respectively, in order to fix errors. This is because the associated qubits are closely related to the logical Pauli operators of the surface code, as discussed below. For this reason, we call $\cL_X$ and $\cL_Z^*$  the recovery sets.

\begin{figure}
\centering 
\begin{subfigure}[b]{0.45\textwidth}
   \centering 
   \includegraphics[width=.9\linewidth]{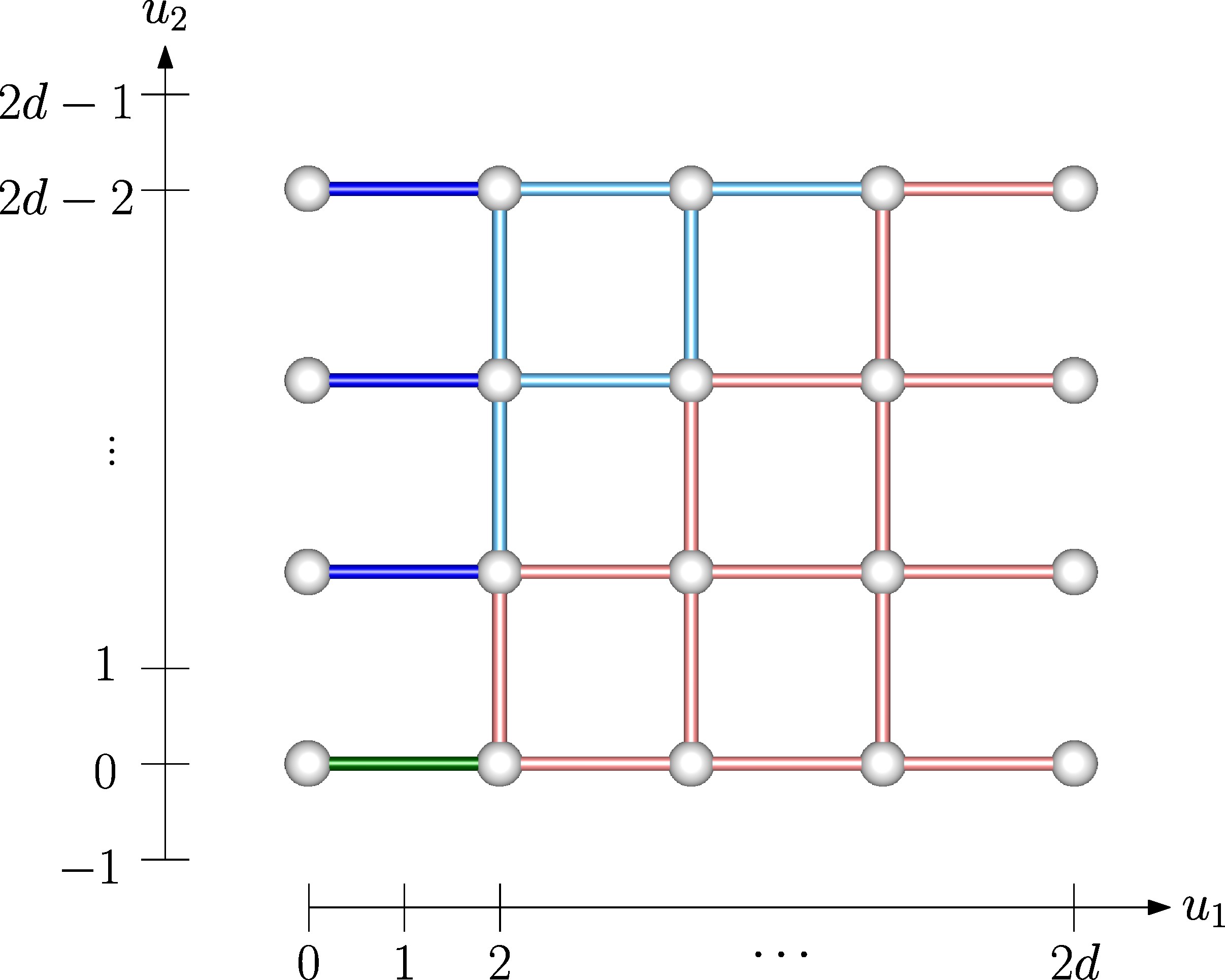}
    \caption{
    \begin{tabular}{ |c||c|c| } 
     \hline
     \multicolumn{3}{|c|}{ $\Tsc=(\Vsc, \Esc)$} \\
     \hline\hline
     subsets of $\Esc$ & color & measurement\\ 
     \hline\hline
     $\cL_X$ &\textcolor{blue}{ \rule{1.5cm}{0.3cm}} & $X$-basis \\
     \hline 
     $\cX \setminus \cL_X$ &\textcolor{lightblue}{ \rule{1.5cm}{0.3cm}} & $X$-basis \\
     \hline 
     $\cZ$ &\textcolor{lightred}{ \rule{1.5cm}{0.3cm}} & $Z$-basis \\
     \hline 
     $\{q\}$ & \textcolor{mygreen}{ \rule{1.5cm}{0.3cm}} & - \\
     \hline 
    \end{tabular}
    }\label{fig:Tscrecoveryset}
\end{subfigure}
\hfill  \quad 
\begin{subfigure}[b]{0.45\textwidth}
   \centering 
   \includegraphics[width=.9\linewidth]{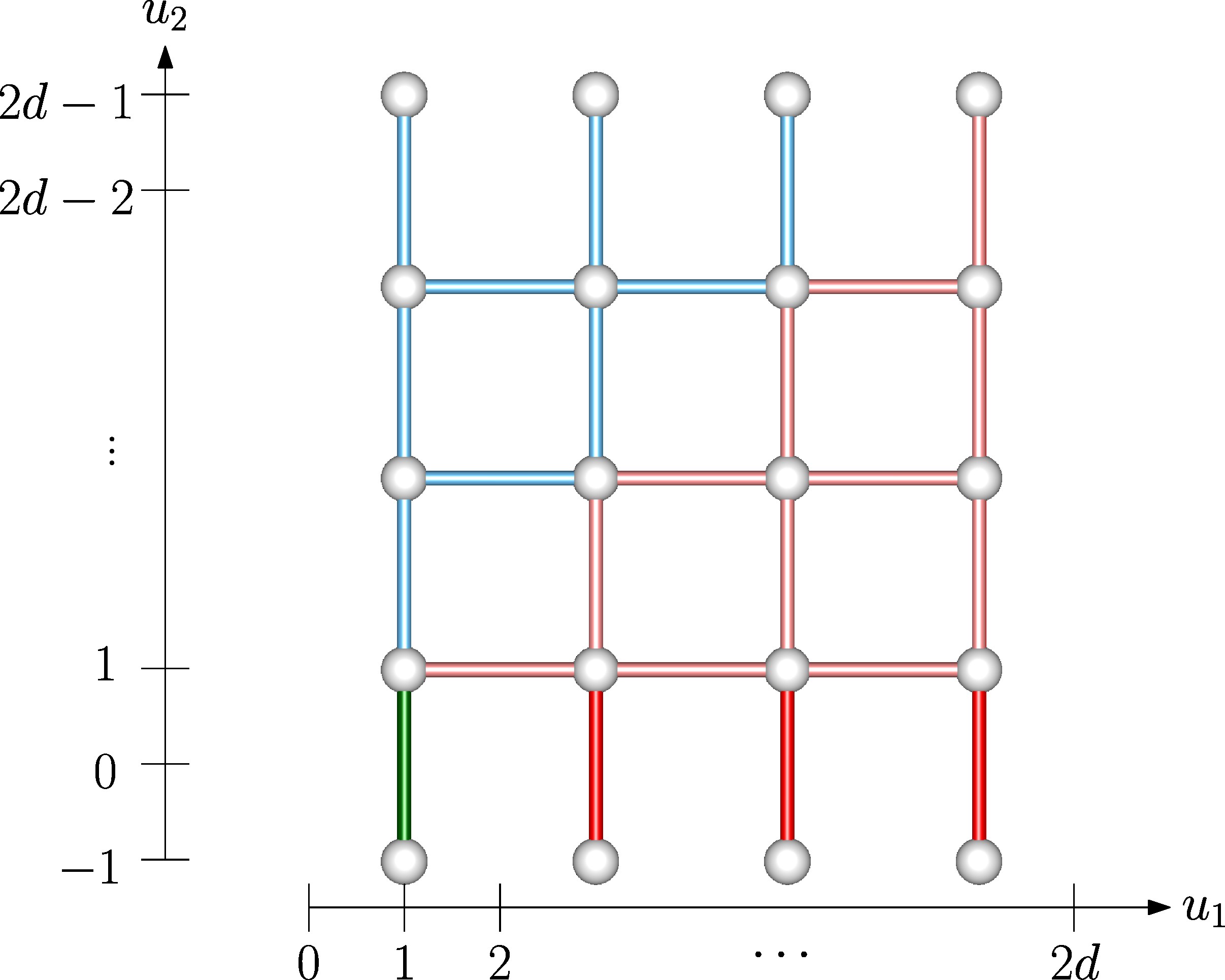}
    \caption{
    \begin{tabular}{ |c||c|c| } 
     \hline
     \multicolumn{3}{|c|}{ $\Tscdual = (\Vscdual, \Escdual)$} \\
     \hline\hline
     subsets of $\Edecdual$ & color & measurement\\ 
     \hline\hline
     $\cX$ &\textcolor{lightblue}{ \rule{1.5cm}{0.3cm}} & $X$-basis \\
     \hline
     $\cL_Z^{*}$ &\textcolor{red}{ \rule{1.5cm}{0.3cm}} & $Z$-basis \\
     \hline 
     $\cZ \setminus \cL_Z^{*}$ &\textcolor{lightred}{ \rule{1.5cm}{0.3cm}} & $Z$-basis \\
     \hline 
     $\{q\}$ & \textcolor{mygreen}{ \rule{1.5cm}{0.3cm}} & - \\
     \hline
    \end{tabular}
    }\label{fig:Tscdualrecoveryset}
\end{subfigure}
\caption{The surface code graph~$\Tsc$ and the dual graph~$\Tscdual$, as well as the associated recovery sets~$\cL_X\subset \Esc$ (blue)  and $\cL_Z^*\subset \Escdual$ (red).  \label{fig:Tdecmarked}}
\end{figure}

\subsection{Decoding subgraphs and internal vertices\label{sec:decodingsubgraphssingle}}
Let  us define the following ``decoding'' subgraph~$\Tdec=(\Vdec,\Edec)$ of~$\Tsc$, see Fig.~\ref{fig:Tdec}. The graph~$\Tdec$ is the subgraph induced by the set of vertices
\begin{align}
\Vdec&= \{(u_1,u_2)\in \Vsc\ |\ u_2\geq u_1>0\}\ .
\end{align}
Alternatively, we may say that~$\Tdec$ is the subgraph of~$\Tsc$ induced by
the set of edges~$\Edec:=\cX$. Explicitly, the edges of~$\Tdec$ are 
\begin{align}
  \Edec&=\left\{ \{u,v\}\in \Vdec\times \Vdec\ |\ d(u,v)=2\right\}\ .
  \end{align}
Observe that the recovery set~$\cL_X\subset\Edec$ is a subset of the edges of~$\Tdec$. We will use~$\Tdec$ to estimate a parity of a certain string restricted to~$\cL_X$.

The graph~$\Tdec$ has distinguished internal vertices
\begin{align}
\Vdecint&=\left\{u\in \Vdec\ |\ \neigh(u)\subset \Edec\right\}\  .\label{eq:tdecinterior} 
\end{align}
A vertex~$u\in \Vdec\setminus\Vdecint$ which is not internal is called external. We denote the set of external vertices of~$\Tdec$ by~$\Vdecext$, see Fig.~\ref{fig:Tdec}. This definition is motivated by the fact that for every internal vertex~$u\in\Vdecint$, the ``star'' operator~$A_u$ of the surface code has support~$\supp(A_u)\subset \Edec$ completely contained in the set of edges of~$\Tdec$.  This is not the case for the external vertices.

\begin{figure}
\centering 
\begin{subfigure}[b]{0.45\textwidth}
   \centering 
   \includegraphics[width=.9\linewidth]{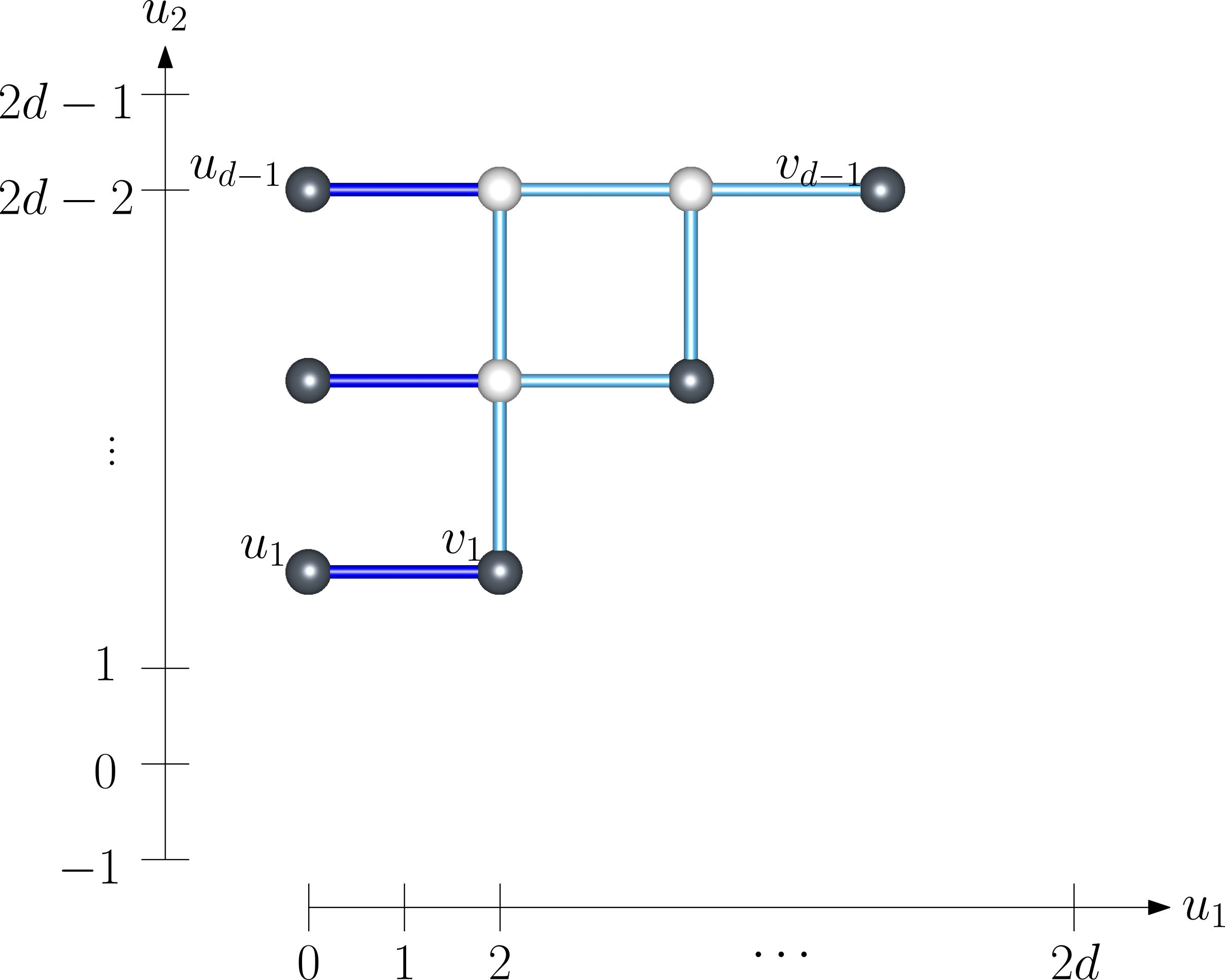}
    \caption{
    \begin{tabular}{ |c||c|c| } 
     \hline
     \multicolumn{3}{|c|}{ $\Tdec=(\Vdec, \Edec)$ } \\
     \hline\hline
     subsets of $\Edec$ & color & measurement\\ 
     \hline\hline
     $\cL_X$ &\textcolor{blue}{ \rule{1.5cm}{0.3cm}} & $X$-basis \\
     \hline 
     $\cX \setminus \cL_X$ &\textcolor{lightblue}{ \rule{1.5cm}{0.3cm}} & $X$-basis \\
     \hline 
    \end{tabular}
    }\label{fig:Tdec}
\end{subfigure}
\hfill  \quad 
\begin{subfigure}[b]{0.45\textwidth}
   \centering 
   \includegraphics[width=.9\linewidth]{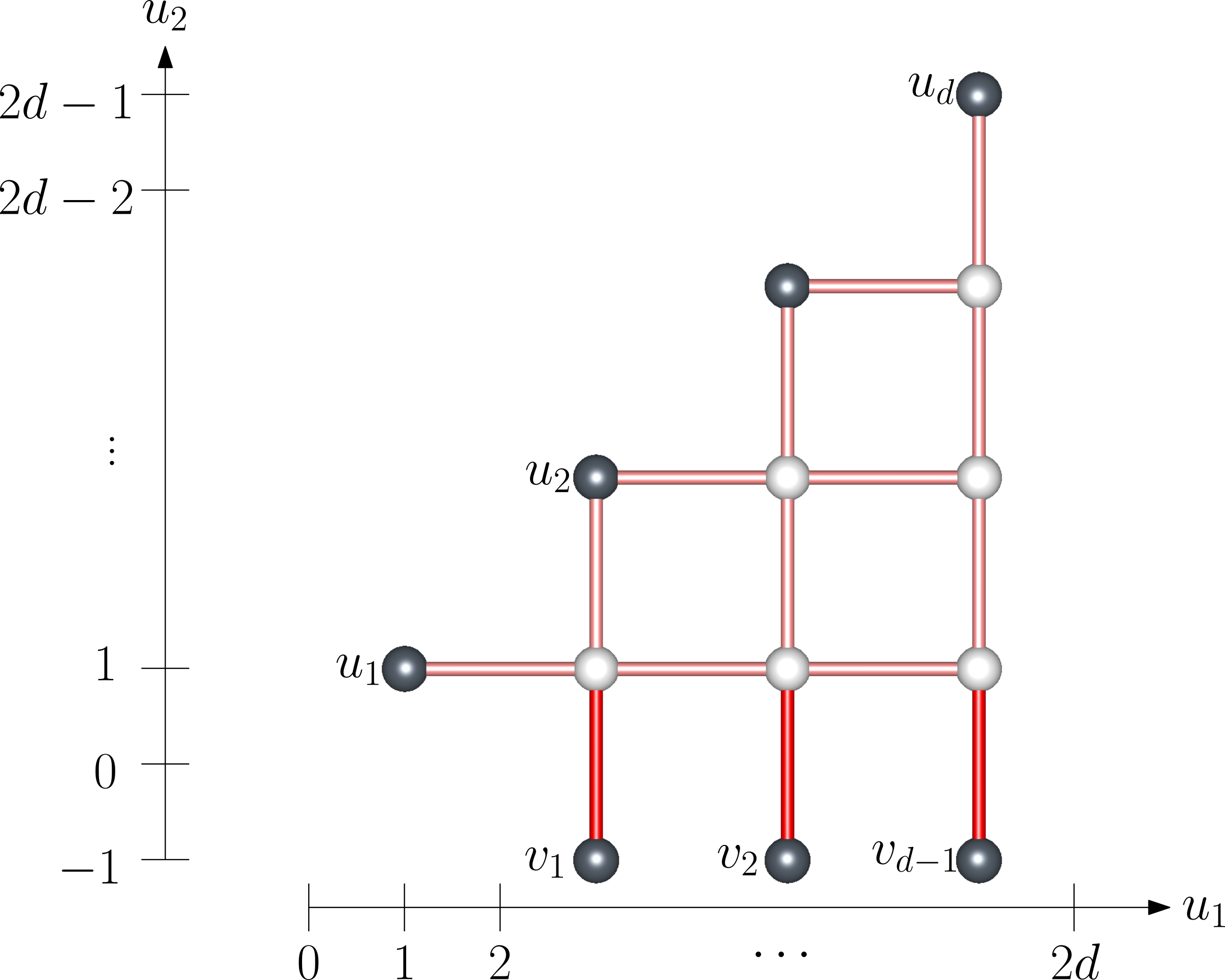}
    \caption{
    \begin{tabular}{ |c||c|c| } 
     \hline
     \multicolumn{3}{|c|}{ $\Tdecdual = (\Vdecdual, \Edecdual)$} \\
     \hline\hline
     subsets of $\Tdec$ & color & measurement\\ 
     \hline\hline
     $\cL_Z^{*}$ &\textcolor{red}{ \rule{1.5cm}{0.3cm}} & $Z$-basis \\
     \hline 
     $\cZ \setminus \cL_Z^{*}$ &\textcolor{lightred}{ \rule{1.5cm}{0.3cm}} & $Z$-basis \\
     \hline 
    \end{tabular}
    }\label{fig:Tdecdual}
\end{subfigure}
  \caption{The decoding graph~$\Tdec$ and the dual decoding graph $\Tdecdual$. 
  Their internal (external) vertices are illustrated as white (black, resp.) spheres. 
  Edges are colored according to the subsets where their midpoints are contained (see the tables in the subfigures).}
  \label{fig:gluedecodinggraphscl}
\end{figure}

The ``dual decoding subgraph''~$\Tdecdual=(\Vdecdual,\Edecdual)$ of $\Tscdual$ is defined similarly: 
The graph $\Tdecdual$ is the subgraph induced by the set of vertices
\begin{align}
  \Vdecdual = \{(u_1, u_2) \in \Vsc \mid u_1 \geq u_2 > 0\} \setminus \{(1, -1)\} \ ,
\end{align}
and its edges are $\Edecdual := \cZ$. Explicitly, the edges of $\Tdecdual$ are
\begin{align}
  \Edecdual = \{\{u,v\} \in \Vdecdual \times \Vdecdual \mid d(u,v) = 2\} \ .
\end{align}
Observe that the recovery set $\cL_Z^{*}$ is a subset of the edges of $\Tdecdual$.
We will use $\Tdecdual$ to estimate the parity of a certain string restricted to $\cL_Z^{*}$.

The internal vertices of $\Tdecdual$ are
\begin{align}
  \Vdecdualint = \{u \in \Vdecdual \mid \neigh(u) \subset \Edecdual\} \ .\label{eq:tdecdualinterior}
\end{align}
Note that a vertex~$u$ of $\Tdecdual$ is internal if and only if the ``plaquette'' operator $B_u$ of the surface code has support $\supp(B_u) \subset \Edecdual$ contained in the set of edges of $\Tdecdual$. A vertex $u \in \Vdecdual \setminus \Vdecdualint$ which is not internal is called external. We denote the set of external vertices of $\Tdecdual$ by $\Vdecdualext$.

\subsection{Description of the single-shot decoding procedure}\label{sec:scdecodingalgorithm}
We provide the description of algorithm for single-shot surface code decoding. 
\begin{algorithm}[H]
\caption{Single-shot surface code decoding}
\label{alg:fdecsingle}
\begin{algorithmic}[1]
\Require A state on $(\mathbb{C}^2)^{\otimes \cC}$, where $\cC$ 
are the locations of qubits  of a distance-$d$ surface code.
  \State Perform the following measurements: measure
  \begin{center}
  \begin{tabular}{c|c|c}
   every qubit in the set & in the  & denote the outcome by \\
   \hline
$\cX\subset \Edec$ & $X$-basis & $x\in \{0,1\}^\cX$\ \ \\
$\cZ\subset \Edecdual$ & $Z$-basis & $z\in \{0,1\}^\cZ$\ .
  \end{tabular}
  \end{center}
  \label{it:stepmeasurementsingle}
  \State Compute the boundaries of the subsets of edges~$x$ and $z$ in $\Tdec$ and $\Tdecdual$, respectively, i.e., set 
  \begin{align}
  s &\leftarrow \partial x\label{eq:sdefinitionboundarysingle}\\
  s^{*} &\leftarrow \partial^*z\ .
  \end{align}\label{state:setsandsdualsingle}
  \State Compute minimal matchings \label{state:setmsingle}
  \begin{align}
  m &\leftarrow \mmatch(s)\\
m^* &\leftarrow \mmatch^*(s^{*})\ 
\end{align}
\State Compute the bits\label{it:syndromebitestimationsingle}
\begin{align}
\hat{s}^X&\leftarrow \ztwoinner{m}{\cL_X}\label{eq:hatsxdefsingle}\\
\hat{s}^Z&\leftarrow \ztwoinner{m^{*}}{\cL^{*}_{Z}}\label{eq:hatszdefsingle}\ .
\end{align}\label{state:sydnromebitsxszsingle}
  \State Determine  the bits
  \begin{align}
  \hat{c}_X &\leftarrow  \ztwoinner{x}{ \cL_X} \oplus \hat{s}^X\\
  \hat{c}_Z& \leftarrow  \ztwoinner{z}{\cL^*_Z} \oplus \hat{s}^Z
  \end{align} \label{state:estimateofcsingle}
  \State Apply $Z^{\hat{c}_X} X^{\hat{c}_Z}$ to $q$.\label{state:estimatecorrectionsingle}
  \State Return the resulting qubit state on qubit~$q$.
  \end{algorithmic}
\end{algorithm}

The following notation is used in the algorithm: For any two sets $A, B \in \cC$, we define
\begin{align}
    \ztwoinner{A}{B} := \bigoplus_{j \in \cC} A_j B_j
\end{align}
where $A$ and $B$ are considered as strings (given by the characteristic function of the corresponding set) in $\{0,1\}^{\cC}$ and $\oplus$ denotes XOR on bits.
The value~$\ztwoinner{A}{B}$ can  equivalently be defined as the parity of $\abs{A \cap B}$.
Considering the set of strings $\{0,1\}^{\cC}$ as a $\mathbb{Z}_2$-linear space, the map $\ztwoinner{\cdot}{\cdot}$ is clearly bilinear.

We also note that the runtime of Algorithm~\ref{alg:fdecsingle} 
is dominated by the computation of the matchings in step~\eqref{state:setmsingle}, i.e., 
runs in time $O(\abs{\cC}^3)$, see the remark at the end of Section~\ref{sec:graphswithinternalvertices}.

\section{A bound on matching  with internal vertices\label{sec:boundmatchingmain}} 
In this section, we 
provide a general combinatorial framework for analyzing 
decoding problems involving minimal matchings.  These arguments will be applied in Section~\ref{sec:singleshotdecode} to analyze the success probability 
of the single-shot decoding Algorithm~\ref{alg:fdecsingle} introduced in Section~\ref{sec:singleshotdecodingprotocol} in the presence of local stochastic noise. We will also use these general concepts in Section~\ref{sec:entanglementgenerationnoisy} to analyze the robustness of our entanglement generation protocol based on the 3D~cluster state.

We first show in Section~\ref{sec:cycledecomposition} that a cycle, a set of edges in a graph with trivial boundary, can be decomposed as a disjoint union of simple closed paths and simple paths connecting two external vertices through internal vertices. 

Next, in Section~\ref{sec:resfunction}, we introduce the notion of a local stochastic subset of edges. This is a random variable~$\cE$ which is
a random subset of edges, and whose distribution is exponentially decaying in the cardinality of the subset. We also introduce what we call  the resilience function associated with a fixed subset~$\cL$ of edges. It is defined by a combinatorial expression in terms  of a sum over simple closed paths and simple paths connecting two external endpoint vertices through internal vertices.
We show that the value of resilience function is an upper bound on the probability that the parities of a local stochastic subset~$\cE$  and the minimum matching $\mmatch{(\partial \cE})$ of the boundary of~$\cE$ restricted to~$\cL$ differ. In particular, this will allow us to use the parity of~$\mmatch(\partial \cE) \cap \cL$ as an estimate of the parity of~$\cE \cap \cL$. This is the key statement that we will subsequently use to analyze resilience against local stochastic noise.

\subsection{Cycle decomposition}\label{sec:cycledecomposition}
Let $(G = (V, E),\Vint)$ be a graph with a distinguished set of internal vertices~$\Vint \subset V$. Let $\partial : 2^E \to 2^{\Vint}$ be the boundary map defined in Section~\ref{sec:graphswithinternalvertices}.
  A subset~$\cE\subset E$ is called a cycle if the boundary~$\partial \cE=\emptyset$ is trivial. We denote by~$Z(G)$ the set of cyles on~$G$. 
  By the linearity of~$\partial$, the set of cycles $Z(G)$ forms a linear subspace of~$\mathbb{Z}_2^{E}$ called the cycle space of $G$. 
  We will consider certain cycles which span~$Z(G)$.

 A closed loop~$L=\{e_1, \dots, e_{\ell}\} \subseteq E$ is given by a sequence~$(e_1, \dots, e_{\ell})$ of distinct edges such that there exists a set of vertices~$\{v_1, \dots, v_{\ell}\}$  with the property  that $e_{j} = \{v_j, v_{j+1}\}$ for all $j = 1, \dots, \ell - 1$ and $e_{\ell} = \{v_{\ell}, v_1\}$.
 In particular, any  closed loop~$L$ belongs to $Z(G)$, that is, $\partial L = \emptyset$.
 A simple closed loop~$L=\{e_1, \dots, e_{\ell}\}$ is a closed loop with the property that the vertices~$\{v_1, \dots, v_{\ell}\}$ are distinct.
 We note that any closed loop can be decomposed into a set of pairwise (edge-)disjoint simple closed loops.
 We will denote the set of simple closed loops on~$G$ by~$\Zcirc(G)$.

 We say that a simple (i.e., not  self-intersecting) path~$P=\{e_1, \dots, e_{\ell}\}\subseteq E$   connects two distinct external vertices through internal vertices if
  there are distinct  vertices~$v_1,v_{\ell+1}\in\Vext $ and~$v_2,\ldots,v_{\ell}\in \Vint$ 
  such that $e_j=\{v_j,v_{j+1}\}$ for $j=1,\ldots,\ell$.
Again, every such path~$P$ belongs to~$Z(G)$ and the edges $\{e_j\}_{j=1}^\ell$ are pairwise distinct by construction.
We will denote by $\Zpathext(G)$ the set of simple paths on~$G$ connecting distinct external vertices through internal vertices.

 The union of $Z_\circ(G)$ and $\Zpathext(G)$    spans the cycle space~$Z(G)$. In fact, a stronger property holds:
\begin{lemma}\label{lem:cycledecomposition}
  Let $\cZ\in Z(G)$ be a cycle. Then $\cZ$ is the disjoint union of simple closed loops and simple paths with endpoints in~$V^{ext}$ through internal vertices.
  That is, there is a family 
\begin{align}
\{\cZ_\alpha\}_{\alpha}\subset Z_\circ(G)\cup \Zpathext(G)
\end{align}
of cycles such that each~$\cZ_\alpha$ is either a simple closed loop or a simple path with distinct endpoints in~$V^{ext}$ passing through internal vertices, and
\begin{align}
\cZ=\bigcup_\alpha \cZ_\alpha\qquad\textrm{ with }\qquad \cZ_\alpha\cap \cZ_\beta=\emptyset\textrm{ for }\alpha\neq\beta\ .
\end{align}
\end{lemma}
\begin{proof}
Let $\cZ\in Z(G)$ be a cycle. 
Suppose that there is a simple closed loop~$L\in Z_{\circ}(G)$ that is contained in~$\cZ$.
Then it follows immediately from the definitions that $\cZ':=\cZ\backslash L$ is also a cycle.
Proceeding inductively, we may decompose~$\cZ$ as
\begin{align}
\cZ&=\tilde{\cZ}\cup \bigcup_{j=1}^r L_j\  \label{eq:czdecompositionclosedloops}
\end{align}
with pairwise disjoint simple closed loops~$L_1,\ldots,L_r$ (with $r=0$ if~$\cZ$ does not contain any simple closed loop) such that~$\tilde{\cZ}$ contain does not contain any (simple) closed loop.

We claim that~$\tilde{\cZ}$ is either empty or a union~
\begin{align}
\tilde{\cZ}=\bigcup_{k=1}^s P_k\label{eq:pairwisedisjointpath}
\end{align} of pairwise disjoint paths~$\{P_k\}_{k=1}^s\subset \Zpathext(G)$ with distinct endpoints in~$V^\mathrm{ext}$.
Consider the following two cases:
\begin{enumerate}[(i)]
\item
If  $\tilde{\cZ}$ does not contain an edge~$e\in \tilde{\cZ}$ with $e\cap \Vext\neq \emptyset$, i.e., if $\tilde{\cZ}$ only touches internal vertices of~$V$, then the condition~$\partial \tilde{\cZ}=\emptyset$ implies that
the number of edges belonging to~$\tilde{\cZ}$ which are incident on any internal vertex~$v\in V^\mathrm{int}$ is even. Because~$\tilde{\cZ}$ does not contain any closed loop by construction, it follows that $\tilde{\cZ}=\emptyset$.

\item
Suppose there is an edge~$e\in\tilde{\cZ}$ such that $e\cap (V \setminus V^\mathrm{int})\neq \emptyset$.
Let us assume that $e=\{v,w\}$ where $v\in V\backslash V^\mathrm{int}$ is an external vertex. 
We can then construct a path~$P=\{e_1,\ldots,e_\ell\}\subset\tilde{\cZ}$ that has~$v=v_s$ as starting point, first traverses the edge~$e_1:=e$, and finally ends at an external vertex~$v_e\neq v_s\in V\backslash V^\mathrm{int}$.
Indeed, if $w\in V\backslash V^\mathrm{int}$ is an external vertex, then $P=\{e\}$ is such a path (of length~$1$). 
If $w\in V^\mathrm{int}$ is an internal vertex, then the condition~$\tilde{\cZ}\in Z(G)$ implies that
the number of edges belonging to~$\tilde{\cZ}$ and incident on~$w$ is even.
In particular, there is an edge~$e_2$ that is distinct from~$e_1$.
We can proceed inductively to construct a path~$\{e_1,\ldots,e_\ell\}\subset\tilde{\cZ}$.
Moreover, the path satisfies the desired properties, i.e., $P \in \Zpathext(G)$:
The condition~$\partial\tilde{\cZ}=\emptyset$ guarantees that the path does not end at an internal vertex, and this implies that
the constructed path only touches internal vertices except for the endpoints~$v_s$ and $v_e$ which are external vertices.
In addition, since~$\tilde{\cZ}$ does not contain any closed loop, the path is simple.
In particular, its endpoint~$v_e$ must be an external vertex which is distinct from~$v_s$.

Once we have found such a path~$P$, we can remove~$P$ from~$\tilde{\cZ}$, i.e., we can consider the set~$\tilde{\cZ}':=\tilde{\cZ}\backslash L$.
It is easy to check that $\tilde{\cZ}'$ is also a cycle which does not contain any closed loops. Proceeding inductively, we obtain the claim~\eqref{eq:pairwisedisjointpath}.

\end{enumerate}
Combining~\eqref{eq:czdecompositionclosedloops} and~\eqref{eq:pairwisedisjointpath} implies the claim.
\end{proof}

 \subsection{Local stochastic subsets and the resilience function}\label{sec:resfunction}
 A random subset~$\cE$ of~$E$ is a random variable specified by a distribution over subsets of~$E$. We call a random subset~$\cE$ local stochastic with parameter~$p\in [0,1]$ if and only if
 \begin{align}
 \Pr\left[\cF\subseteq \cE\right]\leq p^{|\cF|}\qquad\textrm{ for all }\qquad \cF\subset E\ .
 \end{align}
 This will be denoted~$\cE\sim\cN(p)$.

We will consider the difference of a random subset~$\cE$ and a minimum matching~$\mmatch\left(\partial\cE\right)$ computed from the boundary~$\partial\cE$ of~$\cE$. More precisely, we consider the parity of this difference restricted to a subset~$\cL\subset E$. The following definition captures the relevant combinatorics of our estimates.
\begin{definition}\label{def:res}
Let $(G=(V,E),V^\mathrm{int})$ be a graph with internal vertices~$V^\mathrm{int}$. 
Let $\cL\subset E$ be a subset of edges. 
Then the resilience function~$\res_\cL:[0,1]\rightarrow\mathbb{R}$ of~$\cL$ is defined as
\begin{align}\label{eq:defofres}
\res_\cL(p):&=\sum_{\ell=1}^\infty 
\binom{\ell}{\lceil \ell/2\rceil}\cdot \left|\left\{\cR\in
Z_\circ(G)\cup \Zpathext(G)\ |\ |\cR|=\ell\textrm{ and }\ztwoinner{\cR}{\cL}=1\right\} \right|\cdot p^{\lceil \ell/2\rceil}\ .
\end{align}
 \end{definition}
 The resilience function~$\res_\cL:[0,1]\rightarrow\mathbb{R}$
 provides the following upper bound on the probability that the parity of~$\cE$ and~$\mmatch\left(\partial\cE\right)$  restricted to~$\cL$ differ:
\begin{proposition}\label{prop:maincombinatorics}
Let $\cE\sim \cN(p)$ be a local stochastic subset of~$E$ with parameter~$p\in [0,1]$. 
Let $\cL\subset E$ be a subset of edges. Then
\begin{align}
\Pr\left[
\ztwoinner{\cE\oplus \mmatch\left(\partial\cE\right)}{\cL}=1\right]\leq \res_{\cL}(p)\ .
\end{align}
\end{proposition}
\begin{proof}
The proof follows the proof of~\cite[Lemma 21]{bravyiQuantumAdvantageNoisy2020} which is based on ideas of~\cite{fowlerProofFiniteSurface2012,dennisTopologicalQuantumMemory2002}.
Define 
  \begin{align}
    \mathrm{FAIL} := \left\{\cY \subseteq  E \mid \ztwoinner{\cY \oplus \mmatch\left(\partial \cY\right)}{\cL} = 1\right\}
  \end{align}
such that the quantity of interest becomes~$\Pr\left[\cE\in\mathrm{FAIL}\right]$. 
  
For any set~$\cY\subseteq E$, the set  $\cY \oplus \mmatch\left(\partial \cY\right)$ 
has trivial boundary by definition of~$\mmatch$. 
By Lemma~\ref{lem:cycledecomposition}, it belongs to the cycle space~$Z(G)$ and can be   decomposed as a  union 
\begin{align}
\cY \oplus \mmatch\left(\partial \cY\right)&=\bigcup_{\alpha}\cY_\alpha\ ,
\end{align}
where each~$\cY_\alpha\in Z_\circ(G)\cup \Zpathext(G)$ is either a simple closed loop or a simple path connecting external vertices through internal vertices, and where
the sets~$\cY_{\alpha}$ are pairwise disjoint.
It follows that for any~$\cY\subseteq E$ with $\cY\in \mathrm{FAIL}$, there is at least one~$\cY_\alpha$  with the property that~$\ztwoinner{\cY_\alpha}{\cL}=1$.
That is, any such~$\cY$ contains a subset~$\cR \in \Zcirc(G) \cup \Zpathext(G)$ which satisfies~$\ztwoinner{\cR}{\cL}=1$. 

It follows that
\begin{align}
\Pr\left[\cE\in \mathrm{FAIL}\right] & \leq \Pr\left[\exists \cR\in  \Zcirc(G) \cup \Zpathext(G) : \cR\subset \cE\oplus \mmatch\left(\partial\cE\right)\textrm{ and }\ztwoinner{\cR}{\cL}=1\right] \ . \label{eq:pfailexpression}
\end{align}
For simplicity, let us define
\begin{align}
\cP(\cL) := \{\cR \in \Zcirc(G) \cup \Zpathext(G) \mid \ztwoinner{\cR}{\cL}=1\} \ .
\end{align}
With the union bound applied to~\eqref{eq:pfailexpression} we obtain
\begin{align}
\Pr\left[\cE\in \mathrm{FAIL}\right] &\leq \sum_{\cR\in\cP(\cL)}
\Pr\left[\cR\subseteq \cE\oplus \mmatch\left(\partial\cE\right)\right]\  .\label{eq:unionboundmainx}
\end{align}
Let $\cE\subset E$ be arbitrary and set $\cM:=\mmatch\left(\partial\cE\right)$. 
Suppose that  $\cR\in\cP(\cL)$ satisfies
\begin{align}
\cR\subseteq \cE\oplus \cM\ . 
\end{align}
We claim that at least half of the edges of~$\cR$ 
do not belong to $\cM$, i.e.,
\begin{align}
\left|\cR\backslash \cM\right|\geq  |\cR|/2\ .\label{eq:claimzmlowerbound}
\end{align}
For the sake of contradiction, assume  that
 \begin{align}
 \left|\cR\backslash \cM\right|< |\cR|/2\ .\label{eq:contradictionassumption}
 \end{align}  Then the set~$\cM':=\cR\oplus \cM$ is a matching
 of~$\partial\cE$ since
 \begin{align}
 \partial \cM'&=\partial \left(\cR\oplus \cM\right)\\
 &=\partial \cR\oplus \partial \mmatch\left(\partial\cE\right)\\
 &=\partial\cE\ ,
 \end{align}
 where we used that every $\cR\in\cP(\cL)$ has trivial boundary~$\partial\cR=\emptyset$. 
 Furthermore, we have 
 \begin{align}
 |\cM'|&=|\cR\backslash \cM|+|\cM\backslash \cR|\\
 &=|\cR\backslash \cM|+|\cM|-|\cM\cap\cR|\\
 &=|\cR\backslash \cM|+|\cM|-(|\cR|-|\cR\backslash \cM|)\\
 &=(2|\cR\backslash\cM|-|\cR|)+|\cM|\\
 &<|\cM|\ ,
 \end{align}
 where we used our assumption~\eqref{eq:contradictionassumption}. Thus~$\cM'$ is
 a matching of~$\partial\cE$ of smaller cardinality than~$\cM$, contradicting the fact that~$\cM$ is a minimum matching of~$\partial\cE$. 
 
Eq.~\eqref{eq:claimzmlowerbound} shows that~$\cR\subset\cE\oplus \cM$ together with~$\cR\in\cP(\cL)$ imply that
\begin{align}
|\cR\cap \cE|&\geq |\cR|/2\ .\label{eq:zeintersetlowerbound}
\end{align}
Indeed, the fact that $\cR$ is contained in the symmetric difference 
$\cE\oplus \cM$ implies that $\cR$ is the disjoint union~$\cR=(\cR\cap \cE)\cup (\cR\cap \cM)$
and thus $\cR\cap \cE \supseteq \cR\backslash \cM$. This implies the claim~\eqref{eq:zeintersetlowerbound} together with~\eqref{eq:claimzmlowerbound}. 

With~\eqref{eq:zeintersetlowerbound}, we conclude that for every $\cR\in\cP(\cL)$, we have 
$|\cR\cap \cE|\geq \lceil |\cR|/2\rceil$ (because $|\cR\cap \cE|$ is an integer) and thus
\begin{align}
\Pr\left[\cR\subseteq \cE\oplus \mmatch\left(\partial\cE\right)\right]&
\leq \Pr\left[ 
|\cR\cap \cE|\geq |\cR|/2\right]\\
&\leq \Pr\left[\exists \delta\subseteq \cR\ \textrm{ with } |\delta|=\lceil |\cR|/2 \rceil\textrm{ and }\delta\subseteq\cE \right]\ .
\end{align}
Here we used that for any set~$\cE$ with $|\cR\cap \cE|\geq \lceil |\cR|/2\rceil$, there is a subset $\delta\subseteq \cR\cap \cE$ of size exactly equal to~$\lceil |\cR|/2 \rceil$. 
We can sum over subsets~$\delta\subset\cR$ of cardinality~$|\delta|=\lceil |\cR|/2 \rceil$ to obtain with the union bound
\begin{align}
\Pr\left[\cR\subseteq \cE\oplus \mmatch\left(\partial\cE\right)\right]&\leq \sum_{
\substack{
\delta\subseteq \cR\\
|\delta|=\lceil |\cR|/2 \rceil
}}
\Pr\left[\delta\subseteq\cE\right]\\
&\leq \binom{|\cR|}{\lceil |\cR|/2 \rceil}
\cdot p^{\lceil |\cR|/2 \rceil}\ ,
\end{align}
where we used that $\cE\sim\cN(p)$, i.e., the fact that~$\cE$ is local stochastic. The claim follows by inserting this into~\eqref{eq:unionboundmainx} and applying the definition of~$\res_\cL(p)$. 
\end{proof}

\newpage

\section{Single-shot decoding under local stochastic noise\label{sec:singleshotdecode}}
In this section, we establish our main result about single-shot decoding of the 2D surface code. 
We consider a situation where a surface-code encoded state is corrupted by a local stochastic error, and where the output state~$\rho_{\mathrm{out}}$ is obtained by running Algorithm~\ref{alg:fdecsingle} on this corrupted logical state.

In Section~\ref{sec:sccorrectdecodingcondition}, we consider a half-encoded Bell state $\ket{\Phi_{\cC R}}$ defined on the surface code system~$\cC$ of the surface code and an auxiliary qubit system~$R$. Supposing that Algorithm~\ref{alg:fdecsingle} is executed on the surface code system of the Bell state corrupted by a Pauli error~$E$ on the surface code qubits, we show that the output state then is one of the Bell basis states~$\{\Phi_{(\alpha, \beta)}\}_{\alpha,\beta\in \{0,1\}}$, and find corresponding expressions for~$\alpha$ and~$\beta$.

In Section~\ref{sec:resiliencetdectdecdual}, we compute upper bounds on the resilience function of the recovery sets~$\cL_X$ and $\cL_Z^{*}$.
Using these bounds and the expressions for~$(\alpha,\beta)$ obtained in Section~\ref{sec:sccorrectdecodingcondition}, 
 we prove in Section~\ref{sec:boundscdecoding} that for any input state corrupted by a local stochastic error of strength~$p$ below a  certain threshold, the overlap~$\ev{\rho_{\mathrm{out}}}{\Psi}$ between $\Psi$ and the output state~$\rho_{\mathrm{out}}$ is at least $1-O(p)$.

\subsection{Decoding on a corrupted half-encoded Bell state}\label{sec:sccorrectdecodingcondition}
In the following, we  use the logical Pauli operators 
\begin{align}
  \overline{X}&=X(\cL_X)X_q\\
  \overline{Z}&=Z(\cL^*_Z)Z_q\ 
\end{align}
of the surface code. 
We consider the action of  Algorithm~\ref{alg:fdecsingle}  on a certain bipartite state of the form~$\Phi_{\cC R}\in (\mathbb{C}^2)^{\otimes \cC\cup \{R\}}\cong (\mathbb{C}^2)^{\otimes \cC}\otimes\mathbb{C}^2$, where $R$ is a reference qubit. Specifically, we are interested in the Bell state~$\Phi_{\cC R}$ between a surface-code encoded logical qubit and the reference qubit~$R$.
We call this state the half-encoded Bell state (since the qubit~$R$ is ``physical'', i.e., not encoded). The  state has stabilizer generators~$\{A_v\}_v\cup \{B_f\}_f\cup \{S^X,S^Z\}$, where $\{A_v\}_v$
and $\{B_f\}_f$ are the usual surface code generators (associated with vertices and faces  of the surface code lattice), and where
\begin{align}
  \begin{matrix}
    S^X &:=&\overline{X}X_R&=&X(\cL_X)X_qX_R\\
    S^Z &:=&\overline{Z}Z_R&=&Z(\cL^*_Z)Z_qX_R\ ,
    \end{matrix}\label{eq:sxszsingleexpr}
  \end{align}
are the Bell state stabilizer generators. In addition to the stabilizer generators~$\{S^X,S^Z\}$, the following sets~$\{S^u\}_{u\in \Vdecint}$ and~$\{S^{u^*}\}_{u^*\in \Vdecdualint}$ of stabilizers of the state~$\Phi_{\cC R}$ play a special role in our argument. They are defined as follows: Any internal vertex~$u\in \Vdecint$ of~$\Tdec$ is associated with it a stabilizer generator~$S^u$ of the surface code consisting of Pauli-$X$-operators only, namely~$S^u=A_u$.
In a similar vein, any internal vertex~$u^*\in \Vdecdualint$  of $\Tdecdual$ gives rise to a stabilizer~$S^{u^*}$ consisting of $Z$-type operators, namely the face operator~$B_{u^*}$ when $u^*$ is considered as a face of the original surface code lattice~$\Tsc$. 
Observe that for each $u\in \Vdecint$, 
the support  $\supp(S^u)=\incident(u)$ of~$S^u$  is 
equal to the set of edges incident on~$u$ in the graph~$\Tdec$. Similarly, for each $u\in \Vdecdualint$, we have $\supp(S^u)=\incident^*(u)$ where $\incident^*(u)$ is the set of edges incident on~$u$ in the graph~$\Tdecdual$. That is, we have 
\begin{align}
S^u&=\begin{cases}
\prod_{v\in \incident(u)} X_v\qquad\textrm{ for }\qquad u\in \Vdecint\\
\prod_{v\in \incident(u)} Z_v\qquad\textrm{ for }\qquad u\in \Vdecdualint \ .
\end{cases}\label{eq:Suincidentrelations}
\end{align}

Consider the result of running the algorithm on a corrupted state
\begin{align}
\ket{\Psi_{\textrm{in}}}&=E\ket{\Phi_{\cC R}}\ \label{eq:noisyclusterstatesingle}
\end{align}
for a fixed Pauli error~$E$. 
We assume here that $\supp(E)\subseteq \cC$, i.e., the error acts on the surface code qubits only (and does not affect the reference system~$R$). Let
\begin{align} \label{eq:pmstateinalgsingle}
    \ket{\psi_{\mathrm{pm}}(x,z,E)} = \frac{1}{\sqrt{p(x,z|E)}} \left( \bra{x}_{\mathcal{X}}H(\mathcal{X})^\dagger\otimes \bra{z}_{\mathcal{Z}} \otimes I_{\{q,R\}} \right) E \ket{\Phi_{\cC R}}\ 
\end{align}
be the state obtained  after Step~\ref{it:stepmeasurementsingle} of the algorithm  on qubit~$q$. 
Here we denote by $p(x,z|E)$ the probability of obtaining the outcomes~$(x,z)$ given a Pauli error~$E$, i.e., for the input state~\eqref{eq:noisyclusterstatesingle}. 
It will be convenient in the following to use the definition
\begin{align}\label{eq:defsyn}
    \syn(P,Q) = 
    \begin{cases*}
    0 & \textrm{ if $P$ and $Q$ commute}\\
    1 & \textrm{ otherwise }
    \end{cases*}
\end{align}
 for any two Pauli operators $P$ and $Q$. 
Then the following holds:
\begin{lemma} \label{lem:pmstatesingle}
  Let $S^X$ and $S^Z$ be the operators~\eqref{eq:sxszsingleexpr}. Let
  \begin{align}
  c_X &:=  \ztwoinner{x}{\cL_X} \oplus \syn(S^X, E)\label{eq:cxdefinitionsingle}\\
  c_Z &:= \ztwoinner{z}{\cL^*_Z} \oplus \syn(S^Z, E)\ . \label{eq:czdefinitionsingle}
  \end{align}
  Then the post-measurement state~$\psi_{\mathrm{pm}}(x,z,E)$ on the two qubits~$q,R$ 
  after Step~\ref{it:stepmeasurementsingle} 
  of Algorithm~\ref{alg:fdecsingle},   given an initial error~$E$ and measurement outcomes~$(x,z)$ is the Bell state (cf.~\eqref{eq:defbellstate})
  \begin{align}
      \ket{\psi_{\mathrm{pm}}(x,z,E)} &=\ket{\Phi_{(c_X,c_Z)}}\ .
  \end{align}
\end{lemma}
\begin{proof}
  The claim is an immediate consequence of the fact that~$S^X$ and $S^Z$ are stabilizers of $\ket{\Phi_{\cC R}}$,
  as well as the expressions~\eqref{eq:sxszsingleexpr} of the stabilizer generators: Replacing $\ket{\Phi_{\cC R}}$ in \eqref{eq:pmstateinalgsingle} by $S^X\ket{\Phi_{\cC R}}$, we observe that the post-measurement state~$\ket{\psi_{\mathrm{pm}}}=\ket{\psi_{\mathrm{pm}}(x,z,E)}$ 
satisfies
  \begin{align}
    \sqrt{p}  \ket{\psi_{\mathrm{pm}}}    &= \left( \bra{x}_{\mathcal{X}}H(\mathcal{X})^\dagger \otimes \bra{z}_{\mathcal{Z}}\otimes I_q \right) E S^X\ket{\Phi_{\cC R}}\\
   &=  (-1)^{\syn(S^X, E)} \left(  \bra{x}_{\mathcal{X}}H(\mathcal{X})^\dagger\otimes \bra{z}_{\mathcal{Z}} \otimes I_q\otimes I_R \right) S^X E  \ket{\Phi_{\cC R}}\\
    &=  (-1)^{\ztwoinner{x}{\cL_X} \oplus \syn(S^X, E)} \left(\bra{x}_{\mathcal{X}}H(\mathcal{X})^\dagger \otimes \bra{z}_{\mathcal{Z}}\otimes X_qX_R\right) E
\ket{\Phi_{\cC R}}\ ,
  \end{align}
    where $p := p(x,z|E)$ and where we used~\eqref{eq:sxszsingleexpr} as well as the identity
      \begin{align}
      X(\cL_X)H(\cX)\ket{x}_{\cX}&=(-1)^{\ztwoinner{x}{\cL_X}}H(\cX)\ket{x}_{\cX}\ .
      \end{align}
  We conclude that
    \begin{align}
      \ket{\psi_{\mathrm{pm}}}   &= (-1)^{\ztwoinner{x}{\cL_X} \oplus \syn(S^X, E)} X_{q}X_{R}
\ket{\psi_{\mathrm{pm}}}   \\
&= (-1)^{c_X} X_{q} X_{R}\ket{\psi_{\mathrm{pm}}}\ ,
      \end{align}
      by definition of~$c_X$. 

      In an analogous manner, we can check that $\ket{\psi_{\mathrm{pm}}} = (-1)^{c_Z} Z_{q}Z_{R} \ket{\psi_{\mathrm{pm}}}$ by
replacing
$\ket{\Phi_{\cC R}}$ in~\eqref{eq:pmstateinalgsingle} by $S^Z\ket{\Phi_{\cC R}}$.  Thus the two-qubit state~$\ket{\psi_{\mathrm{pm}}}$ on systems~$q,R$      is stabilized by~$\{(-1)^{c_X}X_{q}X_{R},(-1)^{c_Z}Z_{q}Z_{R}\}$. This implies the claim.
\end{proof}

According to Lemma~\ref{lem:pmstatesingle}, the post-measurement state is fully determined by the single-qubit measurement results~$(x,z)$ and the bits~$\syn(S^X, E)$ and~$\syn(S^Z, E)$.
To analyze these bits, we introduce decomposition
\begin{align}\label{eq:xzdecomposition}
    E=E^XE^Z
\end{align}
of $E$ into a product of Pauli-$X$ and Pauli-$Z$ operators $E^X$ and $E^Z$, respectively.  (We note that if $E\sim\cN(p)$ is a local stochastic error of strength~$p$, the same is true for both~$E^X$ and $E^Z$, i.e., $E^X\sim\cN(p)$ and $E^Z\sim\cN(p)$, a fact we will use below.)

In~\eqref{eq:xzdecomposition}, the global phase is ignored since we will only need the supports~$\supp(E^X)$ and $\supp(E^Z)$.
We show that the bits~$\syn(S^X,E)$ and $\syn(S^Z,E)$  are determined by the physical errors on qubits~$\{q_1,q_2\}$ as well as a certain restriction of the error as specified in the following statement:
  \begin{lemma} \label{lem:actualsyndsingle}
  Let $E$ be a Pauli operator on $\mathcal{C}$ and let $S^X$ and $S^Z$ be the Pauli operators defined by~\eqref{eq:sxszsingleexpr}. Decompose~$E$ into a product~$E = E^XE^Z$ of Pauli-$X$- and Pauli-$Z$-operators, respectively. 
  Then 
  \begin{align} 
    \syn(S^X, E) &= \ztwoinner{\{q\}}{\supp(E^Z)} \oplus \ztwoinner{\cL_X}{\supp(E^X)}\label{eq:synviaEglsingle}\\
    \syn(S^Z, E) &= \ztwoinner{\{q\}}{\supp(E^Z)} \oplus \ztwoinner{\cL_Z^*}{\supp(E^Z)}\ .\label{eq:synviaEgldualsingle}
  \end{align}
\end{lemma}
\begin{proof}
  Let $E$ be an arbitrary Pauli operator on~$\cC$.
  With expression~\eqref{eq:sxszsingleexpr} for~$S^X$, we have 
  \begin{align}
    \syn(S^X,E)&=\syn\left(X_{q}X_{R}X\left(\cL_X\right),E\right)\\
    &= \ztwoinner{\{q\}}{\supp(E^X)} \oplus \ztwoinner{\cL_X}{\supp(E^X)}
  \end{align}
  since $E$ acts trivially on the reference system~$R$. The claim for~$\syn(S^Z,E)$ follows similarly.
\end{proof}
Unfortunately, the operators~$E^X$ and $E^Z$ are not determined by the measurement results~$(x,z)$. In particular, the bits~$\syn(S^X,E)$ and $\syn(S^Z,E)$ cannot be computed from the measurement outcomes. Algorithm~\ref{alg:fdecsingle} therefore generates estimates~$\hat{s}^X$ and~$\hat{s}^Z$ for these bits, see step~\ref{it:syndromebitestimationsingle}. 

To motivate the definition of these bits, and to derive a sufficient condition for when the algorithm is correct (i.e., generates the Bell state~$\Phi$), we discuss how the measurement results~$(x,z)$ constrain the error~$E$, or more precisely
the operators~$E^X$ and $E^Z$. That is, we consider what information can be gathered from these measurement outcomes. Observe  that the operators~$(E^X,E^Z)$ are fully determined by the supports~$\supp(E^X)$ and~$\supp(E^Z)$. 
We obtain constraints on these sets expressed in terms of the strings~$s,s^*$ computed in Step~\ref{state:setsandsdualsingle} of  Algorithm~\ref{alg:fdecsingle}.

\begin{lemma} \label{lem:sisboundarysingle}
The strings $s \in \{0,1\}^{\Vdecint}$ and $s^* \in \01^{\Vdecdualint}$ computed (from the measurement outcomes~$(x,z)$) in Step~\ref{state:setsandsdualsingle}  satisfy
  \begin{align}
    s &= \partial \supp(E^Z)\label{eq:stglclaimsingle}\\
    s^{*} &= \partial^* \supp(E^X)\ .\label{eq:stglclaimdualsingle}
 \end{align}
\end{lemma}
\begin{proof}
  Let $s = (s_u) \in \01^{\Vdecint}$ be the string computed in step~\ref{state:setsandsdualsingle}, i.e.,~$s=\partial x$. Let $u\in \Vdecint$.
  Since $s_u=(\partial x)_u=\ztwoinner{x}{\incident(u)}$
  and  $\left(\partial \supp(E^Z)\right)_u=\ztwoinner{\incident(u)}{\supp(E^Z)}$, we need to show that
  \begin{align}
  \ztwoinner{x}{\incident(u)}&=\ztwoinner{\incident(u)}{\supp(E^Z)}\qquad\textrm{ for every }\qquad u\in \Vdecint\ .\label{eq:toprovexuinnersingle}
    \end{align}
    Note that $S^u$ stabilizes the (uncorrupted) state~$\ket{\Phi_{\cC R}}$ (by definition of the latter).
    Replacing $\ket{\Phi_{\cC R}}$ in~\eqref{eq:pmstateinalgsingle} by $S^u\ket{\Phi_{\cC R}}$ yields 
  \begin{align}
    \ket{\psi_{\mathrm{pm}}}
    &= \frac{1}{\sqrt{p}} \left(  \bra{x}_{\mathcal{X}}H(\mathcal{X})^\dagger \otimes \bra{z}_{\mathcal{Z}} \otimes I_{\{q,R\}} \right) E S^u W \ket{0^\cC}\\
    &= \frac{1}{\sqrt{p}} (-1)^{\syn(S^u, E)} \left(\bra{x}_{\mathcal{X}}H(\mathcal{X})^\dagger \otimes \bra{z}_{\mathcal{Z}} \otimes I_{\{q,R\}} \right) S^u E W \ket{0^\cC}\\
    &=  (-1)^{\syn(S^u, E) \oplus  \ztwoinner{x}{\incident(u) }} \ket{\psi_{\mathrm{pm}}} \ .\label{eq:psipmeigenvalsingle}
  \end{align}
Here the third equality follows from the identity
  \begin{align}
  S^u H(\cX)\ket{x}&=
  (-1)^{\ztwoinner{x}{\incident(u) }}H(\cX)\ket{x}\ 
  \end{align}
  which is a consequence of
  the definition of $S^u$ (as an $X$-type
  stabilizer generator the surface code), as well as the fact that for $u\in \Vdecint$, the set~$\supp(S^u)$ does not contain~$q$ for any $u\in \Vdecint$, see Eq.~\eqref{eq:Suincidentrelations}.
  Identity~\eqref{eq:psipmeigenvalsingle}
implies that 
\begin{align}
\syn(S^u,E)=\ztwoinner{x}{\incident(u)}\qquad\textrm{ for  }\qquad u\in V_{dec}^{int}\ .\label{eq:firststepclaimsingle}
\end{align}
 Observe that according to~\eqref{eq:Suincidentrelations}, we also have
      \begin{align}
      \syn(S^u, E) &=
      \syn\left(\left(\prod_{v \in \incident(u)} X_v\right),E\right)      = \ztwoinner{\incident(u) }{\supp(E^Z)}\qquad\textrm{ for any }\qquad u\in \Vdecint\ . \label{eq:claimsecsingle}
    \end{align}
Combining~\eqref{eq:firststepclaimsingle} with~\eqref{eq:claimsecsingle} gives the claim~\eqref{eq:toprovexuinnersingle}.

  The claim for $s^{*}$ can be proved in an analogous manner using~\eqref{eq:Suincidentrelations}.  
\end{proof}
We now argue that the sets~$m\subset \Edec$ and $m^*\subset \Edecdual$ computed
in Step~\ref{state:setmsingle}
of Algorithm~\ref{alg:fdecsingle} can be considered to be proxys for the sets~$\supp(E^Z)\subset \Edec$ and $\supp(E^X)\subset \Edecdual$. More precisely, we may define
Pauli errors
\begin{align}
\begin{matrix}
\widehat{E}^X&=&X(m^*)\\
\widehat{E}^Z&=&Z(m)\ .
\end{matrix}\label{eq:widehatexzex}
\end{align}
We note that by definition of $m$ and $m^*$ as boundaries of subsets of edges in the graph~$\Tdec$ and $\Tdecdual$, respectively,  and the definition of these graphs, neither~$m$ nor $m^*$ contains the qubit~$q$. Thus
\begin{align}
\{q\}\cap \supp(\widehat{E}^X)&=\{q\}\cap \supp(\widehat{E}^Z)=\emptyset\ .\label{eq:nosupportonqhatXhatz}
\end{align}
We trivially have 
\begin{align}
m&=\supp(\widehat{E}^Z)\\
m^*&=\supp(\widehat{E}^X)
\end{align}   by the definition~\eqref{eq:widehatexzex} of these operators, and it  follows immediately from the
   definition of~$m$ and~$m^*$
  that~$(\widehat{E}^X,\widehat{E}^Z)$
  satisfy constraints analogous to   the constraints~\eqref{eq:stglclaimsingle},~\eqref{eq:stglclaimdualsingle}   obeyed by~$(E^X,E^Z)$, i.e.,  we have
  \begin{align}
    s &= \partial\supp(\widehat{E}^Z)\\
    s^{*} &= \partial^*\supp(\widehat{E}^X)\ .
  \end{align}
  Thus $\widehat{E}:=\widehat{E}^X\widehat{E}^Z$ is an error consistent with the observed syndrome~$(s,s^*)$ with the property that
  $\widehat{E}^Z$ and $\widehat{E}^X$
  each are minimum weight matchings in~$\Tdec$ and $\Tdecdual$, respectively.

Because the error~$\widehat{E}$ can be computed from the measurement outcomes, this motivates using the pair~$\left(
\syn(S^X,\widehat{E}),
\syn(S^Z,\widehat{E})
\right)$ as an estimate for 
$\left(
\syn(S^X,E),
\syn(S^Z,E)
\right)$.  Indeed, this is the reasoning underlying the definition
\begin{align}
  \begin{matrix}
  \hat{s}^X&:=&\ztwoinner{\cL_X}{m}\\
  \hat{s}^Z&:=&\ztwoinner{\cL_Z^*}{m^{*}}
  \end{matrix}\label{eq:hatsxszdefclaimsingle}
\end{align}
in  Step~\ref{state:sydnromebitsxszsingle} of the algorithm: It is easy to check that~$(\hat{s}^X,\hat{s}^Z)$ defined in this way satisfy
\begin{align}
\begin{matrix}
  \hat{s}^X&=&\syn(S^X, \widehat{E})\\
  \hat{s}^Z&=&\syn(S^Z, \widehat{E})\ 
  \end{matrix}\label{eq:sxszhatconseqsingle}
\end{align}
because of Expressions~\eqref{eq:synviaEglsingle},~\eqref{eq:synviaEgldualsingle}: Indeed, we have for example 
\begin{align}
\syn(S^X,\widehat{E})&=\syn\left(X(\cL_X)X_qX_R,\widehat{E}^X\widehat{E}^Z\right)\\
&=\ztwoinner{\cL_X}{\supp(\hat{E}^Z)}\\
&=\ztwoinner{\cL_X}{m}\ 
\end{align}
where we used the definition~\eqref{eq:sxszsingleexpr} in the first step, and the fact that $\hat{E}^Z$ has no support on~$\{q,R\}$ in the second step (cf.~\eqref{eq:nosupportonqhatXhatz}).
Identity~\eqref{eq:sxszhatconseqsingle} shows that the pair~$(\hat{s}^X,\hat{s}^Z)$ is a natural estimate for
the bits~$\left(\syn(S^X,E),\syn(S^Z, E)\right)$, the idea being  that~$\widehat{E}$ is a decent estimate of~$E$.

Having motivated the algorithm, let us show that Algorithm~\ref{alg:fdecsingle} produces a Bell state~$\Phi_{(\alpha,\beta)}$. Below, we will use the expressions for~$(\alpha,\beta)$ computed here below to analyze the success probability of the protocol (i.e., the probability that~$(\alpha,\beta)=(0,0)$) for local stochastic noise.

\begin{theorem} \label{thm:successconditionsurfacecode}
  Let $E=E^XE^Z$ be an arbitrary Pauli operator on $(\mathbb{C}^2)^{\otimes \mathcal{C}}$.
  Suppose Algorithm~\ref{alg:fdecsingle} is executed
 with an auxiliary qubit system~$R$
 and the initial state $
 \ket{\Psi_{\textrm{in}}}$ specified by~\eqref{eq:noisyclusterstatesingle}, i.e., the (half-encoded)
 Bell state~$\Phi_{\cC R}$ corrupted by the error~$E$.  
  Then the output
  of   Algorithm~\ref{alg:fdecsingle} on qubits~$q,R$ 
  is the Bell state~$\ket{\Phi_{(\alpha,\beta)}}$, where
  \begin{align}
    \alpha &= \ztwoinner{\{q\}}{\supp(E^Z)} \oplus \ztwoinner{\supp(E^Z) \oplus \mmatch(\partial \supp(E^Z))}{\cL_X} \qquad \textrm{ and }\\
    \beta  &= \ztwoinner{\{q\}}{\supp(E^X)} \oplus \ztwoinner{\supp(E^X) \oplus \mmatch^*(\partial \supp(E^X))}{\cL_Z^*}\ .
    \label{eq:Ldecodingsuccesscondition}
  \end{align}
\end{theorem}
\begin{proof}
We first show that the output is the Bell state~$\Phi_{(\alpha,\beta)}$ with
  \begin{align}
    \alpha &= \ztwoinner{\{q\}}{\supp(E^Z)} \oplus \ztwoinner{\supp(E^Z) \oplus m}{\cL_X} \qquad \textrm{ and }\\
    \beta  &= \ztwoinner{\{q\}}{\supp(E^X)} \oplus \ztwoinner{\supp(E^X) \oplus m^{*}}{\cL_Z^*} 
    \label{eq:toproveqexezmstar}
  \end{align}
  where $m\subset \Edec$ and $m^*\subset \Edecdual$ are the minimum matchings computed in Step~\ref{state:setmsingle} of the algorithm. 

  According to Lemma~\ref{lem:pmstatesingle}, the
  post-measurement state
  after Step~\ref{it:stepmeasurementsingle} 
  of Algorithm~\ref{alg:fdecsingle} is the Bell state~$\ket{\Phi_{(c_X,c_Z)}}$
  with
  \begin{align}
c_X&=  \ztwoinner{x}{\cL_X}\oplus \syn(S^X, E)\\
c_Z&=  \ztwoinner{z}{\cL_Z^*}\oplus  \syn(S^Z, E)\ .
    \end{align}
    Since the algorithm applies
    the correction operation $Z^{\hat{c}_X}X^{\hat{c}_Z}$
    in Step~\ref{state:estimatecorrectionsingle}, the output of the algorithm is  the Bell state~$\ket{\Psi_{(c_X\oplus\hat{c}_X,c_Z\oplus\hat{c}_Z)}}$. It follows with the definition of $(\hat{c}_X,\hat{c}_Z)$ in Step~\ref{state:estimateofcsingle},
    i.e.,
    \begin{align}
  \hat{c}_X &= \ztwoinner{x}{\cL_X} \oplus \hat{s}^X\\
  \hat{c}_Z&=\ztwoinner{z}{\cL_Z^*} \oplus \hat{s}^Z
    \end{align}
    that
    the output is the Bell state~$\ket{\Phi_{(\alpha,\beta)}}$ where
    \begin{align}
      \alpha &= \hat{s}^X \oplus \syn(S^X,E) \qquad \textrm{ and } \\
      \beta &= \hat{s}^Z \oplus \syn(S^Z,E)\ .
    \end{align}
    The claim~\eqref{eq:toproveqexezmstar} now follows by combining the expressions~\eqref{eq:synviaEglsingle}, \eqref{eq:synviaEgldualsingle} for $\left(\syn(S^X,E),\syn(S^Z,E)\right)$    from Lemma~\ref{lem:actualsyndsingle}      with the
 Definitions~\eqref{eq:hatsxszdefclaimsingle} of $(\hat{s}^X,\hat{s}^Z)$ used in  Step~\ref{state:sydnromebitsxszsingle} of the algorithm.

 The claim of the theorem now follows
 by combining~\eqref{eq:toproveqexezmstar} with the fact that
 \begin{align}
 m&= \mmatch(s)=\mmatch(\partial \supp(E^Z))\\
m^* &=\mmatch^*(s^{*})=\mmatch^*(\partial^* \supp(E^X))\ 
\end{align}
by the definition of $(m,m^*)$, and the expressions for $(s,s^*)$ established in  Lemma~\ref{lem:sisboundarysingle}.
\end{proof}

\subsection{Bounds on the resilience functions  for~$\Tdec$ and $\Tdecdual$}\label{sec:resiliencetdectdecdual}
Here we establish the following bounds on the resilience functions of interest:
\begin{lemma}\label{lem:upperboundresiliencetdec}
Consider the graph~$\Tdec$  with internal vertices $\Vdecint$ defined by~\eqref{eq:tdecinterior}, see Fig.~\ref{fig:Tdec}.
 Let $\cL_X\subset \Edec$ be the subset of edges on the vertical strip in Fig.~\ref{fig:Tdec} (cf.~\eqref{eq:cLXdefinition}). 
Then the resilience function~$\res_{\cL_X}$ satisfies 
\begin{align}
\res_{\cL_X}(p)&\leq 54p \label{eq:upperboundresiliencefunctiontdec}
\end{align}
for any $p\in \interval{0}{\frac{1}{144}}$.
\end{lemma}
\begin{proof}
Consider a simple closed loop $L\in Z_\circ(\Tdec)$.
Clearly, the definition of~$\cL_X$ implies that~$L$ and~$\cL_X$ have no overlap.
In particular, $\ztwoinner{L}{\cL_X}=0$, implying that closed loops do not contribute to the resilience function~$\res_{\cL_X}$ in this case.

To describe the set $\Zpathext(\Tdec)$ of simple paths with endpoints in~$\Vdecext$ through internal vertices, we label the external vertices in~$\Vdecext$ as
\begin{align}
\Vdecext&=\left\{u_1,\ldots,u_{d-1},v_1,\ldots,v_{d-1}\right\}\ ,
\end{align}
see Fig.~\ref{fig:Tdec}.
Explicitly, the external vertices are
\begin{align}
  u_j = (0, 2j), \quad v_j = (2, 2j) \qquad \textrm{ for } \qquad j = \{1, \dots, d-1\} \ .
\end{align}

Clearly, a simple path $P\in \Zpathext(\Tdec)$ with endpoints in $\Vdecext$ through internal vertices satisfies $\ztwoinner{P}{\cL_X}=1$ if and only if
$P$ starts at~$u_j$ (with $j\in \{1,\ldots,d-1\}$) and ends at~$v_k$ (with $k\in \{1,\ldots,d-1\}$).
For $j\in \{1,\ldots,d-1\}$ and $\ell\in \mathbb{N}$, let $\Delta(u_j,\ell)$ be the set of all such paths~$P \in \Zpathext(\Tdec)$ of length~$\ell$ that start at $u_j$ and end at some vertex in the set~$\{v_k\}_{k=1}^{d-1}$.
Then we have 
\begin{align}
\res_{\cL_X}(p)&=\sum_{\ell=1}^\infty \binom{\ell}{\lceil \ell/2\rceil}\cdot
 |\bigcup_{j=1}^{d-1}\Delta(u_j,\ell)|\cdot p^{\lceil \ell/2\rceil}\\
 &=\sum_{j=1}^{d-1}\sum_{\ell=L^{\min}_j}^{L^{\max}_j} \binom{\ell}{\lceil \ell/2\rceil}\cdot
 |\Delta(u_j,\ell)|\cdot p^{\lceil \ell/2\rceil}\ .\label{eq:deltaformofresLx}
\end{align}
because $\Delta(u_j,\ell)$ and $\Delta(u_{j'},\ell)$ are disjoint for~$j\neq j'$ since correspondings paths have distinct starting points~$u_j\neq u_{j'}$.
Here $L^{\min}_j$ and $L^{\max}_j$ are the minimal and maximal lengths of a simple path~$P \in \Zpathext(\Tdec)$ starting at $u_j$ and ending in~$\{v_k\}_{k=1}^{d-1}$ without touching any other external vertices.

The set $\Delta(u_1,1)$ consists of a single path, i.e., $|\Delta(u_1,1)|=1$. For $j\in \{2,\ldots,d-1\}$, we use that the graph~$T_{dec}$ has vertices of degree at most~$4$, which implies that
\begin{align}
|\Delta(u_j,\ell)|\leq 3^{\ell-1}\qquad\textrm{ for all }\qquad \ell\in\mathbb{N}\ .
\end{align}
 This is because any path~$P\in \Delta(u_j,\ell)$ has
the same first edge (connected to~$u_j$), and there are  at most $3$~choices for each of the remaining edges.)  Furthermore, it is clear that
\begin{align}
L_j^{\min}\geq j\ .
\end{align}
We conclude  (where for $\ell>1$, we use the inequalities $p^{\lceil \ell/2\rceil}\leq p^{\ell/2}$ and 
$\binom{\ell}{\lceil\ell/2 \rceil}\leq 2^\ell$) that

\begin{align}
\res_{\cL_X}(p)&\leq p+\sum_{j=2}^{d-1}\sum_{\ell=j}^{L_j^{\max}} 2^\ell 3^{\ell-1} p^{\ell/2}\\
&\leq p+\frac{1}{3}\sum_{j=2}^{d-1}\sum_{\ell=j}^{\infty} q^\ell\qquad\textrm{ where }\qquad q:=6\sqrt{p}\\
&=\frac{q^2}{6}+\frac{q^2-q^d}{3(1-q)^2} \\
&\leq p + \frac{4}{3}\left(q^2 - q^d\right) \\ 
&\leq p + \frac{4}{3} \cdot \left(6 \sqrt{p}\right)^2\\
&\leq 54p \ .
\end{align} 
Here we used the fact $q \leq 6 \sqrt{\frac{1}{144}} = \frac{1}{2}$ and that $\frac{1}{1-q} \leq 2$.
\end{proof}

A similar bound applies to the dual decoding graph~$\Tdecdual$.
\begin{lemma}\label{lem:upperboundresiliencetdecdual}
Consider the dual graph~$\Tdecdual$  with internal vertices $\Vdecdualint$ defined by~\eqref{eq:tdecdualinterior}, see Fig.~\ref{fig:Tdecdual}.
 Let $\cL_Z^{*} \subset \Edecdual$ be the subset of edges on the horizontal strip in Fig.~\ref{fig:Tdecdual} (cf.~\eqref{eq:cLZdefinition}). 
Then the resilience function~$\res_{\cL_Z^{*}}$ satisfies 
\begin{align}
\res_{\cL_Z^{*}}(p)&\leq 38p \label{eq:upperboundresiliencefunctiontdecdual}
\end{align}
for any $p \in \interval{0}{\frac{1}{144}}$.
\end{lemma}
The proof is analogous to that of Lemma~\ref{lem:upperboundresiliencetdec}. We include it in Appendix~\ref{app:upperboundresiliencetdecdual} for completeness.

\subsection{A bound on single-shot decoding under local stochastic noise}\label{sec:boundscdecoding}

Consider a distance-$d$ surface code of $\abs{\cC}=2d^2-2d+1$ qubits with $d\geq 2$ and logical operators~$\overline{X}=\prod_{j\in \cL_X\cup \{q\}}X_j$ and $\overline{Z}=\prod_{j\in\cL_Z\cup\{q\}}$.
Let $\ket{\overline{0}}$ and $\ket{\overline{1}}$ be code states to eigenvalues~$+1$ and $-1$~of $\overline{Z}$. Let
\begin{align}
\begin{matrix}
V:&\mathbb{C}^2 & \mapsto & \left( \mathbb{C}^{2} \right)^{\otimes \cC}\\
& \alpha\ket{0}+\beta\ket{1} & \mapsto & \alpha\ket{\overline{0}}+\beta\ket{\overline{1}}\ 
\end{matrix}
\end{align}
be an isometric encoding map, and let 
\begin{align}
\begin{matrix}
  \cE: & \cB(\mathbb{C}^2) & \rightarrow & \cB( \left( \mathbb{C}^{2} \right)^{ \otimes \cC} )\\
 & \rho & \mapsto & V \rho V^{\dagger}
\end{matrix}
\end{align}
be an encoding channel associated with~$V$.
For any state~$\rho\in\cB(\mathbb{C}^2)$, the state~$\cE(\rho)$  is the corresponding encoded (logical) state.  

 We  describe what single-shot decoding protocol achieves in terms of its logical action. For this purpose, let us define the  single-qubit 
random Pauli channel
\begin{align}
\cP_\nu(\rho):=\nu_0 \rho +\nu_1 X\rho X^\dagger +\nu_2 Y\rho Y^\dagger +\nu_3 Z\rho Z^\dagger\ .
\end{align}
for any  probability distribution~$\nu=(\nu_0,\nu_1,\nu_2,\nu_3)$, $\nu_j\geq 0$ and $\sum_{j=0}^3 \nu_j=1$. 

For $p\in [0,1]$, let
$E\sim \cN(p)$ be a local stochastic error with parameter~$p$. This defines a ``noise channel''~$\Lambda_p$ by
\begin{align}
\begin{matrix}
\Lambda_p:&\cB((\mathbb{C}^2)^{\otimes \cC})&\rightarrow &\cB((\mathbb{C}^2)^{\otimes \cC})\\
& \rho & \mapsto & \Lambda_p(\rho):=\sum_{E_0\textrm{ Pauli on }(\mathbb{C}^2)^{\otimes \cC}}\Pr\left[E=E_0\right]E_0\rho E_0^\dagger
\end{matrix}\label{eq:localstochasticnoisechanneldefb}
\end{align}
be the CPTP map that applies the error~$E$ to a state. Finally, let $\cD:\cB((\mathbb{C}^2)^{\otimes \cC})\rightarrow \cB((\mathbb{C}^2)^{\{q\}})$ be the CPTP map corresponding to an ideal execution of Algorithm~\ref{alg:fdecsingle}. Then the following holds:
\begin{theorem}[Single-shot decoding of a surface code]\label{thm:mainsingleshot}
Consider local stochastic noise as expressed by a CPTP map~$\Lambda_p$ of the form~\eqref{eq:localstochasticnoisechanneldefb}, with noise strength
\begin{align}
    p\leq \frac{1}{144}\ .
\end{align}
There is a probability distribution~$\nu=(\nu_0,\nu_1,\nu_2,\nu_3)$ over~$(I,X,Y,Z)$ such that the following holds:
\begin{enumerate}[(i)]
\item Suppose we run algorithm~\ref{alg:fdecsingle} on the noisy encoded  state~$\Lambda_p \circ \cE (\rho)$, where 
 $\rho\in \cB(\mathbb{C}^2)$ is an arbitrary single-qubit state.
Then the resulting output state $\cD\circ\Lambda_p\circ \cE(\rho)$ on qubit~$q$ is the state~$\cP_\nu(\rho)$. 
\item We have $1-\nu_0\leq 94p$.
\end{enumerate}
\end{theorem}
In particular, Theorem~\ref{thm:mainsingleshot} implies that applying Algorithm~\ref{alg:fdecsingle} 
to any encoded pure state~$\Lambda_p(V\ketbra{\Psi}V^{\dagger})$ corrupted by~$\Lambda_p$ yields an output state~$\rho_{\textrm{out}}=\cD\circ\Lambda_p\circ \cE(\proj{\Psi})$ with overlap at least
\begin{align}
\bra{\Psi}\rho_{\textrm{out}}\ket{\Psi}\geq \nu_0\geq 1- 94p
\end{align}
with the original state. 
\begin{proof}
Let us denote by~$L\cong\mathbb{C}^2$ the ``logical qubit'' to be encoded, and let  $R\cong\mathbb{C}^2$ be a reference system. We consider the action of the noise
and the decoding protocol when applied to the two-qubit Bell state~$\Phi_{RL}:=\Phi_{(0,0)}$, i.e., we compute the Choi-Jamiolkowski state
\begin{align}
J_{Rq}:=(\mathsf{id}_R\otimes \cD\circ\Lambda_p\circ \cE)(\Phi_{RL})\ .
\end{align} 
Observe that $\Phi'_{R\cC}:=(\mathsf{id}_R\otimes \cE)(\Phi_{RL})$ is the half-encoded Bell state by definition. It thus follows from Theorem~\ref{thm:successconditionsurfacecode} that 
\begin{align}
J_{Rq}:=(\mathsf{id}_R\otimes \cP_{\nu})(\proj{\Phi_{(0,0)}})\label{eq:Jrqexpression}
\end{align}
for a probability distribution~$\nu=(\nu_0,\nu_1,\nu_2,\nu_3)$ over $(I,X,Y,Z)$. In this expression, $\nu_0$ is the probability that the condition~\eqref{eq:Ldecodingsuccesscondition} is satisfied.  Because of the Choi-Jamiolkowski isomorphism, Eq.~\eqref{eq:Jrqexpression} shows that the channel~$\cD\circ\Lambda_p\circ\cE=\cP_{\nu}$ is a random Pauli channel. It thus remains to establish an upper bound on~$1-\nu_0$. 

Let us introduce the following four random variables:
\begin{align}
A_Z&:=\ztwoinner{\{q\}}{\supp(E^Z)}\\
B_Z&=\ztwoinner{\supp(E^Z) \oplus \mmatch(\partial \supp(E^Z))}{\cL_X}\\
A_X&:=\ztwoinner{\{q\}}{\supp(E^X)}\\
B_X&:=\ztwoinner{\supp(E^X) \oplus \mmatch^*(\partial \supp(E^X))}{\cL_Z^*}\ .
\end{align}
Then~\eqref{eq:Ldecodingsuccesscondition} translates to
\begin{align}
\nu_0&=\Pr\left[(A_Z,B_Z)\in \{(0,0),(1,1)\}\textrm{ and }(A_X,B_X)\in \{(0,0),(1,1)\}\right]\ .
\end{align}
The union bound gives
\begin{align}
1-\nu_0 &\leq \Pr\left[(A_Z,B_Z)\in \{(0,1),(1,0)\}\right]+\Pr\left[(A_X,B_X)\in \{(0,1),(1,0)\}\right]\ .
\end{align}
Again using the union bound, we have 
\begin{align}
\Pr\left[(A_Z,B_Z)\in \{(0,1),(1,0)\}\right]&\leq \Pr\left[(A_Z,B_Z)=(0,1)\right]+\Pr\left[(A_Z,B_Z)=(1,0)\right]\\
&\leq \Pr\left[B_Z=1\right]+\Pr\left[A_Z=1\right]\ .
\end{align}
Applying the same reasoning to~$(A_X,B_X)$, we conclude that 
\begin{align}
1-\nu_0 & \leq\Pr\left[A_X=1\right] +\Pr\left[A_Z=1\right]+\Pr\left[B_X=1\right]+\Pr\left[B_Z=1\right]\ .
\end{align}
By definition, we have 
\begin{align}
\Pr\left[A_Z=1\right]&=\Pr\left[\ztwoinner{\{q\}}{\supp(E^Z)}=1\right]\\
&\leq p
\end{align}
where we used that~$E^Z\sim\cN(p)$ is local stochastic noise with the same parameter~$p$ as~$E$. By similar reasoning, we have 
\begin{align}
\Pr\left[A_X=1\right]&\leq  p\ .
\end{align}

Furthermore,  we have 
\begin{align}
\Pr\left[B_Z=1\right]&=\Pr\left[\ztwoinner{\supp(E^Z) \oplus \mmatch(\partial \supp(E^Z))}{\cL_X}=1\right]\\
&\leq \mathsf{res}_{\cL_X}(p)\\
&\leq 54p 
\end{align}
where we used the definition of~$A_Z$, the fact that $E^Z\sim\cN(p)$ and Proposition~\ref{prop:maincombinatorics}, as well as the upper bound~\eqref{eq:upperboundresiliencefunctiontdec}.  Similarly, we have 
\begin{align}
\Pr\left[B_X=1\right]&\leq \mathsf{res}_{\cL_Z^{*}}(p)\\
&\leq 38 p \ .
\end{align}
Combining these inequalities, we obtain
\begin{align}
1-\nu_0& \leq  p + p + 38p + 54p \\ 
& \leq 94p   
\end{align}
as claimed.
\end{proof}

\section{Entanglement generation from a noisy 3D cluster state\label{sec:entanglementgenerationnoisy}}
In the following, we describe a protocol for generating a two-qubit Bell state in a noise-robust fashion. It is based on the cluster state~$W\ket{0^{\cC}}$ on a 3D~lattice~$\cC$, a state which can be generated by a depth-$6$ circuit~$W$.  This state is measured using single-qubit measurements on all but two qubits. Based on the measurement outcomes, a certain correction operation is applied to one of the remaining two qubits.

In Section~\ref{sec:clusterstate} we discuss the relevant definitions underlying the cluster state. In Section~\ref{sec:measurementpatternclusterstate}
 we describe the single-qubit measurement pattern used in the protocol. Section~\ref{sec:subgraphsm} introduces certain subgraphs of the cluster state lattice and additional data such as  distinguished subsets of qubits having to do with the encoded logical information. These  are relevant  for the definition of certain decoding functions used in the protocol.
 
 Section~\ref{sec:stabilizersgraphstate} identifies certain stabilizers of the graph state that are key to the construction and analysis of the entanglement generation protocol. We give a complete description of the protocol in Section~\ref{sec:alg}. Finally, in Section~\ref{sec:alganalysiscorrectness}, we compute the output state of this algorithm when it is run on a noisy cluster state~$EW\ket{0^{\cC}}$ which is corrupted by a Pauli error~$E$. This will be used to analyze the success probability for local stochastic errors in Section~\ref{sec:resiliencenoiselocalstochastic}.

\subsection{Cluster state: Definition\label{sec:clusterstate}}
Our construction is based on a 3D cluster state. To introduce this state, we mostly follow the conventions of~\cite{bravyiQuantumAdvantageNoisy2020} but
use different conventions to describe so-called  dangling edges. Contrary to Ref.~\cite{bravyiQuantumAdvantageNoisy2020} we will consider their  degree-$1$ endpoints as regular  vertices.

Let $d \geq 2$ and let $R\geq 3$ be an odd integer. Consider the set
\begin{align}
\Cprime[d\times d \times R] := \{(u_1, u_2, u_3) \in \mathbb{Z}^3 \mid 0 \leq u_1 \leq 2d,~ 0 \leq u_2 \leq 2d-2,~ 1 \leq u_3 \leq R\}
\end{align} 
(From the perspective of coding, the integers~$d$ and $R$ are associated with the distance~$d$ of a surface code, and a ``cluster state distance''~$\frac{R+1}{2}$, respectively.) In the following, we simply write~$\Cprime = \Cprime[d\times d \times R]$. 
Elements $(u_1,u_2,u_3)\in \Cprime$ will be referred to as sites.

Qubits are located at each site belonging to the subset
\begin{align}\label{eq:defcCddR}
\cC[d \times d \times R]  := \{u \in \Ctilde \mid 1 \leq u_1 \leq 2d-1\} \setminus \left\{ (o,o,o), (e,e,e) \right\} \subset \Ctilde
\end{align}
where each $e~(o)$ refers to any even (odd) integer. The set $\cC=\cC = \cC[d \times d \times R]$ is called the cluster state lattice of the form $d \times d \times R$.  We will refer to a site~$u\in \cC$ interchangeably as site~$u$ or qubit~$u$.
 In the following, we use the convention that e.g., 
 $\{(e,e,R)\}\subset\Cprime$  denotes the set of all $(u_1,u_2,u_3)\in \Cprime$ with $u_1,u_2$ even and $u_3=R$.

Given a site~$u\in \Cprime$, the set of nearest  neighbors of~$u$ is defined as 
\begin{align}
\neigh(u) = \left\{v \in \mathcal{C} \mid 
d(u,v)=1 \right\}\qquad\textrm{ where }\qquad d(u,v):=\sum_{i=1}^3 \abs{u_i - v_i}\ ,\label{eq:nearestneighborset}
\end{align}
i.e., the  nearest neighbors of~$u$ are associated with qubits at Manhattan distance~$1$ from~$u$. Because elements of~$\Cprime \setminus \cC$ are excluded in~\eqref{eq:nearestneighborset}, any  qubit~$u\in\cC$ has at most four nearest neighbor qubits.

The cluster state of interest is associated with the graph whose vertex set is~$\cC$, and where
the set of edges consists of pairs $\{u,v\}$ with $v\in \neigh(u)$, see Fig.~\ref{fig:3Dclusterstate} for an illustration of this graph.
It is defined as 
\begin{align}
  W\ket{0^\cC}\ ,
\end{align} where
\begin{align}
  W = H(\cC)\left( \prod_{\left\{ u,v \right\}} CZ_{u,v}\right) H(\cC)\ ,
\end{align}
with the product taken over all edges of the graph, i.e., all pairs $u,v\in\cC$ such that $v\in \neigh(u)$.
Here $CZ_{u,v}$ denotes a controlled-$Z$ gate on qubits~$u$ and $v$. By definition, the unitary~$W$ can be implemented using a depth-$6$ circuit with nearest-neighbor gates. (Note that
our definition of the graph state differs from the more commonly used definition by an additional layer of Hadamard-gates.) The state~$W\ket{0^{\cC}}$ is a stabilizer state with stabilizer group generated by
\begin{align}
  G_u = Z_u \prod_{v \in \neigh(u)} X_v \qquad \textrm{ with } \qquad u \in \mathcal{C}\ .\label{eq:stabilizergeneratorsgraph}
\end{align}

\begin{figure}[h]
  \centering
  \includegraphics[width=.7\linewidth]{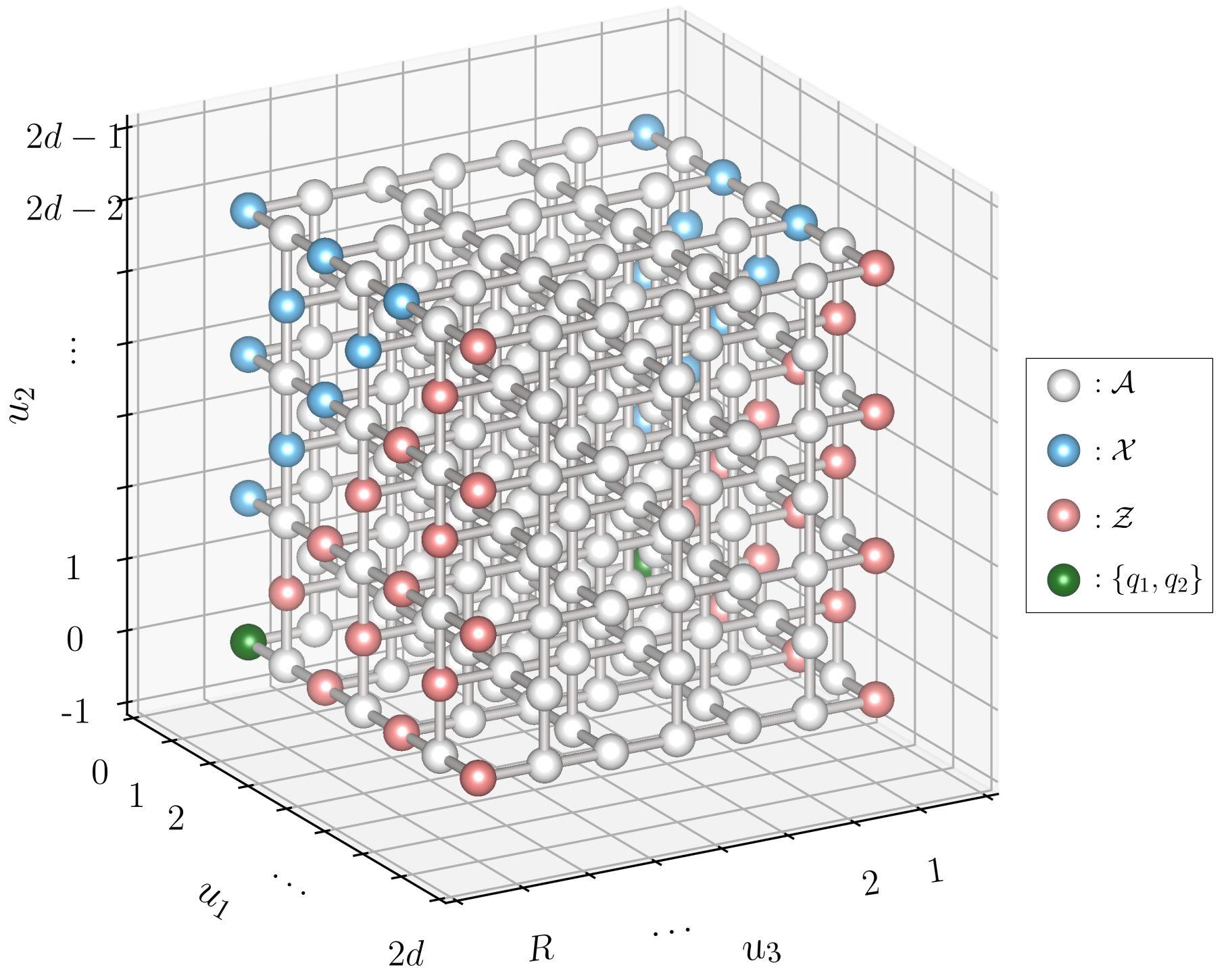}
  \caption{The graph associated with the three-dimensional cluster state $W\ket{0^{\cC}}$. Each qubit is located at a vertex (sphere) in $\cC$. Measurement pattern used in protocol is illustrated with colors.}
  \label{fig:3Dclusterstate}
\end{figure}

As in the case of surface codes, it will also need a dual description. For this, let $\Cprimedual= \Ctildedual[d \times d \times R]$ be the set of dual sites defined as
\begin{align}
\Ctildedual[d \times d \times R]:= \left\{(u_1,u_2,u_3)\in\mathbb{Z}^3\ |\ 1\leq u_1\leq 2d-1\textrm{ and }-1\leq u_2\leq 2d-1\textrm{ and }1\leq u_3\leq R \right\}\ .
\end{align}
Then we can equivalently express the location of qubits as 
\begin{align}
  \cC^{*} = \{u \in \Ctildedual \mid 0 \leq u_2 \leq r-1\} \setminus \{(o,o,o),(e,e,e)\} \subset \Ctildedual \ ,
\end{align}
i.e., we have $\cC = \cC^{*}$.

\subsection{Measurement pattern used in protocol}\label{sec:measurementpatternclusterstate}
In this section, we describe the measurement pattern used in our protocol. We partition the set of qubits as~$\mathcal{C}=\cA\cup \cB$ where
\begin{align}
  \mathcal{B}  &= \{(e,o,1),(o,e,1), (e,o,R), (o,e,R) \in \mathcal{C}\}, \\
  \mathcal{A}  &= \mathcal{C} \setminus \mathcal{B}.
\end{align}
The set~$\mathcal{B}$ is the union of~$\{(e,o,1),(o,e,1)\}$, i.e., a subset of all qubits at the ``front'' layer (i.e., with $u_3=1$) of~$\cC$, and ~$\{(e,o,R),(o,e,R)\}$, a subset of all qubits at the ``back'' layer (i.e., with $u_3=R$) of~$\cC$. In prior work~\cite{raussendorfLongrangeQuantumEntanglement2005,bravyiQuantumAdvantageNoisy2020},
no measurements were applied to these qubits; instead, long-range encoded entanglement between two surface codes of distance~$d$ associated with these front- and back-layer qubits was established.

The remaining qubits~$\{(e,e,1)\}$ in the front layer are omitted from~$\cB$ but belong to~$\cA$. Similarly, qubits~$\{(e,e,R)\}$ in the back layer belong to~$\cA$. These are auxiliary qubits whose measurement (specified below) reveals information about the associated surface code stabilizers in the corresponding layer. Similarly, the measurement results associated with the remaining qubits of~$\cA$ belonging to the ``bulk'' reveal information about the logical error of the encoded surface code state (or, in our scheme, the final two-qubit entangled state). 

As in~\cite{raussendorfLongrangeQuantumEntanglement2005,bravyiQuantumAdvantageNoisy2020}, all qubits belonging to~$\cA$ are measured in the $Z$-basis. Because our goal is to establish entanglement between two (physical) qubits, we additionally measure all qubits in~$\cB$, except for two qubits
\begin{align}
q_1&:=(1,0,1)\\
q_2&:=(1,0,R)
\end{align} belonging to the front and back layers, respectively.  Our protocol generates a Bell state between these two qubits.

The measurement pattern for the front layer qubits of~$\cB$ mirrors the single-shot decoding procedure for the surface code: Qubits belonging to the upper left triangular region are measured in the~$X$-basis, whereas qubits situated inside the lower right triangular region are measured in the~$Z$-basis. Qubits in the back layer of~$\cB$ are measured with an analogous pattern.  Formally, the sets of qubits in $\cB \setminus \{q_1, q_2\}$ measured in the~$X$- respectively~$Z$-basis are
\begin{align}
  \cX &= \{(u_1, u_2, u_3) \in \mathcal{B} \setminus \{q_1,q_2\} \mid u_2 \geq 1 + u_1\}\\
  \cZ &= \{(u_1, u_2, u_3) \in \mathcal{B} \setminus \{q_1,q_2\} \mid u_2 < 1 + u_1\}\ .
\end{align}
The measurement outcome will be denoted by $(a,x,z)$, where
\begin{align}
    a &\in \{0,1\}^{\cA} \\
    x &\in \{0,1\}^{\Xcl} \\
    z &\in \{0,1\}^{\Zcl} \ 
\end{align}
are the measurement results associated with the qubits belonging to~$\cA$, $\Xcl$ and $\Zcl$, respectively. 
This completes the description of the measurements used in the protocol, see Fig.~\ref{fig:3Dclusterstate}.

\subsection{Subgraphs, recovery graphs and internal vertices~\label{sec:subgraphsm}}
In this section, we introduce various subgraph of the cluster state lattice. We also
 define certain subsets of qubits that are relevant for our entanglement generation protocol.

Specifically, in Section~\ref{sec:gl}, 
we introduce the ``glued graph'' ~$\Tgl$,
as well as the ``dual glued graph''~$\Tgldual$. They are obtained by gluing together surface code graphs respectively dual surface code graphs with the cluster state lattice. These graphs are at the basis of our entanglement generation protocol, and clarify the relationship to the single-shot decoding procedure for the surface code discussed in Section~\ref{sec:singleshotdecodingprotocol}.

In Section~\ref{sec:decodinggraphs}, we then introduce the decoding graph~$\Tcldec$ and its  dual counterpart~$\Tcldecdual$. In Section~\ref{sec:internalverticestcl}, we additionally define certain subsets of their edges, the so-called recovery sets. We also identify the associated internal vertices. Our entanglement generation protocol involves  solving corresponding matching problems.

\subsubsection{Glued graphs}\label{sec:gl}
We introduce certain graphs in~$\Ctilde$ and their dual in~$\Ctildedual$ as discussed in~\cite{bravyiQuantumAdvantageNoisy2020}.
Those graphs are  subgraphs or supergraphs of certain graphs that will be introduced in Section~\ref{sec:decodinggraphs}.

We first define an ``even graph'' $\Teven = (\Veven, \Eeven)$ in $\Ctilde$.
The set~$\Veven = \Vcirceven \cup \Vdangeven$ of vertices has two disjoint subsets
\begin{align}
  \Vcirceven &= \{u \in \{(e,e,e)\} \mid 2 \leq u_2 \leq 2d - 2\} \subset \Ctilde \qquad \textrm{ and } \\
  \Vdangeven &=\{u \in \left\{ (e,e,1), (e,e,R) \right\} \mid u_1 \not \in \{0, 2d\}\} \cup \{(0,e,e)\}\cup \{(2d,e,e)\} \subset \Ctilde \ .
\end{align}
Here we avoid the terminology of dangling edges used in~\cite{bravyiQuantumAdvantageNoisy2020} and instead collect endpoints of such edges using the vertex set~$\Vdangeven$.
We note that~$\Vdangeven$ consists of vertices~$u = (u_1,u_2,u_3)$ in the ``front'' ($u_3 = 1$), ``back'' ($u_3=R$), ``left'' ($u_1=0$), or ``right'' ($u_1=2d$) boundary.
Any pair $u,v\in \Veven$ of vertices at distance~$d(u,v)=2$ is connected by an edge in~$\Teven$, except when $u,v$ both belong to~$\Vdangeven$.
Other edges in $\Teven$ are of the form~$\{u,v\}$ with 
\begin{align}\label{eq:distanceoneedgesinTeven}
u &= (u_1, u_2, u_3) \in \{(e,e,1), (e,e,R)\} \qquad \textrm{ and }\qquad 
   v= \begin{cases*}
    (u_1,u_2,u_3 + 1) & \textrm{ if $u_3 = 1$}\\
    (u_1,u_2,u_3 - 1) & \textrm{ if $u_3 = R$} \ .
  \end{cases*} 
\end{align}
It is easy to check that these edges constitute all pairs of vertices~$u, v \in \Veven$ satisfying $d(u,v)=1$.
In summary, the set of edges of~$\Teven$ is
\begin{align}
  \Eeven &=\left\{\{u,v\} \in \Veven^2 \mid d(u,v)=2\textrm{ and } (u,v)\not\in \Vdangeven\times \Vdangeven\right\} \\
  & \quad \cup \left\{ \{u,v\} \in \Veven^2 \mid d(u,v) = 1 \right\} \ .
\end{align}

The ``odd graph''~$\Todd=(\Vodd, \Eodd)$ in~$\Ctildedual$ is a dual version of~$\Teven$. The set~$\Vodd = \Vcircodd \cup \Vdangodd$ of vertices consists of  two disjoint subsets
\begin{align}
  \Vcircodd &= \{u \in \{(o,o,o)\} \mid u_2 \not \in \{-1, r\}, u_3 \not \in \{1, R\}\} \subset \Ctildedual \qquad \textrm{ and } \\
  \Vdangodd &=\{u \in \left\{ (o,o,1), (o,o,R) \right\} \mid u_2 \not \in \{-1, r\}\} \cup \{u \in \{(o,-1,o), (o, r,o)\} \mid u_3 \not \in \{1, R\}\} \subset \Ctildedual \ .
\end{align}
We note that $\Vdangodd$ consists of vertices~$u = (u_1, u_2, u_3)$ in the ``front'' ($u_3 = 1$), ``back'' ($u_3=R$), ``top'' ($u_2=2d$), or ``bottom'' ($u_2=-1$) boundary.
The set of edges of $\Todd$ is
\begin{align}
  \Eodd = \{\{u,v\} \in \Vodd^2 \mid d(u,v) = 2 \textrm{ and } (u, v) \not \in \Vdangodd \times \Vdangodd\} \ .
\end{align}

The surface code graph~$\Tclsc$ has vertices belonging to the set of primal sites~$\Cprime$.
It has two connected components lying in the planes~$u_3=1$ and $u_3=R$. Each connected component is isomorphic to the surface code graph~$\Tsc=(\Vsc,\Esc)$ introduced  in Section~\ref{sec:surfacecodelattice}.
Let us simply write
\begin{align}
\Tclsc &=\Tsc\times \{1,R\}\ 
\end{align}
for this graph to emphasize this fact. In more detail,  we have
\begin{align}
  \Vclsc &=\left(\bigcup_{(u_1,u_2)\in \Vsc} \{(u_1,u_2,1)\}\right)\cup \left(\bigcup_{(u_1,u_2)\in \Vsc} \{(u_1,u_2,R)\}\right)\\
  \Eclsc&=
          \left(\bigcup_{
          {\tiny \begin{matrix}
              \{(u_1, u_2), (v_1, v_2)\}\in \Esc \end{matrix}}}
          \Bigl\{
          \left\{(u_1,u_2,1),(v_1,v_2,1)\right\}
          \Bigr\}
          \right)
          \cup
          \left(\bigcup_{
          {\tiny \begin{matrix}
              \{(u_1, u_2), (v_1, v_2)\}\in \Esc \end{matrix}}}
          \Bigl\{
          \left\{(u_1,u_2,R),(v_1,v_2,R)\right\}
          \Bigr\}
          \right) \ .
\end{align}

The graph~$\Tclscdual$ similarly consists of two copies of the dual surface code graph~$\Tsc$ located in the planes $u_3=1$ and $u_3=R$, respectively, which we express as 
\begin{align}
\Tclscdual &=\Tscdual\times \{1,R\}\ .
\end{align}
Written out, this means that
\begin{align}
 \Vclscdual &=\left(\bigcup_{(u_1,u_2)\in \Vscdual} \{(u_1,u_2,1)\}\right)\cup \left(\bigcup_{(u_1,u_2)\in \Vscdual} \{(u_1,u_2,R)\}\right)\\
 \Eclscdual&=
 \left(\bigcup_{
{\tiny \begin{matrix}
\{(u_1,u_2),(v_1,v_2)\}\in \Escdual\end{matrix}}}
\left\{(u_1,u_2,1),(v_1,v_2,1)\right\}\right)
\cup
 \left(\bigcup_{
{\tiny \begin{matrix}
\{(u_1,u_2), (v_1,v_2) \}\in \Escdual\end{matrix}}}
\left\{(u_1,u_2,R),(v_1,v_2,R)\right\}\right)\ .
 \end{align}

Now we are able to introduce a certain ``glued graph''~$\Tgl=(\Vgl,\Egl)$,
as well as a ``dual glued graph''~$\Tgldual=(\Vgldual,\Egldual)$, in terms of the subgraphs~$\Teven$, $\Tclsc$, $\Todd$, and $\Tclscdual$, see Fig.~\ref{fig:gluedgraphsdefinitioncluster}.

\begin{figure}
  \begin{subfigure}[b]{0.45\textwidth}
    \centering
    \includegraphics[width=.9\linewidth]{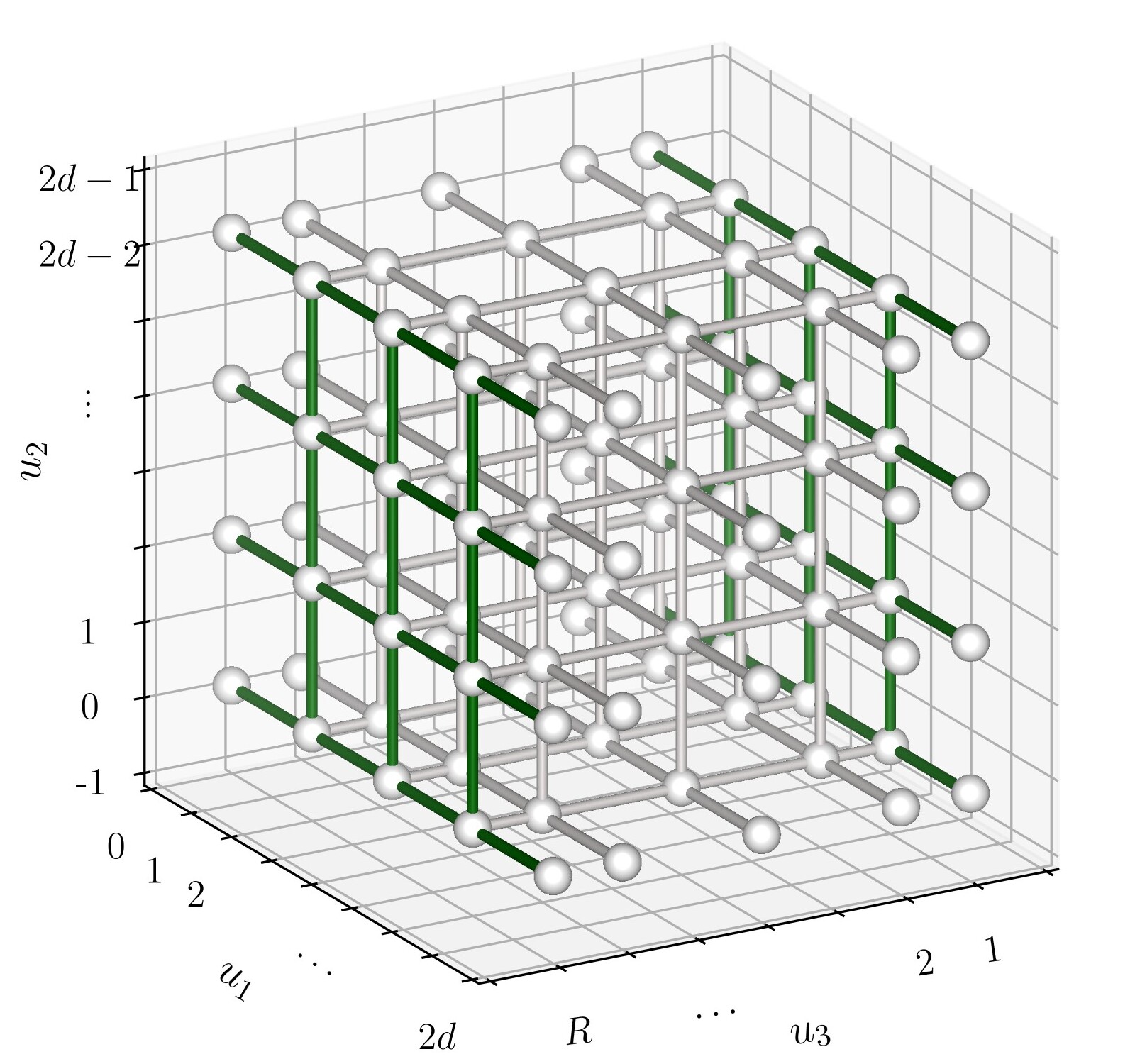}
    \caption{The glued graph $\Tgl$ as a union of two subgraphs: $\Teven$ (gray) and $\Tclsc$ (green).}\label{fig:gluedgraph}
  \end{subfigure}
  \quad
  \begin{subfigure}[b]{0.45\textwidth}
    \centering
    \includegraphics[width=.9\linewidth]{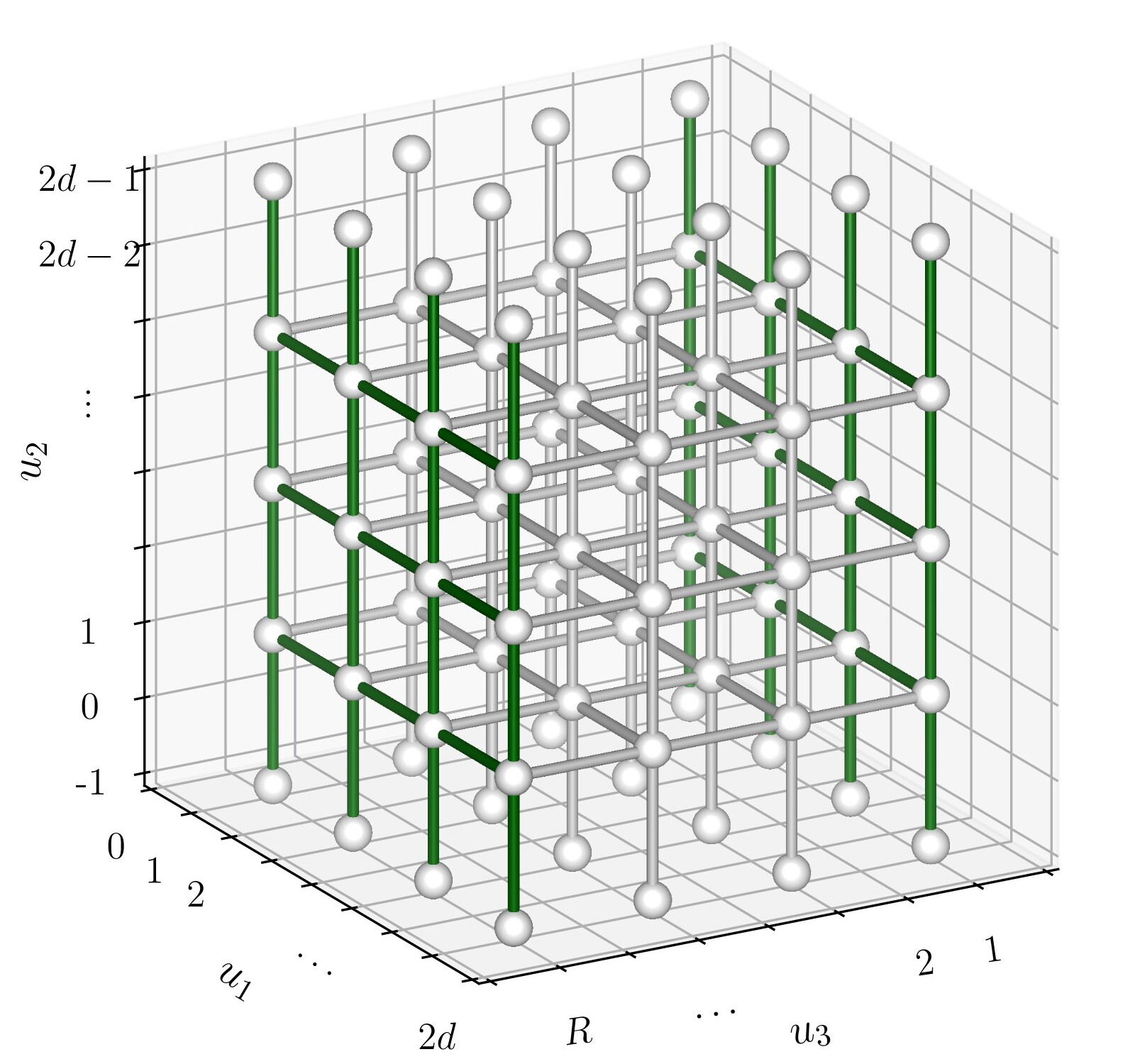}
    \caption{The dual glued graph $\Tgldual$ as a union of two subgraphs: $\Todd$ (gray) and $\Tclscdual$ (green).}\label{fig:dualgluedgraph}
  \end{subfigure}
  \caption{
    The glued graph~$\Tgl = \Teven \cup \Tclsc$ and the dual glued graph~$\Tgldual = \Todd \cup \Tclscdual$.}\label{fig:gluedgraphsdefinitioncluster}
\end{figure}

The glued graph~$\Tgl$ is obtained by  attaching the surface graph~$\Tclsc$ to the even graph~$\Teven$; here the two copies of the surface code
end up being connected to the rough front- and back boundaries of the even graph~$\Teven$. The graph~$\Tgl$ is formally defined by setting
\begin{align}
  \Vgl &=\Veven\cup \Vclsc\\
  \Egl &=\Eeven\cup \Eclsc\ .
\end{align}
We denote this construction simply by
\begin{align}
\Tgl&=\Teven\cup \Tclsc\ .
\end{align}
The  ``dual glued graph''~$\Tgldual=(\Vgldual,\Egldual)$ is defined by
gluing the dual surface graph~$\Tclsc$ to  the odd graph~$\Todd$, i.e.,
\begin{align}
\Tgldual &=\Todd\cup \Tclscdual\ .
  \end{align}

\subsubsection{Decoding subgraphs}\label{sec:decodinggraphs}
In the following, we describe how to process the measurement results in order to determine a suitable correction operation on the post-measurement on qubits~$\{q_1,q_2\}$ to obtain the Bell state~$\Phi$. 
More precisely, the correction is a Pauli operator of the form $Z^{\hat{c}_X}X^{\hat{c}_Z}$, where $\hat{c}_X$ and $\hat{c}_Z$ are computed from the single-qubit measurements results by a classical computation we describe below.
The computation of~$\hat{c}_X$  requires finding a minimal matching of a subset of marked vertices on a certain ``decoding'' graph~$\Tcldec$ with interior and exterior vertices.
The computation of~$\hat{c}_Z$ proceeds in a similar ``dual'' manner and involves a ``dual decoding graph''~$\Tcldecdual$, see Fig.~\ref{fig:internalexternalverticescl} for an illustration of these graphs.

The graphs~$\Tcldec$ and~$\Tcldecdual$ play a central role in our protocol:  We will  relate the edges of the  graphs~$\Tcldec$ and~$\Tcldecdual$ with the measured qubits. We will argue how the measurement results  determine certain subsets of edges in these two graphs. These subsets are used  to obtain certain subsets of marked vertices (by taking  boundaries). The pair ~$(\hat{c}_X,\hat{c}_Z)$ is computed from these subsets of marked vertices. 

\begin{figure}
  \centering
  \begin{subfigure}[b]{0.45\textwidth}
    \centering
    \includegraphics[width=0.9\textwidth]{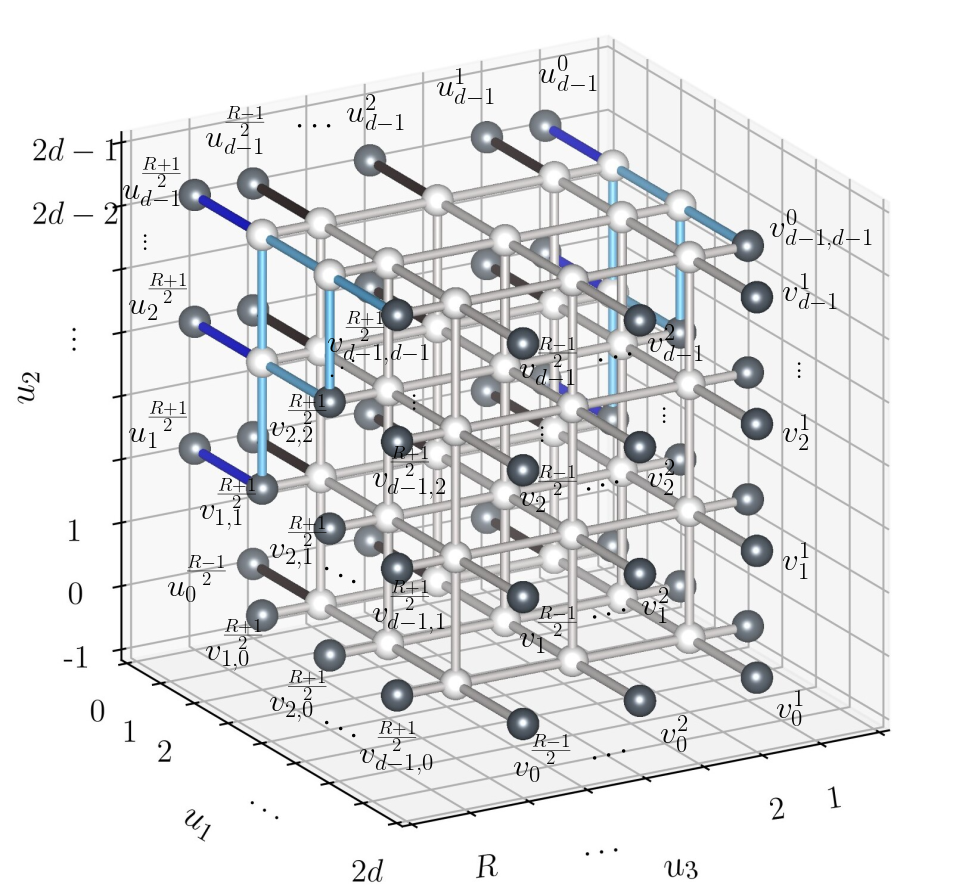}    
    \caption{
    \begin{tabular}{ |c||c|c| } 
     \hline
     \multicolumn{3}{|c|}{ $\Tcldec=\left(\Vcldec, \Ecldec\right)$} \\
     \hline\hline
     subsets of $\Ecldec$ & color & measurement\\ 
     \hline\hline
     $\cLXcl \cap \cA$ &\textcolor{bulklogical}{ \rule{1.2cm}{0.3cm}} & $Z$-basis \\
     \hline 
     $\left(\Ecldec \setminus \cLXcl \right) \cap \cA$&\textcolor{bulk}{ \rule{1.2cm}{0.3cm}} & $Z$-basis \\
     \hline 
     $\cLXcl \cap \cX$ &\textcolor{blue}{ \rule{1.2cm}{0.3cm}} & $X$-basis \\
     \hline 
     $\left(\Ecldec \setminus \cLXcl \right) \cap \cX$ &\textcolor{lightblue}{ \rule{1.2cm}{0.3cm}} & $X$-basis \\
     \hline 
    \end{tabular}
    }\label{fig:Tdecclexternal}
  \end{subfigure}
  \quad
  \begin{subfigure}[b]{0.45\textwidth}
    \centering
    \includegraphics[width=.9\linewidth]{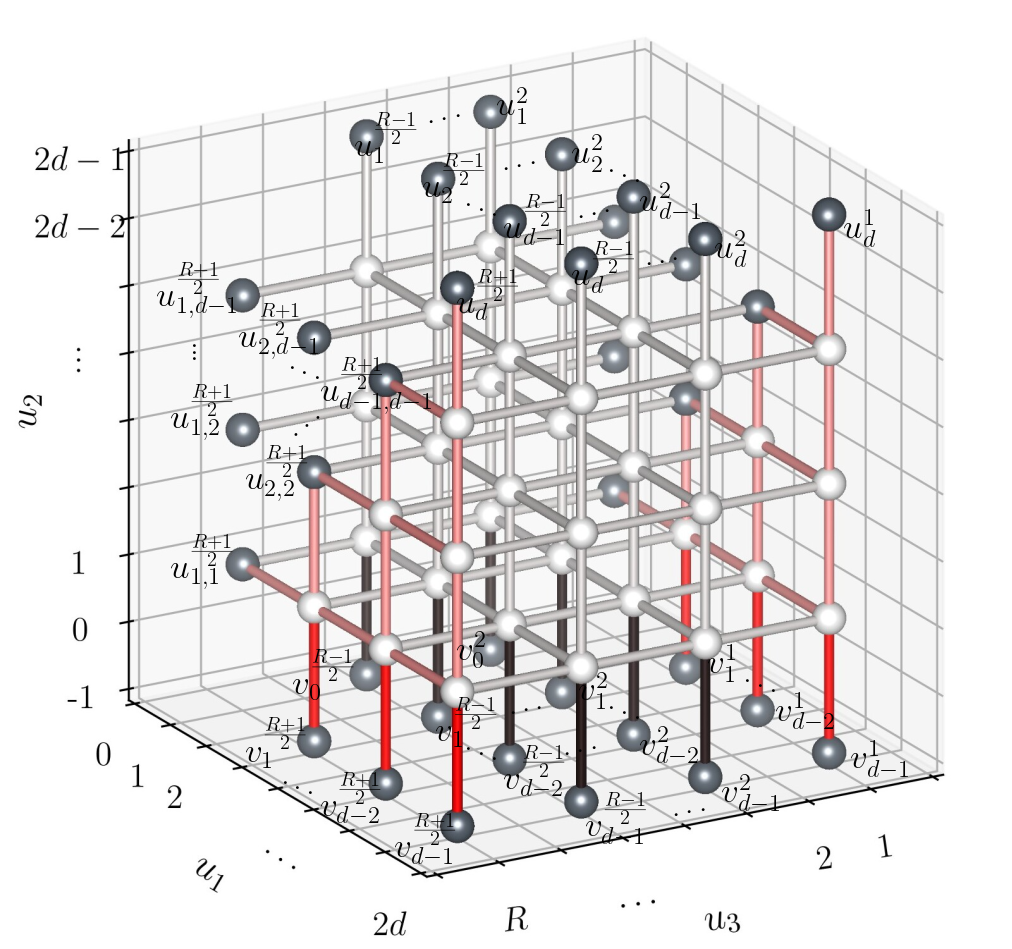}    
    \caption{
    \begin{tabular}{ |c||c|c| } 
     \hline
     \multicolumn{3}{|c|}{ $\Tcldecdual=\left(\Vcldec, \Ecldec\right)$} \\
     \hline\hline
     subsets of $\Ecldecdual$ & color & measurement\\ 
     \hline\hline
     $\cLZcl \cap \cA$ &\textcolor{bulklogical}{ \rule{1.2cm}{0.3cm}} & $Z$-basis \\
     \hline 
     $(\Ecldecdual \setminus \cLZcl)\cap \cA$ &\textcolor{bulk}{ \rule{1.2cm}{0.3cm}} & $Z$-basis \\
     \hline 
     $\cLZcl \cap \cZ$ &\textcolor{red}{ \rule{1.2cm}{0.3cm}} & $Z$-basis \\
     \hline 
     $(\Ecldecdual \setminus \cLZcl) \cap \cZ$ &\textcolor{lightred}{ \rule{1.2cm}{0.3cm}} & $Z$-basis \\
     \hline 
    \end{tabular}
    }\label{fig:Tdeccldualexternal}
  \end{subfigure}
  \caption{
    The decoding graph~$\Tcldec$  as well as the dual decoding graph~$\Tcldecdual$. External vertices~ $\Vcldecext$ and $\Vcldecdualext$ (black spheres) and internal vertices $\Vcldecint$ and $\Vcldecdualint$ (white spheres) are illustrated.
    The tables describe in which basis qubits corresponding to the edges are measured.
  }\label{fig:internalexternalverticescl}
\end{figure}

The decoding graph~$\Tcldec=(\Vcldec,\Ecldec)$ is the result of gluing two copies of the subgraph~$\Tdec\subset \Tsc$  (instead of two copies of the entire graph~$\Tsc$) to $\Teven$, see Fig.~\ref{fig:Tdecclexternal}. 
By construction, the resulting decoding graph~$\Tcldec$ is a subgraph of~$\Tgl$. Concisely, we have 
\begin{align}
  \Tcldec &= \Teven\cup \left(\Tdec\times  \{1,R\}\right)\subset \Tgl\ .\label{eq:tcldecdef}
\end{align}

The dual decoding graph~$\Tcldecdual=(\Vcldecdual,\Ecldecdual)$ is a subgraph of the dual glued graph~$\Tgldual$ and is  defined in a similar manner by
\begin{align}
  \Tcldecdual &= \Todd\cup \left(\Tdecdual\times  \{1,R\}\right)\subset \Tgldual\ ,\label{eq:tcldecdualdef}
\end{align}
see Fig.~\ref{fig:Tdeccldualexternal}. 

We explain the relationship between the qubits that are measured  and the edges of the graphs~$\Tcldec$ and~$\Tcldecdual$ using a bijection 
\begin{align}
\Lambda : \Ecldec \cup \Ecldecdual \to \cC \setminus \{q_1, q_2\} \ .
\end{align}
called the ``labeling'' map. The map~$\Lambda$ identifies  each edge of the graphs $\Tcldec$ and $\Tcldecdual$ with  the site of a qubit located at the midpoint of the edge or at one of the vertices adjacent to the edge: Let~$e=\{u,v\}\in \Ecldec \cup \Ecldecdual $ be one of edges of the graphs $\Tcldec$ and $\Tcldecdual$.
If~$d(u,v) = 2$, then~$\Lambda$ maps the edge~$e$ to the midpoint~$\Lambda(e)=\frac{u+v}{2}$  of the vertices~$u$ and~$v$.
Otherwise, i.e., if $d(u,v)\neq 2$, then we must have~$d(u,v)=1$ and it follows that $e$ belongs to~$\Eeven \subset \Ecldec$, and is of the form
\begin{align}\label{eq:shortedgesintcldec}
\{u, v\} &= \begin{cases*}
 \{(u_1,u_2,1),(u_1,u_2,2)\} & \textrm{ if $\{u_3,v_3\} = \{1,2\}$}\\
 \{(u_1,u_2,R-1),(u_1,u_2,R)\} & \textrm{ if $\{u_3,v_3\} = \{R-1,R\}$} 
\end{cases*}& \textrm{ with } \qquad u_1, u_2 \textrm{ even }
\end{align}
as stated in~\eqref{eq:distanceoneedgesinTeven} of Section~\ref{sec:gl}. 
In this case, the edge~$e$ is mapped by~$\Lambda$ to the site~$(u_1,u_2,1)$ if~$\{u_3,v_3\}=\{1,2\}$ or the site~$(u_1,u_2,R)$ otherwise.
In summary, we can write the map~$\Lambda$ as
\begin{equation}
  \begin{aligned}\label{eq:defoflabeling}
     \Lambda(\{u,v\})= \begin{cases*}
      (u_1, u_2, 1) & \textrm{ if $\{u,v\} \in \Ecldec$, $d(u, v) = 1$ and $\{u_3, v_3\} = \{1, 2\}$} \\
      (u_1, u_2, R) & \textrm{ if $\{u,v\} \in \Ecldec$, $d(u, v) = 1$ and $\{u_3, v_3\} = \{R-1, R\}$} \\
      \left( \frac{u_1+v_1}{2}, \frac{u_2+v_2}{2}, \frac{u_3+v_3}{2} \right) & \textrm{ otherwise.}
    \end{cases*}
  \end{aligned}
\end{equation}
It is easy to check that the map~$\Lambda$ is a bijection.

Using the labeling map~$\Lambda$, we can identify each edge of~$\Tcldec$ or~$\Tcldecdual$ with one of  the measured qubits, or equivalently, the site of the latter. We say an edge~$e\in \Ecldec \cup \Ecldecdual$ is labeled by a site~$w \in \cC \setminus \{q_1,q_2\}$ if $\Lambda(e) = w$.

With the labeling map~$\Lambda$, we have 
\begin{align}
\begin{matrix}
    \cX&= &\Lambda(\Edec \times \{1,R\}) \\
    \cZ&= &\Lambda(\Edecdual \times \{1,R\}) \\
    \Lambda(\Ecldec) \cap \cA &=& \Lambda(\Eeven) \\
    \Lambda(\Ecldecdual) \cap \cA &=& \Lambda(\Eodd) \ .
    \end{matrix}\label{eq:lambdamvs}
\end{align}
Sometimes we simply denote an edge~$e$ by the corresponding label~$\Lambda(e)$, leaving the map~$\Lambda$ implicit. Following this convention  Eq.~\eqref{eq:lambdamvs} becomes
\begin{align}
\begin{matrix}
    \cX&= &\Edec \times \{1,R\} \\
    \cZ&= &\Edecdual \times \{1,R\} \\
    \Ecldec \cap \cA &=& \Eeven \\
    \Ecldecdual \cap \cA &=& \Eodd \ . \end{matrix}
\end{align}
Because of the definitions~\eqref{eq:tcldecdef}
and~\eqref{eq:tcldecdualdef}
of the graphs~$\Tcldec$ and~$\Tcldecdual$, this allows us to write the set of edges $\Ecldec$ and $\Ecldecdual$ as unions of pairwise disjoint subsets $\cX$,  $\cZ$, $\Ecldec \cap \cA$ and $\Ecldecdual \cap \cA$ of the measured qubits~$\cC \setminus \{q_1,q_2\}$, namely
\begin{align}
\begin{matrix}
    \Ecldec &= &\cX \cup \left(\Ecldec \cap \cA\right) \\
    \Ecldecdual &=& \cZ \cup \left(\Ecldecdual \cap \cA\right)
    \end{matrix}\label{eq:ecldececldecstar}
\end{align}
In our entanglement generation protocol, we will use~\eqref{eq:ecldececldecstar} to
interpret the measurement outcomes $(a,x,z) \in \{0,1\}^{\cA\cup\Xcl \cup \Zcl}$ as defining two subsets of $\Ecldec$ and $\Ecldecdual$, respectively.

\subsubsection{Recovery sets and internal vertices of $\Tcldec$ and $\Tcldecdual$\label{sec:internalverticestcl}}
We define recovery sets~$\cLXcl\subset \Ecldec$ and $\cLZcl\subset\Ecldecdual$ to be used later, see Fig.~\ref{fig:internalexternalverticescl} for an illustration. The role of these sets in the context of the graph state will be clarified below, see Lemma~\ref{lem:logicalopsSXSZ}.

The recovery set~$\cLXcl\subset\Ecldec$ consists of edges incident on degree-$1$ vertices (so-called dangling edges in the terminology of~\cite{bravyiQuantumAdvantageNoisy2020}) on the left rough boundary of~$\Tcldec$.
That is,
  \begin{align}
    \cLXcl &= \left\{\{(0,u_2,u_3),(2,u_2,u_3)\}\ |\ u_2\geq 2, u_2\textrm{ even, } u_3\in \{1,R\}\right\}\\
    &\quad \cup
    \left\{\{(0,u_2,u_3),(2,u_2,u_3)\} \mid u_2, u_3\textrm{ even }\right\}\ ,
  \end{align}
  or more succinctly (labeling edges by sites with the map~$\Lambda$)
  \begin{align}
    \cLXcl &= \left\{(1,u_2,u_3)\ |\ u_2\geq 2, u_2\textrm{ even, } u_3\in \{1,R\}\right\}\cup
    \left\{(1,u_2,u_3)\ |\ u_2\geq 2, u_2, u_3\textrm{ even }\right\}\ .\label{eq:setunionxm}
  \end{align}
  Observe that the first set in this union is simply $\cL_X\times \{1,R\}$, i.e., it consists of two copies of the set~$\cL_X$ introduced for the surface code in Section~\ref{sec:recoverysetsurface}.
  We can also write this set as $\cLXcl \cap \cB$ or equivalently as $\cLXcl \cap \cX$ (by the fact that $\cB \cap \Ecldec = \cX$) using the labels. 
  Similarly, the second set in the union~\eqref{eq:setunionxm} is $\cLXcl \cap \cA$.
  In summary, the recovery set~$\cLXcl$ is the disjoint union
  \begin{align}
      \cLXcl = \left(\cLXcl \cap \cA\right) \cup \left(\cLXcl \cap \cB\right) \ .
  \end{align}

  The dual recovery set~$\cLZcl\subset\Ecldecdual$ consists of edges incident on degree-$1$ vertices at the bottom of~$\Ecldecdual$.
  For every site $(u_1, u_2, u_3) \in \left\{ (o,-1,o) \right\} \setminus \left\{ q_1, q_2 \right\}\in \Cprimedual$,  there exists such an edge incident on the vertex at~$(u_1, 1, u_3)$; these edges are labeled by their midpoint~$(u_1, 0, u_3)$.
  Such an edge belongs to~$\cA$ if and only if~$u_3 \not \in \{1,R\}$, and to~$\cB$ otherwise.
  In summary, the recovery set~$\cLZcl$ is the disjoint union
  \begin{align}
    \cLZcl &=(\cLZcl \cap \cA)\cup (\cLZcl \cap \cB)
  \end{align}
  where
  \begin{align}
    \cLZcl \cap \cA  &= \left\{ (u_1,u_2, u_3) \in \left\{ (o,0,o) \right\} \mid u_1 \not \in \{1, R\} \right\}\\
    \cLZcl \cap \cB  &= \left\{ (u_1,u_2, u_3) \in \left\{ (o,0,o) \right\} \mid u_1 \in \{1, R\} \right\} \setminus \left\{ q_1, q_2 \right\}\ .
\end{align}

Finally, let us define the set $\Vcldecint\subset\Vcldec$ of internal vertices of the decoding graph~$\Tcldec$ and similarly, the set $\Vcldecdualint\subset\Vcldecdual$ of internal vertices of the dual decoding graph~$\Tcldecdual$. 
We  call their complements $\Vcldecext:=\Vcldec\backslash \Vcldecint$ and $\Vcldecdualext:=\Vcldecdual\backslash \Vcldecdualint$ the sets of external vertices, see Fig.~\ref{fig:internalexternalverticescl}.

For the decoding graph~$\Tcldec$, the set of internal vertices is
\begin{align}
\Vcldecint&=\left\{u\in \Veven\ |\ \neigh(u)\subset \Eeven \right\} \cup \left( \Vdecint \times \left\{ 1, R \right\} \right) \ .\label{eq:tcldecinterior} 
\end{align}
The set of internal vertices of the graph~$\Tcldecdual$ is defined similarly
as
\begin{align}
  \Vcldecdualint&=\left\{u\in\Vodd\ |\ \neigh(u)\subset \Eodd\right\} \cup \left( \Vdecdualint \times \{1,R\} \right)\ ,
\end{align}
or equivalently as 
\begin{align}
  \Vcldecdualint&=\left\{u\in\Vcldecdual\ |\ \neigh(u)\subset \Ecldecdual\right\}\ .\label{eq:tcldecdualinterior} 
\end{align}

\subsection{Stabilizers of the cluster state\label{sec:stabilizersgraphstate}}
In this section, we consider
certain products of the stabilizer generators~$\{G_u\}_{u\in\cC}$ of the cluster state~$W\ket{0^\cC}$ defined by~\eqref{eq:stabilizergeneratorsgraph}.
The first set of operators will be used to define an error syndrome computable from single-qubit measurements performed in our protocol. To introduce
these operators, we  denote by~$\incident_G(v)$ the set of all edges of a graph~$G$ which are incident on a vertex~$v$.
Here the graph $G$ is either $\Tcldec$ or $\Tcldecdual$.
Again using the labeling map~$\Lambda$ (cf.~\eqref{eq:defoflabeling}), we will sometimes regard~$\incident_G(v)$ as a set of sites (associated with the measured qubits).  
Key to our procedure is the fact that certain products of the stabilizer generators~\eqref{eq:stabilizergeneratorsgraph} of the cluster  state~$W\ket{0^\cC}$ have supports given by sets of the form~$\incident_{\Tcldec}(v)$ for $v \in \Vcldec$ or $\incident_{\Tcldecdual}(v)$ for $v \in \Vcldecdual$. More precisely, we have the following:
\begin{lemma}\label{lem:productsofstabilizergeneratorsone}
  Define $S^u$ for $u \in \Vcldecint$ as
  \begin{align}\label{eq:defofSu}
    S^u := \begin{cases*}
      \prod_{v \in \neigh(u)} G_v\qquad & \text{if $u_3\not\in \{1,R\}$}\\
      G_u & \text{if $u_3\in \{1,R\}$}\ .
      \end{cases*}
  \end{align}
  Then
  \begin{align}\label{eq:Suandincident}
    S^u&=  \left(\prod_{v \in \incident_{\Tcldec}(u) \cap \cA} Z_v\right)\left(\prod_{v \in \incident_{\Tcldec}(u) \cap \cB} X_v\right)\ .
    \end{align}
      Similarly, define $S^u$ for $u \in \Vcldecdualint$ as 
  \begin{align}\label{eq:defofSudual}
    S^u = \begin{cases*}
      \prod_{v \in \neigh(u)} G_v \qquad &\textrm{ if $u_3\not\in \{1,R\}$} \\
      G_{u \pm (0,0,1)} \prod_{v \in \neigh(u) \cap \mathcal{B}}G_v  \qquad &\textrm{ if        $u_3\in \{1,R\}$}
    \end{cases*}
  \end{align}
  where the sign $\pm$ in \eqref{eq:defofSudual} is ``$+$'' if $u$ is situated in the front layer, and ``$-$'' if $u$ is situated in the back layer so that $u \pm (0,0,1) \in \mathcal{C}$.
  Then
  \begin{align}\label{eq:Sudualandincident}
    S^u = \prod_{v \in \incident_{\Tcldecdual}(u)}Z_v \qquad \textrm{ for all } \qquad u \in \ \Vcldecdualint\ .
  \end{align}
\end{lemma}
We include a proof of this lemma in Appendix~\ref{sec:appendixproductsofstabilizer}.

We also need two additional products~$S^X$ and $S^Z$ of the stabilizer generators~$\{G_u\}_{u\in\cC}$. These have support on  the two qubits~$\{q_1,q_2\}$
between which the protocol establishes entanglement, and their eigenvalues fix the state of these two qubits as described below (see Lemma~\ref{lem:pmstateisbellbasis}).
\begin{lemma} \label{lem:logicalopsSXSZ}
  Define the operators 
  \begin{align}
    S^X& := \prod_{u \in\cLXcl \cap \cB} G_u \qquad \textrm{ and }\qquad
    S^Z := \prod_{u \in\cLZcl \cup \left\{ q_1, q_2 \right\}} G_u \ .
  \end{align}
  Then
  \begin{align}
    S^X &= X_{q_1}X_{q_2}X\left(\cLXcl\cap \cB\right) Z\left(\cLXcl\cap \cA\right)\label{eq:sxdefinition} \\
    S^Z &= Z_{q_1}Z_{q_2}Z\left(\cLZcl\cap \cB\right) Z\left( \cLZcl\cap \cA\right)\ .\label{eq:szdefinition}
  \end{align}
\end{lemma}
We note that similar expressions for products of stabilizer generators have been exploited in the seminal work~\cite{raussendorfLongrangeQuantumEntanglement2005}. For completeness, we give a proof of this lemma in Appendix~\ref{sec:appendixproductsofstabilizer}.

\subsection{Description of the entanglement generation protocol}\label{sec:alg}
The entanglement generation protocol is given as Algorithm~\ref{alg:ftes} below. It takes as input a state~$\ket{\Psi_{\textrm{in}}}$ of the qubits on~$\cC$ and produces  a two-qubit state~$\ket{\Psi_{\textrm{out}}}$ on qubits~$\{q_1,q_2\}$. 

In the following, we again identify a  string~$a\in\{0,1\}^{\cC}$ with a subset of~$\cC$ (interpreting~$a$ as its characteristic function). The bit-wise XOR~$a\oplus b$ of two strings~$a,b\in \{0,1\}^{\cC}$ then corresponds to taking the symmetric difference of the associated sets.

\begin{algorithm}[H]
\caption{Fault-tolerant entanglement generation protocol}
\label{alg:ftes}
\begin{algorithmic}[1]
\Require Input: A state on $(\mathbb{C}^2)^{\otimes \cC}$, where $\cC = \cC[d\times d\times R]$ is the set of the locations of qubits.

  \State Perform the following measurements: measure
  \begin{center}
  \begin{tabular}{c|c|c}
   every qubit in the set & in the  & denote the outcome by \\
   \hline
    $\Xcl$ & $X$-basis & $x\in \{0,1\}^{\Xcl}$\\
    $\Zcl$ & $Z$-basis & $z\in \{0,1\}^{\Zcl}$\\
    $\cA$ & $Z$-basis & $a \in \{0,1\}^{\cA}$\\
  \end{tabular}
  \end{center}
  \label{it:stepmeasurement}
    \State Compute the boundaries of the edges $(a \cap \Ecldec) \oplus x$ and $(a \cap \Ecldecdual) \oplus z$, respectively, i.e., set
  \begin{align}
  s &\leftarrow \partial_{\Tcldec}\left( (a \cap \Ecldec)\oplus x \right)\label{eq:sdefinitionboundary}\\
  s^{*} &\leftarrow \partial_{\Tcldecdual}\left( (a \cap \Ecldecdual) \oplus z \right)
  \end{align}\label{state:setsandsdual}
  \State Compute minimal matchings 
  \begin{align}
  m &\leftarrow \mmatch_{\Tcldec}(s)\\
m^* &\leftarrow \mmatch_{\Tcldecdual}(s^{*})\ 
\end{align}\label{state:setm}
\State Compute the bits\label{it:syndromebitestimation}
\begin{align}
  \hat{s}^X&\leftarrow \ztwoinner{m}{\cLXcl}\label{eq:hatsxdef}\\
\hat{s}^Z&\leftarrow \ztwoinner{m^{*}}{\cLZcl}\label{eq:hatszdef}\ 
\end{align}\label{state:sydnromebitsxsz}

  \State Determine  the bits
  \begin{align}
    \hat{c}_X &\leftarrow \ztwoinner{a}{\cLXcl\cap \cA} \oplus \ztwoinner{x}{\cLXcl\cap \cB} \oplus \hat{s}^X\\
    \hat{c}_Z& \leftarrow \ztwoinner{a}{\cLZcl\cap \cA} \oplus \ztwoinner{z}{\cLZcl\cap \cB} \oplus \hat{s}^Z
  \end{align} \label{state:estimateofc}
  \State Apply $Z^{\hat{c}_X} X^{\hat{c}_Z}$ to $q_1$.\label{state:estimatecorrection}
  \State Return the two-qubit state on the qubits~$\{q_1,q_2\}$.
  \end{algorithmic}
\end{algorithm}

We will show that if the input state~$\ket{\Psi_{\textrm{in}}}$ is the cluster state~$W\ket{0^\cC}$, then the output of Algorithm~\ref{alg:ftes} is the two-qubit Bell state~$\ket{\Phi}$. 
More generally, we will show that the procedure  yields a Bell state~$\ket{\Phi_{(\alpha,\beta)}}$  if we start with a cluster state corrupted by a Pauli error.

\subsection{Entanglement generation from a corrupted cluster state}\label{sec:alganalysiscorrectness}
In this section, we examine the result of running Algorithm~\ref{alg:ftes} on the corrupted cluster state
\begin{align}
  \ket{\Psi^{\cC}_{\textrm{in}}}&=EW\ket{0^\cC}\ \label{eq:noisyclusterstate}
\end{align}
for a fixed Pauli error~$E$. 
We show that the final state on qubits~$q_1,q_2$ is one of the four Bell states~$\{\ket{\Phi_{(\alpha,\beta)}}\}_{\alpha,\beta\in \{0,1\}}$, and give a formula for $(\alpha,\beta) \in \{0,1\}^2$ in terms of the error~$E$ (see Theorem~\ref{thm:successcondition}). This is the basis of our analysis of local stochastic noise in Section~\ref{sec:resiliencenoiselocalstochastic}.

To determine the state at the end of the protocol, first consider the post-measurement state
\begin{align} \label{eq:pmstateinalg}
    \ket{\psi_{\mathrm{pm}}(a,x,z,E)} = \frac{1}{\sqrt{p(a,x,z|E)}} \left( \bra{a}_{\mathcal{A}} \otimes  \bra{x}_{\mathcal{X}}H(\mathcal{X})^\dagger\otimes \bra{z}_{\Zcl} \otimes I_{\{q_1,q_2\}} \right) E   W\ket{0^{\cC}}\ 
  \end{align}
obtained after step~\ref{it:stepmeasurement} of the algorithm  on qubits~$\{q_1,q_2\}$. Here  $p(a,x,z \mid E)$ denotes the probability of obtaining the outcomes $(a,x,z)$ given a Pauli error~$E$, i.e., for the input state~\eqref{eq:noisyclusterstate}.  This state is a Bell state determined by the measurement outcomes~$(a,x,z)$ and the error~$E$  as follows:
\begin{lemma} \label{lem:pmstateisbellbasis}
  Let $S^X$ and $S^Z$ be the operators introduced in Lemma~\ref{lem:logicalopsSXSZ}. Let
  \begin{align}
    c_X &:= \ztwoinner{a}{\cLXcl\cap \cA} \oplus \ztwoinner{x}{\cLXcl \cap \cB} \oplus \syn(S^X, E)\label{eq:cxdefinition}\\
  c_Z &:= \ztwoinner{a}{\cLZcl \cap\cA} \oplus \ztwoinner{z}{\cLZcl \cap \cB} \oplus \syn(S^Z, E)\ .\label{eq:czdefinition}
  \end{align}
  Then, given measurement outcomes~$(a,x,z)$,
  the post-measurement state after Step~\ref{it:stepmeasurement} of Algorithm~\ref{alg:ftes} running on a corrupted input state $EW\ket{0^{\cC}}$ is 
  \begin{align}
      \ket{\psi_{\mathrm{pm}}(a,x,z,E)} &=\ket{\Phi_{(c_X,c_Z)}}\ .
  \end{align}
  \end{lemma}
\begin{proof}
The claim is an immediate consequence of the following two facts:
\begin{enumerate}[(i)]
\item
The operators~$S^X$ and $S^Z$
are defined as products of the stabilizer generators~$\{G_u\}_{u\in\cC}$ and thus stabilize the ideal cluster state~$W\ket{0^\cC}$.\label{it:propertyonesxszstabilize}
\item
  The operators $S^X$ and $S^Z$ are -- up to additional Pauli operators on the measured qubits -- proportional to the stabilizer generators~$\{X_{q_1}X_{q_2},Z_{q_1}Z_{q_2}\}$ of the Bell state~$\ket{\Phi_{(0,0)}}$, see~\eqref{eq:sxdefinition} and~\eqref{eq:szdefinition} of Lemma~\ref{lem:logicalopsSXSZ}.
  \end{enumerate}
Replacing $W\ket{0^{\cC}}$ in \eqref{eq:pmstateinalg} by $S^XW\ket{0^{\cC}}$ using property~\eqref{it:propertyonesxszstabilize}, we observe that the post-measurement state~$\psi_{\mathrm{pm}}=\psi_{\mathrm{pm}}(a,x,z,E)$ 
satisfies
  \begin{align}
\sqrt{p}    \ket{\psi_{\mathrm{pm}}}
    &= \left( \bra{a}_{\mathcal{A}} \otimes \bra{x}_{\mathcal{X}}H(\mathcal{X})^\dagger \otimes \bra{z}_{\mathcal{Z}}\otimes I_{\{q_1,q_2\}} \right) E S^X E\ket{0^{\cC}} \\
    &=  (-1)^{\syn(S^X, E)} \left( \bra{a}_{\mathcal{A}} \otimes   \bra{x}_{\mathcal{X}}H(\mathcal{X})^\dagger\otimes \bra{z}_{\mathcal{Z}} \otimes I_{\{q_1,q_2\}} \right) S^X E  \ket{0^{\cC}} \\
    &=  (-1)^{\ztwoinner{a}{\cLXcl \cap \cA} \oplus \ztwoinner{x}{\cLXcl \cap \cB} \oplus \syn(S^X, E)} \left( \bra{a}_{\mathcal{A}} \otimes \bra{x}_{\mathcal{X}}H(\mathcal{X})^\dagger \otimes \bra{z}_{\mathcal{Z}}\otimes X_{q_1}X_{q_2}\right) E  W\ket{0^{\cC}}\ ,
    \end{align}
      where $p := p(a,x,z|E)$ and where we used~\eqref{eq:sxdefinition} as well as the identities
      \begin{align}
      X(\cLXcl\cap\cB)H(\cX)\ket{x}_{\cX}&=(-1)^{\ztwoinner{x}{\cLXcl\cap\cB}}H(\cX)\ket{x}_{\cX}\\
      Z(\cLXcl\cap \cA)\ket{a}_{\cA}&=(-1)^{\ztwoinner{a}{\cLXcl\cap\cA}}\ket{a}_{\cA}\ .
      \end{align}
      We conclude that 
    \begin{align}
\ket{\psi_{\mathrm{pm}}}        &= (-1)^{\ztwoinner{a}{\cLXcl\cap\cA} \oplus \ztwoinner{x}{\cLXcl\cap\cB} \oplus \syn(S^X, E)} X_{q_1}X_{q_2} \ket{\psi_{\mathrm{pm}}} \\
    &= (-1)^{c_X} X_{q_1} X_{q_2} \ket{\psi_{\mathrm{pm}}}\ ,
      \end{align}
      by definition of~$c_X$. 

  In an analogous manner, we can check $\ket{\psi_{\mathrm{pm}}} = (-1)^{c_Z} Z_{q_1}Z_{q_2} \ket{\psi_{\mathrm{pm}}}$ by replacing $W\ket{0^{\cC}}$ in~\eqref{eq:pmstateinalg} by $S^Z W \ket{0^{\cC}}$. Thus $\psi_{\mathrm{pm}}$
      is stabilized by~$\{(-1)^{c_X}X_{q_1}X_{q_2},(-1)^{c_Z}Z_{q_1}Z_{q_2}\}$. This implies the claim.
\end{proof}

According to Lemma~\ref{lem:pmstateisbellbasis}, the post-measurement state is fully determined by the single-qubit measurement results~$(a,x,z)$ and the bits~$\syn(S^X, E)$ and~$\syn(S^Z, E)$. These bits in turn are
determined by the physical errors on qubits~$\{q_1,q_2\}$ as well as a certain restriction of the error as specified in the following statement:
  \begin{lemma} \label{lem:actualsynd}
  Let $E$ be a Pauli operator on $\mathcal{C}$ and let $S^X$ and $S^Z$ be the Pauli operators defined in Lemma~\ref{lem:logicalopsSXSZ}.
Let us decompose~$E$ into a product~$E = E^XE^Z$ of Pauli-$X$- and Pauli-$Z$-operators, respectively. Define  
  \begin{align} 
    E_{\mathrm{gl}} &= E^{X}|_{\Ecldec \cap \cA} E^{Z}|_{\cX}  \qquad \textrm{ and } \label{eq:defofEgl}\\
    E_{\mathrm{gl}}^{*} &= E^{X}|_{\Ecldecdual}\ \label{eq:defofEgldual}
  \end{align}
  where $F|_{\Omega} := \prod_{j \in \Omega}F_j$ for any Pauli operator $F=\prod_{j\in\cC}F_j$ on $\cC$ and $\Omega \subset \cC$.
  Then 
  \begin{align} 
    \syn(S^X, E) &= \ztwoinner{\{q_1,q_2\}}{\supp(E^Z)} \oplus \ztwoinner{\cLXcl}{\supp(E_{\mathrm{gl}})}\label{eq:synviaEgl}\\
    \syn(S^Z, E) &= \ztwoinner{\{q_1,q_2\}}{\supp(E^X)} \oplus \ztwoinner{\cLZcl}{\supp(E_{\mathrm{gl}}^{*})}\ .\label{eq:synviaEgldual}
  \end{align}
\end{lemma}
\begin{proof}
  Let $E$ be an arbitrary Pauli operator on~$\cC$.
  Using expression~\eqref{eq:sxdefinition} for~$S^X$, we have 
  \begin{align}
    \syn(S^X,E)&=\syn\left(X_{q_1}X_{q_2}X\left(\cLXcl\cap\cB \right) Z\left( \cLXcl \cap \cA \right),E\right)\\
    &= \ztwoinner{\{q_1,q_2\}}{\supp(E^Z)} \oplus \ztwoinner{\cLXcl\cap\cB}{\supp(E^Z)} \oplus \ztwoinner{\cLXcl \cap \cA }{\supp(E^X)} \\
    &= \ztwoinner{\{q_1,q_2\}}{\supp(E^Z)} \oplus \ztwoinner{\cLXcl\cap\cB}{\supp(E^Z|_{\cX})} \oplus \ztwoinner{\cLXcl \cap \cA }{\supp(E^X|_{\Adec})} \\
              &= \ztwoinner{\{q_1,q_2\}}{\supp(E^Z)} \oplus \ztwoinner{\cLXcl}{\supp(E_{\mathrm{gl}})} 
  \end{align}
  by definition of $E_{\mathrm{gl}}$, see Eq.~\eqref{eq:defofEgl}.  
  A similar calculation yields
  \begin{align}
    \syn(S^Z, E)
    &= \ztwoinner{\{q_1,q_2\}}{\supp(E^X)} \oplus \ztwoinner{\cLZcl\cap\cB}{\supp(E^X|_{\cZ})} \oplus \ztwoinner{\cLZcl\cap\cA}{\supp\left(E^X|_{\Adecdual}\right)} \\
              &= \ztwoinner{\{q_1,q_2\}}{\supp(E^X)} \oplus \ztwoinner{\cLZcl}{\supp(E_{\mathrm{gl}}^{*})} \ ,
  \end{align}
  as claimed.
\end{proof}
Unfortunately, the operators~$E_{\mathrm{gl}}$ and $E_{\mathrm{gl}}^{*}$ defined in~\eqref{eq:defofEgl} and~\eqref{eq:defofEgldual} are not determined by the measurement results~$(a,x,z)$. In particular, $\syn(S^X,E)$ and $\syn(S^Z,E)$ cannot be computed from the measurement outcomes. Algorithm~\ref{alg:ftes} therefore generates estimates~$\hat{s}^X$ and~$\hat{s}^Z$ for these bits, see step~\eqref{it:syndromebitestimation}. 

To motivate the definition of these bits, and compute the final state after the entanglement generation protocol, we discuss how the measurement results~$(a,x,z)$ constrain the error~$E$, or more precisely
the operators~$E_{\mathrm{gl}}$ and $E_{\mathrm{gl}}^{*}$. Observe that by their definitions~\eqref{eq:defofEgl} and~\eqref{eq:defofEgldual}, these operators are fully determined by the supports~$\supp(E_{\mathrm{gl}})$ and~$\supp(E_{\mathrm{gl}}^*)$.
We obtain constraints on these sets expressed in terms of the strings~$s,s^*$ computed in step~\ref{state:setsandsdual} of  Algorithm~\ref{alg:ftes}. 

\begin{lemma} \label{lem:sisboundary}
The subsets $s\subset \Vcldecint$ and $s^{*}\subset \Vcldecdualint$ computed (from the measurement outcomes~$(a,x,z)$) in 
Step~\ref{state:setsandsdual}  satisfy
  \begin{align}
    s &= \partial_{\Tcldec}(\supp(E_{\mathrm{gl}})) \label{eq:stglclaim}\\
    s^{*} &= \partial_{\Tcldecdual}(\supp(E_{\mathrm{gl}}^{*}))\ .\label{eq:stglclaimdual}
  \end{align}
\end{lemma}
\begin{proof}
  Let $s = (s_u)\in \{0,1\}^{\Vcldecint}$ be the string computed in step~\ref{state:setsandsdual}, i.e., 
  \begin{align}
  s=\partial_{\Tcldec}\left( \adec \oplus x \right)\ ,
  \end{align}
  and let $u \in \Vcldecint$. Linearity of the map~$\partial_{\Tcldec}$ implies that
  \begin{align}
    s_u &= \left( \partial_{\Tcldec} \left( a \right) \right)_u \oplus \left(  \partial_{\Tcldec} \left( x \right) \right)_u\\
        &=\ztwoinner{a}{\incident_{\Tcldec}(u)} \oplus \ztwoinner{x}{\incident_{\Tcldec}(u) } \ .\label{eq:sudefinitionax}
  \end{align}
  Since
  $\left( \partial_{\Tcldec} \left( \supp(E_{\mathrm{gl}}) \right) \right)_u = \ztwoinner{\incident_{\Tcldec}(u)}{\supp(E_{\mathrm{gl}})}$, we need to show that
  \begin{align}
    \ztwoinner{a}{\incident_{\Tcldec}(u)} \oplus \ztwoinner{x}{\incident_{\Tcldec}(u) } =
    \ztwoinner{\incident_{\Tcldec}(u)}{\supp(E_{\mathrm{gl}})} \ .\label{eq:suandtu}
  \end{align}
  We first show that
  \begin{align}
    \ztwoinner{a}{\incident_{\Tcldec}(u)} \oplus \ztwoinner{x}{\incident_{\Tcldec}(u) } =
    \syn(S^u, E) \ .\label{eq:suandSu}
  \end{align}
   Note that $S^u$ stabilizes the state~$W\ket{0^\cC}$ since it is a product of  the stabilizer generators~$\{G_v\}_v$ of~$W\ket{0^{\cC}}$.
  Replacing $W\ket{0^{\cC}}$ in~\eqref{eq:pmstateinalg} by $S^uW\ket{0^{\cC}}$ yields 
  \begin{align}
    \ket{\psi_{\mathrm{pm}}}
    &= \frac{1}{\sqrt{p}} \left( \bra{a}_{\mathcal{A}} \otimes \bra{x}_{\mathcal{X}}H(\mathcal{X})^\dagger \otimes \bra{z}_{\mathcal{Z}} \otimes I_{\{q_1,q_2\}} \right) E S^uW
    \ket{0^{\cC}}\\
    &= \frac{1}{\sqrt{p}} (-1)^{\syn(S^u, E)} \left( \bra{a}_{\mathcal{A}} \otimes \bra{x}_{\mathcal{X}}H(\mathcal{X})^\dagger \otimes \bra{z}_{\mathcal{Z}} \otimes I_{\{q_1,q_2\}} \right) S^u E    W\ket{0^{\cC}}\\
    &=  (-1)^{\syn(S^u, E) \oplus \ztwoinner{a}{\incident_{\Tcldec}(u)} \oplus \ztwoinner{x}{\incident_{\Tcldec}(u) }} \ket{\psi_{\mathrm{pm}}} \ .\label{eq:psipmeigenval}
  \end{align}
Here the third equality follows from the identity
  \begin{align}
  S^u \left(\ket{a}_\cA(H(\cX)\ket{x}\right)&=
  (-1)^{\ztwoinner{a}{\incident_{\Tcldec}(u)} \oplus \ztwoinner{x}{\incident_{\Tcldec}(u) }} \left(\ket{a}_\cA(H(\cX)\ket{x}\right)\ 
  \end{align}
  which is a consequence of~\eqref{eq:Suandincident} and the fact that~$a\subset \cA$ and $x\subset \cX\subset\cB$ are disjoint.
  Identity~\eqref{eq:psipmeigenval} together with~\eqref{eq:sudefinitionax} establishes~\eqref{eq:suandSu}.

  Next, we show that
  \begin{align}
    \ztwoinner{\incident_{\Tcldec}(u)}{\supp(E_{\mathrm{gl}})}
    = \syn(S^u, E) \ .\label{eq:tuandSu}
  \end{align}
  Observe that according to~\eqref{eq:Suandincident}, we have
  \begin{align}
    \syn(S^u, E) &=
                   \syn\left(\left(\prod_{v \in \incident_{\Tcldec}(u) \cap \cA} Z_v\right)\left(\prod_{v \in \incident_{\Tcldec}(u) \cap \cB} X_v\right),E\right)\\
                 &=\ztwoinner{\incident_{\Tcldec}(u) \cap \mathcal{A}}{\supp(E^X)} \oplus \ztwoinner{\incident_{\Tcldec}(u) \cap \mathcal{B}}{\supp(E^Z)} 
  \end{align}
  for any $u\in \Vcldecint$.
  From this expression, the definition~\eqref{eq:defofEgl}
  of~$E_{\mathrm{gl}}$ and the fact that~$\cA$ and $\cB$ are disjoint, the identity~\eqref{eq:tuandSu} follows.
   Together~\eqref{eq:suandSu}, the claim~\eqref{eq:suandtu} follows. 
  
  The claim for $s^{*}$ can be proved in an analogous manner using \eqref{eq:Sudualandincident}.  
\end{proof}
We now argue that the strings~$m$ and $m^*$ computed
in Step~\ref{state:setm}
of Algorithm~\ref{alg:ftes} can be considered to be proxys for the sets~$\supp(E_{\mathrm{gl}})$ and $\supp(E_{\mathrm{gl}}^*)$. More precisely, we may define
a Pauli error
\begin{align}
\widehat{E}:=\widehat{E}_{\mathrm{gl}}\widehat{E}_{\mathrm{gl}}^*\qquad\textrm{ with }\qquad \begin{matrix}
    \widehat{E}_{\mathrm{gl}} &:=& X(m\cap \cA) Z(m \cap \cB) \\ 
    \widehat{E}_{\mathrm{gl}}^{*} &:=& Z(m^{*})\ .
    \end{matrix}
\end{align}
Then we have $m=\supp(\widehat{E}_{\mathrm{gl}})$ and $m^*=\supp(\widehat{E}_{\mathrm{gl}}^*)$   by definition of these operators, and it  follows immediately from the
   definition of~$m$ and~$m^*$
  that~$(\widehat{E}_{\mathrm{gl}},\widehat{E}_{\mathrm{gl}}^{*})$
  satisfy constraints analogous to   the constraints~\eqref{eq:stglclaim},~\eqref{eq:stglclaimdual}   obeyed by $(E_{\mathrm{gl}},E_{\mathrm{gl}}^*)$, i.e.,  
  \begin{align}
    s &= \partial_{\Tcldec}(\supp(\widehat{E}_{\mathrm{gl}})) \\
    s^{*} &= \partial_{\Tcldecdual}(\supp(\widehat{E}_{\mathrm{gl}}^{*}))\ .
  \end{align}
  Thus~$\widehat{E}= \widehat{E}_{\mathrm{gl}} \widehat{E}_{\mathrm{gl}}^*$ is an error consistent with the observed syndrome~$(s,s^*)$ with the property that
  $\widehat{E}_{\mathrm{gl}}$ and $\widehat{E}_{\mathrm{gl}}^*$
  each are of minimum weight relative to~$\Vcldecint$ and $\Vcldecdualint$, respectively.

Because the error~$\widehat{E}$ can be computed from the measurement outcomes, this motivates using the pair~$\left(
\syn(S^X,\widehat{E}),
\syn(S^Z,\widehat{E})
\right)$ as an estimate for 
$\left(
\syn(S^X,E),
\syn(S^Z,E)
\right)$.  Indeed, this is the reasoning underlying the definition
\begin{align}
  \begin{matrix}
  \hat{s}^X&:=&\ztwoinner{\cLXcl}{m}\\
  \hat{s}^Z&:=&\ztwoinner{\cLZcl}{m^{*}}
  \end{matrix}\label{eq:hatsxszdefclaim}
\end{align}
in  Step~\ref{state:sydnromebitsxsz} of the algorithm: It is easy to check that~$(\hat{s}^X,\hat{s}^Z)$ defined in this way satisfy
\begin{align}
  \hat{s}^X&=\syn(S^X, \widehat{E}_{\mathrm{gl}})\\
  \hat{s}^Z&=\syn(S^Z, \widehat{E}_{\mathrm{gl}}^*)\ 
\end{align}
because of expressions~\eqref{eq:synviaEgl} and since $\widehat{E}_{\mathrm{gl}}$ and $\widehat{E}_{\mathrm{gl}}^{*}$ have no support on $\{q_1,q_2\}$.

Having motivated the algorithm, let us next give an expression for the resulting final state when the algorithm is run on a corrupted cluster state.
\begin{theorem} \label{thm:successcondition}
Let $E$ be a Pauli error on~$\cC$. Suppose we run Algorithm~\ref{alg:ftes}  on the input state~$EW\ket{0^\cC}$.
  Then the output is the Bell state~$\Phi_{(\alpha,\beta)}$ where
    \begin{align}
    \alpha &=\ztwoinner{\{q_1,q_2\}}{\supp(E^Z)} \oplus \ztwoinner{\supp(E_{\mathrm{gl}}) \oplus \mmatch_{\Tcldec}\left(\partial_{\Tcldec}(\supp(E_{\mathrm{gl}}))\right)}{\cLXcl}\\
    \beta &=\ztwoinner{\{q_1,q_2\}}{\supp(E^X)} \oplus \ztwoinner{\supp(E_{\mathrm{gl}}^{*}) \oplus\mmatch_{\Tcldecdual}\left(\partial_{\Tcldecdual}(\supp(E_{\mathrm{gl}}^{*})\right)}{\cLZcl}\ . 
  \end{align}
\end{theorem}
\begin{proof}
  Because $m,m^*$ are equal to
  \begin{align}
  m &= \mmatch_{\Tcldec}(s)=\mmatch_{\Tcldec}\left(\partial_{\Tcldec}(\supp(E_{\mathrm{gl}}))\right)\\
m^* &= \mmatch_{\Tcldecdual}(s^{*})=\mmatch_{\Tcldecdual}\left(\partial_{\Tcldecdual}(\supp(E_{\mathrm{gl}}^{*})\right)\ ,
  \end{align}
  see Step~\ref{state:setm} of the algorithm and Lemma~\ref{lem:sisboundary}, the claim amounts to the statement that the output state is the Bell state~$\Phi_{(\alpha,\beta)}$ where
  \begin{align}
    \alpha &=\ztwoinner{\{q_1,q_2\}}{\supp(E^Z)} \oplus \ztwoinner{\supp(E_{\mathrm{gl}}) \oplus m}{\cLXcl}\\
    \beta &=\ztwoinner{\{q_1,q_2\}}{\supp(E^X)} \oplus \ztwoinner{\supp(E_{\mathrm{gl}}^{*}) \oplus m^{*}}{\cLZcl}\ .
  \end{align}

  According to Lemma~\ref{lem:pmstateisbellbasis}, the
  post-measurement state
  after Step~\ref{it:stepmeasurement} 
  of Algorithm~\ref{alg:ftes} is the Bell state~$\ket{\Phi_{(c_X,c_Z)}}$
  with
  \begin{align}
c_X&= \ztwoinner{a}{\cLXcl\cap\cA} \oplus \ztwoinner{x}{\cLXcl\cap\cB}\oplus \syn(S^X, E)\\
c_Z&= \ztwoinner{a}{\cLZcl\cap\cA} \oplus \ztwoinner{z}{\cLZcl\cap\cB}\oplus  \syn(S^Z, E)\ .
    \end{align}
    Since the algorithm applies
    the correction operation $Z^{\hat{c}_X}X^{\hat{c}_Z}$
    in Step~\ref{state:estimatecorrection}, the output of the algorithm is  the Bell state~$\ket{\Psi_{(c_X\oplus\hat{c}_X,c_Z\oplus\hat{c}_Z)}}$. It follows with the definition of $(\hat{c}_X,\hat{c}_Z)$ in Step~\ref{state:estimateofc},
    i.e.,
    \begin{align}
  \hat{c}_X &= \ztwoinner{a}{\cLXcl\cap\cA} \oplus \ztwoinner{x}{\cLXcl\cap\cB} \oplus \hat{s}^X\\
  \hat{c}_Z&=\ztwoinner{a}{\cLZcl\cap\cA} \oplus \ztwoinner{z}{\cLZcl\cap\cB} \oplus \hat{s}^Z
    \end{align}
    that
    the output is the Bell state~$\ket{\Phi_{(\alpha,\beta)}}$ where
    \begin{align}
      \alpha &= \hat{s}^X \oplus \syn(S^X,E)\\
      \beta &= \hat{s}^Z \oplus \syn(S^Z,E)\ .
    \end{align}
    The claim now follows by combining the expressions~\eqref{eq:synviaEgl}, \eqref{eq:synviaEgldual} for $\left(\syn(S^X,E),\syn(S^Z,E)\right)$    from Lemma~\ref{lem:actualsynd}      with the
 Definitions~\eqref{eq:hatsxszdefclaim} of $(\hat{s}^X,\hat{s}^Z)$ used in  Step~\eqref{state:sydnromebitsxsz} of the algorithm. 
\end{proof}

We note that all theorems and lemmas in this section still hold if we replace the function $\mmatch$ on $\Tcldec$ and $\Tcldecdual$ with any function on $\Vcldecint$ and $\Vcldecdualint$ with binary output.
The fact that $\mmatch$ produces a matching is -- up to this point -- only used to motivate  the definition of $\hat{s}^X$ and $\hat{s}^Z$ in Algorithm~\ref{alg:ftes}. 
We will use this property of $\mmatch$ in Section~\ref{sec:resiliencenoiselocalstochastic} to show that our protocol is robust against local stochastic noise.

\section{Entanglement generation with local stochastic noise\label{sec:resiliencenoiselocalstochastic}}
In this section, we show that the entanglement generation protocol given by Algorithm~\ref{alg:ftes} applied to the cluster state~$W\ket{0^{\cC}}$ is robust against local stochastic noise $E$ of any noise strength $p$ below some threshold~$p_0$. In more detail, we consider a  cluster state lattice~$\cC$  of the form $\cC=\cC[d \times d \times R]$ (with $d=O(\log R)$) and show that 
even when the ideal cluster state~$W\ket{0^{\cC}}$  is replaced by~$EW\ket{0^{\cC}}$, the resulting state has fidelity at least $1-O(p)$ with the Bell state~$\Phi$ independently of~$R$ (see Corollary~\ref{cor:achievability}). In particular, the range~$R$ of the entanglement generation can be exponential in the surface code distance~$d$ while still generating constant-fidelity entanglement.

We proceed as follows: In Section~\ref{sec:boundrescl}, we construct upper bounds on the resilience functions~$\res_{\cLXcl}(p)$ and $\res_{\cLZcl}(p)$ of the recovery sets $\cLXcl$ and $\cLZcl$. In Section~\ref{sec:boundprobsucccl}, we then establish a lower bound on the success probability of the entanglement generation  protocol using the results of Section~\ref{sec:boundrescl}. 

\subsection{Bounds on the resilience function for $\Tcldec$ and $\Tcldecdual$}\label{sec:boundrescl}
Here we establish upper bounds on the resilience functions of the recovery sets~$\cLXcl$ and $\cLZcl$
for the decoding graphs~$\Tcldec$ and $\Tcldecdual$ for small values of~$p$. These upper bounds immediately give lower bounds on the success probability of the entanglement generation protocol, see Section~\ref{sec:boundprobsucccl}.

\begin{lemma}\label{lem:rescLXclbound}
  Consider the decoding graph~$\Tcldec=(\Vcldec, \Ecldec)$ with set of internal vertices~$\Vcldecint$ defined on the lattice $\Ctilde[d \times d \times R]$. Let  $\cLXcl\subset \Ecldec$ be the recovery set introduced in Section~\ref{sec:decodinggraphs}. Then the resilience function~$\res_{\cLXcl}$ satisfies
  \begin{align}\label{eq:rescLXclbound}
    \res_{\cLXcl}(p) \leq \left(2402 + 100d(R-1)\left(10 \sqrt{p}\right)^{d-2}  \right) \cdot p \ 
  \end{align}
  for any $p\in \interval{0}{\frac{1}{400}}$.
\end{lemma}
\begin{proof}
  Let $L \in \Zcirc(\Tcldec)$ be a simple closed loop. Then it is clear from the definition of~$\cLXcl$ that~$L$ and $\cLXcl$ have no overlap, and in particular, $\ztwoinner{L}{\cLXcl}=0$. This implies that simple closed loops do  not contribute to the resilience function~$\res_{\cLXcl}(p)$ (see Definition~\eqref{eq:defofres}).
  
  To describe the set~$\Zpathext(\Tcldec)$ of simple paths connecting external vertices through internal vertices,
  we partition the set~$\Vcldecext = \Vcldecextleft \cup \Vcldecextright$ of external vertices into the set~$\Vcldecextleft$ of vertices forming the ``left'' boundary of the decoding graph~$\Tcldec$
  and the set~$\Vcldecextright$ of remaining external vertices forming the ``right'', ``front'', and ``back'' boundaries.
  We label the external vertices in~$\Vcldecext$ as 
  \begin{align}
    \Vcldecextleft &= \left\{u_{j}^{k} \mid 1 \leq j \leq d-1 , k \in \left\{0, \frac{R+1}{2}\right\}\right\} \cup \left\{u_j^{k} \mid 0 \leq j \leq d-1, 1 \leq k \leq \frac{R-1}{2} \right\} \\
    \Vcldecextright &= \left\{v_{i,j}^{k} \mid 1 \leq i \leq d-1, 0 \leq j \leq i, k \in \left\{0, \frac{R+1}{2}\right\} \right\} \cup \left\{v_j^{k} \mid 0 \leq j \leq d-1, 1 \leq k \leq \frac{R-1}{2} \right\} \ ,
  \end{align}
  see Fig.~\ref{fig:internalexternalverticescl}. 
  Written out, the external vertices are
  \begin{align}
    u_j^{k} &= \begin{cases*}
     (0, 2j, 1) & \textrm{ for $k = 0$, $1 \leq j \leq d-1$} \\
     (0, 2j, 2k) & \textrm{ for $1 \leq k \leq \frac{R-1}{2}$, $0 \leq j \leq d-1$} \\
     (0, 2j, R) & \textrm{ for $k = \frac{R+1}{2}$, $1 \leq j \leq d-1$}
    \end{cases*}\label{eq:tcldecleftlabel} \\
    v_{i,j}^{k} &= \begin{cases*}
      (2i, 2j, 1) & \textrm{ for  $k = 0$, $1 \leq i \leq d-1$, $0 \leq j \leq i$}  \\
      (2i, 2j, R) & \textrm{ for  $k = \frac{R+1}{2}$, $1 \leq i \leq d-1$, $0 \leq j \leq i$}  
    \end{cases*}\\
    v_j^{k} &= (2d, 2j, 2k) \qquad \textrm{ for  $1 \leq k \leq \frac{R-1}{2}$, $0 \leq j \leq d-1$} \ .\label{eq:tcldecrightlabel}
  \end{align}

  It is clear from the definition of~$\Tcldec$ that a simple path~$P \in \Zpathext(\Tcldec)$ satisfies~$\ztwoinner{P}{\cLXcl}=1$ if and only if~$P$ connects an external vertex~$u \in \Vcldecextleft$ to another external vertex~$v \in \Vcldecextright$.
  Here we say such a path~$P$ (through internal vertices) starts at~$u$ and ends at~$v$.
  Let us denote the set of all simple paths of length~$\ell$ which start at~$u_j^{k} \in \Vcldecextleft$ and end at~$v \in \Vcldecextright$ through internal vertices
   by~$\Delta(u_j^{k}, \ell)$. 
  Then we can write the set of all simple paths of length~$\ell$ in~$\Zpathext(\Tcldec)$ satisfying~$\ztwoinner{P}{\cLXcl}=1$ as
  \begin{align}\label{eq:zpathexttcldecdecompose}
    \{P \in \Zpathext(\Tcldec) \mid \abs{P}=\ell \textrm{ and } \ztwoinner{P}{\cLXcl} = 1\} &= \bigcup_{u \in \Vcldecextleft} \Delta(u, \ell) \ .
  \end{align}
  Note that~$\Delta(u, \ell)$ and~$\Delta(u', \ell')$ are disjoint for~$(u,\ell) \neq (u', \ell^{\prime})$
  since the corresponding paths either start at different vertices or have different lengths.
  We also note that~$\Delta(u, \ell)$ is empty for sufficiently large~$\ell$ since we are considering simple paths and the decoding graph~$\Tcldec$ is finite.
  It follows from the definition~\eqref{eq:defofres} of the resilience function and from~\eqref{eq:zpathexttcldecdecompose} that
  \begin{align}
    \res_{\cLXcl}(p) &= \sum_{\ell=1}^{\infty} \binom{\ell}{\ceil{\ell / u}} \cdot \abs{\bigcup_{u \in \Vcldecextleft} \Delta(u, \ell)} \cdot p^{\ceil{ \ell / 2 }}\\
                     &= \sum_{u \in \Vcldecextleft}\sum_{\ell = L_{u}^{\min}}^{L_{u}^{\max}} \binom{\ell}{\ceil{\ell / 2}} \cdot \abs{\Delta(u, \ell)} \cdot p^{\ceil{\ell/2}} \ . \label{eq:rescLXclwithLminLmax}
  \end{align}
  Here~$L_{u}^{\min}$ and~$L_{u}^{\max}$ are the minimal and maximal lengths of paths in $\bigcup_{\ell = 1}^{\infty} \Delta(u, \ell)$ for $u \in \Vcldecextleft$.

  We rewrite the expression~\eqref{eq:rescLXclwithLminLmax} by partitioning the set~$\Vcldecextleft$ as follows.
  Let us define
  \begin{align}
    \Vextleft{s} = \left\{u \in \Vcldecextleft \mid L_{u}^{\min} = s \right\} \qquad \textrm{ for } \qquad s \geq 1 \ .
  \end{align}
  In other words, $\Vextleft{s}$ is the set of  external vertices~$u$ in~$\Vcldecextleft$ such that the minimal distance to~$\Vcldecextright$ is $s$ in terms of the length of paths connecting~$u$ and~$\Vcldecextright$.
  It is clear from the definition of~$\Tcldec$ that the set~$\Vextleft{s}$ is empty for all~$s > d$.
  Furthermore, the sets~$\{ \Vextleft{s} \}_{s = 1}^d$ are pairwise disjoint and their union is~$\Vcldecextleft$.
  Thus we can write the expression~\eqref{eq:rescLXclwithLminLmax} as
  \begin{align}\label{eq:rescLXclAs}
    \res_{\cLXcl}(p) = \sum_{s=1}^d A_s \qquad \textrm{ where } \qquad A_s = \sum_{u \in \Vextleft{s}} \sum_{\ell = s}^{L_{u}^{\max}} \binom{\ell}{\ceil{\ell / 2}} \cdot \abs{\Delta(u, \ell)} \cdot p^{\ceil{\ell / 2}} \ .
  \end{align}

  We first compute~$A_1$.
  It is clear from the definition of~$\Tcldec$ that~$u_1^{0}$ and~$u_1^{\frac{R+1}{2}}$ are the only vertices in~$\Vextleft{1}$, i.e.,
  \begin{align}\label{eq:tcldeclengthonevertices}
    \Vextleft{1} &= \left\{u_1^{0}, u_1^{\frac{R+1}{2}}\right\} \ .
  \end{align}
  In addition, the set of all simple paths through internal vertices starting at either~$u_1^{0}$ or~$u_1^{\frac{R+1}{1}}$ and ending at any external vertex in~$\Vcldecextright$ consists of only two length-$1$ paths, the path $\left(\left\{u_1^{0}, v_{1,1}^{0}\right\}\right)$ and the  path~$\left(\left\{u_1^{\frac{R+1}{2}}, v_{1,1}^{\frac{R+1}{2}}\right\}\right)$.
  In particular, we have
  \begin{align}\label{eq:tcldeclengthoneverticesmax}
    L_{u_1^{0}}^{\max} = L_{u_1^{\frac{R+1}{2}}}^{\max} = 1 \ .
  \end{align}
  According to~\eqref{eq:tcldeclengthonevertices} and~\eqref{eq:tcldeclengthoneverticesmax} with the definition of~$A_1$ from~\eqref{eq:rescLXclAs}, we have
  \begin{align}\label{eq:boundA1}
    A_1 &= \binom{1}{\ceil{1/2}} \cdot \abs{\Delta(u_1^{0}, 1)} \cdot p^{\ceil{1/2}} + \binom{1}{\ceil{1/2}} \cdot \abs{\Delta(u_1^{\frac{R+1}{2}}, 1)} \cdot p^{\ceil{1/2}}\\
        &= 2 \cdot p \ .
  \end{align}

  Next, we consider~$A_d$.
  It is clear from the definition of~$\Tcldec$ that~$u_j^{k} \not \in \Vextleft{d}$ for~$1 \leq j \leq d-1$ and $k \in \left\{0, \frac{R+1}{2}\right\}$.
  This gives the upper bound
  \begin{align}\label{eq:upperboundvextleftd}
    \abs{\Vextleft{d}} &\leq \abs{\Vcldecextleft \setminus \left\{u_j^{k} \mid 1 \leq j \leq d-1, k \in \left\{0,\frac{R+1}{2}\right\}\right\}}\\
    &= d \cdot \left(\frac{R-1}{2}\right) \ .
  \end{align}
  The fact that the decoding graph~$\Tcldec$ has vertices of degree at most 6 implies that
  \begin{align}\label{eq:upperboundpaths3D}
    \abs{\Delta(u, \ell)} \leq 5^{\ell} \qquad \textrm{ for } \qquad u \in \Vcldecextleft \ .
  \end{align}
  Combining~\eqref{eq:upperboundvextleftd},~\eqref{eq:upperboundpaths3D} and the fact that~$\binom{\ell}{\ceil{\ell/2}}\leq 2^{\ell}$ and $p^{\ceil{\ell/2}} \leq p^{\ell/2}$ with the definition~\eqref{eq:rescLXclAs} of~$A_s$ with $s=d$, we obtain
  \begin{align}\label{eq:boundAd}
    A_d &\leq d \cdot \left(\frac{R-1}{2}\right) \sum_{\ell = d}^{\infty} 2^{\ell} \cdot 5^{\ell} \cdot p^{\ell / 2} \\
        &= d \cdot \left(\frac{R-1}{2}\right) \sum_{\ell = d}^{\infty} (10 \sqrt{p})^{\ell} \\
        &= d \cdot \left(\frac{R-1}{2}\right) \frac{q^d}{1 - q}  \qquad \textrm{ where } \qquad q:= 10 \sqrt{p} \ .
  \end{align}
  Here the geometric series converges since the assumption $0 \leq p \leq  \frac{1}{400}$ implies $0 \leq q \leq \frac{1}{2}$.

  It remains  to calculate an upper bound of~$\sum_{s = 2}^{d-1}A_s$.
  Straightforward counting yields 
  \begin{align}\label{eq:upperboundvextleftmiddle}
    \abs{\Vextleft{s}} \leq 2 \cdot s \qquad \textrm{ for } \qquad 2 \leq s \leq d-1 \ .
  \end{align}
  It follows from an analogous argument as for the upper bound on~$A_d$ that
  \begin{align}\label{eq:boundAmiddle}
    \sum_{s = 2}^{d - 1} A_s &\leq \sum_{s = 2}^{d - 1} \left(2 \cdot s \sum_{\ell = s}^{\infty} 2^{\ell} \cdot 5^{\ell} \cdot p^{\ell / 2} \right) \\
    &= \sum_{s = 2}^{d - 1} \left(  2 \cdot s \cdot \frac{q^s}{1-q}  \right) \qquad \textrm{ where } \qquad q:= 10 \sqrt{p}\\
    &\leq \frac{4q^2}{(1-q)^2} + \frac{2q^3}{(1-q)^3} \ .
  \end{align}
  Here again, the geometric- and the arithmetico-geometric series converge since $0 \leq q \leq \frac{1}{2}$.
  Combining \eqref{eq:boundA1}, \eqref{eq:boundAd} and \eqref{eq:boundAmiddle}, we have
  \begin{align}
    \res_{\cLXcl}(p) &= \sum_{s=1}^d A_s \\
    &\leq 2p + \frac{4q^2}{(1-q)^2} + \frac{2q^3}{(1-q)^3} +d \cdot \frac{R-1}{2} \cdot \frac{q^d}{1 - q} \\
    &= \left[2 + \frac{400}{(1-q)^2} + \frac{200q}{(1-q)^3} + 50d(R-1)  \cdot \frac{(10\sqrt{p})^{d-2}}{1 - q}\right] \cdot p \ .\label{eq:rescLXclreciprocal}
  \end{align}
  Here we used the definition of $q$ for the last equality.
  We have $q \leq \frac{1}{2}$ and $\frac{1}{1-q} \leq 2$, hence we obtain the desired bound~\eqref{eq:rescLXclbound} from~\eqref{eq:rescLXclreciprocal}.
\end{proof}

The resilience function~$\res_{\cLZcl}(p)$ of the dual decoding graph and the recovery set~$\cLZcl$ can be upper bounded in a similar manner.
\begin{lemma}\label{lem:rescLZclbound}
  Consider the dual decoding graph~$\Tcldecdual=(\Vcldecdual, \Ecldecdual)$ with set of internal vertices~$\Vcldecdualint$ defined on the dual lattice~$\Ctildedual [d \times d \times R]$. Let $\cLZcl\subset \Ecldecdual$ be the recovery set introduced in Section~\ref{sec:decodinggraphs}. Then the resilience function~$\res_{\cLZcl}$ satisfies
  \begin{align}\label{eq:rescLZclbound}
    \res_{\cLZcl}(p) &\leq  \left[  2400 + 100d(R+1)(10\sqrt{p})^{d-2}  \right] \cdot p\ .
  \end{align}
  for any $p\in \interval{0}{\frac{1}{400}}$.
\end{lemma}
The proof mirrors that of Lemma~\ref{lem:rescLXclbound}
and is given in Appendix~\ref{sec:resiliencefunctiondual} for completeness.

\subsection{A lower bound on the success probability of the protocol}\label{sec:boundprobsucccl}
Now we state our main result,  a lower bound on  the success probability of our entanglement generation protocol. Here the success probability (denoted~$\nu_{(0,0)}$ in the following)  is the probability (taken over the random choice of local stochastic error~$E$ as well as the measurements outcomes) that the final state of the protocol is the Bell state~$\Phi=\Phi_{(0,0)}$. Note that this quantity is in one-to-one-correspondence with the fidelity of the (average) output state with the Bell state~$\Phi$. This is because the protocol always produces one of the four orthogonal Bell states. 

\begin{theorem}[Entanglement generation with  local stochastic noise]\label{thm:successprobability}
  Let $d \geq 2$, and let $R \geq 3$ be an odd integer.
  Suppose we run Algorithm~\ref{alg:ftes} on a noisy cluster state $E W\ket{0^{\cC}}$ on a lattice~$\cC\left[d\times d \times R\right]$, where the cluster state is corrupted by a local stochastic error~$E \sim \cN(p)$ with $p \in \interval{0}{\frac{1}{400}}$.
  Then there is a probability distribution~$\nu = (\nu_{(0,0)}, \nu_{(0,1)}, \nu_{(1,0)}, \nu_{(1,1)})$  such that the following holds:
  \begin{enumerate}[(i)]
  \item
    The resulting output state is the two-qubit Bell state~$\ket{\Phi_{(\alpha,\beta)}}$ with probability~$\nu_{(\alpha, \beta)}$ for every pair~$(\alpha,\beta) \in \{0,1\}^2$.
  \item The success probability $\nu_{(0,0)}$ of the entanglement generation protocol satisfies
    \begin{align}
    \nu_{(0,0)} \geq 1 - \left[4806 + 200dR (10\sqrt{p})^{d-2} \right] \cdot p \ . \label{eq:successprobabilitybound} 
    \end{align}
  \end{enumerate}
\end{theorem}
\begin{proof}
  Due to Lemma~\ref{lem:pmstateisbellbasis}, it suffices to establish an upper bound on $1-\nu_{(0,0)}$.
  Let us introduce the following four random variables:
  \begin{align}\label{eq:defofAZBZAXBX}
    \begin{aligned}
    A_Z &:= \ztwoinner{\{q_1,q_2\}}{\supp(E^Z)} \\
    B_Z &:= \ztwoinner{\supp(E_{\mathrm{gl}}) \oplus \mmatch_{\Tcldec}(\partial_{\Tcldec}\supp(E_{\mathrm{gl}}))}{\cLXcl} \\
    A_X &:= \ztwoinner{\{q_1,q_2\}}{\supp(E^X)} \\
    B_X &:= \ztwoinner{\supp(E_{\mathrm{gl}}^{*}) \oplus \mmatch_{\Tcldecdual}(\partial_{\Tcldecdual}\supp(E_{\mathrm{gl}}^{*}))}{\cLZcl} \ .
    \end{aligned}
  \end{align}
  Then Theorem~\ref{thm:successcondition} translates to
  \begin{align}\label{eq:nu00withAZBZAXBX}
    \nu_{(0,0)} = \Pr[(A_Z, B_Z) \in \{(0,0),(1,1)\} \textrm{ and } (A_X, B_X) \in \{(0,0),(1,1)\}] \ .
  \end{align}
  The union bound gives
  \begin{align}\label{eq:nu00unionbound}
    1-\nu_{(0,0)} \leq \Pr[(A_Z, B_Z) \in \{(0,1),(1,0)\}] + \Pr[(A_X, B_X) \in \{(0,1),(1,0)\}] \ .
  \end{align}
  Using the union bound on the first term of~\eqref{eq:nu00unionbound}, we have
  \begin{align}
    \Pr[(A_Z, B_Z) \in \{(0,1),(1,0)\}] &\leq \Pr[(A_Z, B_Z) = (0,1)] + \Pr[(A_Z, B_Z) = (1,0)] \\
    &\leq \Pr[B_Z = 1] + \Pr[A_Z = 1] \ .
  \end{align}
  Applying the same reasoning to the second term of~\eqref{eq:nu00unionbound} yields
  \begin{align}\label{eq:nu00unionboundAXAZBXBZ}
    1-\nu_{(0,0)} \leq \Pr[A_X = 1] + \Pr[A_Z = 1] + \Pr[B_X = 1] + \Pr[B_Z = 1] \ .
  \end{align}
  By the definition of~$A_X$ and (bi-)linearity of~$\ztwoinner{\cdot}{\cdot}$ together with the union bound, we have
  \begin{align}
    \Pr[A_X = 1] &\leq \Pr[\ztwoinner{\{q_1\}}{\supp(E^Z)} \oplus \ztwoinner{\{q_2\}}{\supp(E^Z)} = 1] \\
                 &\leq \Pr[\ztwoinner{\{q_1\}}{\supp(E^Z)}=1] + \Pr[\ztwoinner{\{q_2\}}{\supp(E^Z)}=1] \\
                 &\leq p + p = 2p\label{eq:boundAXcl}
  \end{align}
  where we used that $E^Z \sim \cN(p)$ is local stochastic noise with the same parameter as~$E$. In a similar manner, we obtain 
  \begin{align}\label{eq:boundAZcl}
    \Pr[A_Z = 1] = 2p \ .
  \end{align}
  We also have
  \begin{align}
    \Pr[B_Z=1] &= \Pr[\ztwoinner{\supp(E_{\mathrm{gl}}) \oplus \mmatch_{\Tcldec}(\partial_{\Tcldec}\supp(E_{\mathrm{gl}}))}{\cLXcl}=1] \\
               &\leq \res_{\cLXcl}(p) \\
               &\leq \left[ 2402 + 100d(R-1)(10\sqrt{p})^{d-2}  \right] \cdot p \                 \label{eq:boundBZcl}
  \end{align}
  where the definition of~$B_Z$, the fact that~$E_{\mathrm{gl}} \sim \cN(p)$, Proposition~\ref{prop:maincombinatorics} and Lemma~\ref{lem:rescLXclbound} are used.
  By  similar reasoning, we get
  \begin{align}
    \Pr[B_X = 1] &= \Pr[\ztwoinner{\supp(E_{\mathrm{gl}}^{*}) \oplus \mmatch_{\Tcldecdual}(\partial_{\Tcldecdual}\supp(E_{\mathrm{gl}}^{*}))}{\cLZcl}] \\
                 &\leq \res_{\cLZcl}(p) \\
                 &\leq \left[2400 + 100d(R+1)(10\sqrt{p})^{d-2}  \right] \cdot p \ , \label{eq:boundBXcl}
  \end{align}
  see Lemma~\ref{lem:rescLZclbound}. 
  Combining inequalities~\eqref{eq:boundAXcl},~\eqref{eq:boundAZcl},~\eqref{eq:boundBZcl} and~\eqref{eq:boundBXcl} together with \eqref{eq:nu00unionboundAXAZBXBZ}, we obtain the claim~\eqref{eq:successprobabilitybound}.
\end{proof}

Theorem~\ref{thm:successprobability} immediately implies the following corollary. It shows that for any sufficiently small constant noise-strength, a constant-fidelity Bell state can be established over a distance~$R$ which is exponential in~$d$. 

\begin{corollary}\label{cor:achievability}
  Let $p \in \linterval{0}{\frac{1}{5006}}$, $d \in \mathbb{N}$ with $d \geq 3$ and $R \in \mathbb{N}$ be odd with $R \geq 3$.
  Suppose we run Algorithm~\ref{alg:ftes} on the noisy cluster state $EW\ket{0^\cC}$ with local stochastic noise $E \sim \cN(p)$ in the lattice $\cC\left[ d \times d \times R \right]$ with
  \begin{align}\label{eq:upperboundonR}
      R \leq \frac{1}{d} \left( \frac{1}{10\sqrt{p}} \right)^{d-2} \ . 
  \end{align}
  Then the probability~$\nu_{(0,0)}$ of successful entanglement generation  satisfies 
  \begin{align}
    \nu_{(0,0)} \geq 1 - 5006p\ .
  \end{align}
\end{corollary}
\begin{proof}
  Due to the condition $p \in \linterval{0}{\frac{1}{5006}} \subset \interval{0}{\frac{1}{400}}$, the inequality~\eqref{eq:successprobabilitybound}  holds. Assuming that~$R$ satisfies the  upper bound~\eqref{eq:upperboundonR}, the claim follows. 
\end{proof}

\section{Limits on low-latency entanglement generation}~\label{sec:converse}
Here we establish limits on fault-tolerant long-range entanglement generation using adaptive  shallow (i.e., constant-depth) circuits
 along a line of repeater stations. Such circuits capture low-latency schemes where the entanglement is generated in a constant amount of time (up to local corrections).
We show that any such scheme which is resilient to arbitrary local stochastic noise of strength below some threshold can only establish constant-fidelity entanglement up to a certain distance~$R$: We show that~$R$ must be upper bounded by a function which is exponential in the number $m$~of qubits at each repeater station.

In Section~\ref{sec:converseboundlowlatency}, we establish a general converse bound for low-latency entanglement generation, applicable to all protocols realized by constant-depth adaptive circuits.  
In Section~\ref{sec:converseboundcluster}, we establish a converse bound for  any scheme based on the cluster state $W\ket{0^{\cC}}$ which uses the same measurement pattern and syndrome information as Algorithm~\ref{alg:ftes}, but possibly different classical decoding strategies. We discuss how these two converse bounds compare to each other and to our  achievability results for Algorithm~\ref{alg:ftes} from Section~\ref{sec:boundprobsucccl}.

\subsection{A general converse bound for low-latency entanglement generation\label{sec:converseboundlowlatency}}

\newcommand*{\CNOT}{\mathsf{CNOT}}
In the following, we consider constant-depth protocols for entanglement generation across a line of~$R$ site with $m$~qubits per site (i.e., per repeater). Constant-depth here means that the entanglement generation protocol can be written as a circuit of  depth independent of~$(m,R)$ composed of single-qubit gates and measurements, and two-qubit gates between any pair of qubits located at the same or two neighboring repeaters. Throughout, we consider adaptive circuits, i.e., circuits that apply unitaries which  are classically controlled by (functions of) measurement results. Classical computations (i.e., function evaluations) to determine the unitaries which are applied are not counted in the definition of circuit depth.  Furthermore, for our converse bound, it actually suffices to consider the circuit depth ``between repeaters'', i.e., gate sequences within a repeater station can be assumed to contribute only a single gate layer to the overall circuit depth~$\Delta$.

We assume that the considered protocols generate entanglement between a qubit~$q_1$ at the first repeater~($u_3=1$)  and a qubit~$q_2$ at the last repeater  ($u_3=R$). The main result of this section is the following converse. It shows that long-range entanglement cannot be generated fault-tolerantly by constant-depth adaptive circuits if the distance is more than exponential in the number $m$~of qubits at each repeater.

\begin{corollary}\label{cor:separabilityconverse}
Let $\pi$ be a depth-$\Delta$ circuit for entanglement generation between two qubits at distance~$R$ using~$m$ qubits per repeater. 
Suppose there is some constant threshold $p_0<1$ such  that~$\pi$ produces the Bell state~$\Phi$ with probability at least~$1/e\approx 0.368$ for any local stochastic noise of strength~$p\leq p_0$. Then 
\begin{align}
   R&\leq \left(\frac{3}{4p_0^2}\right)^{\Delta m}\ .
\end{align}
\end{corollary}

Corollary~\ref{cor:separabilityconverse} is an immediate consequence of the following general no-go result for entanglement generation with  limited-depth quantum circuits:
\begin{theorem}\label{thm:generalconverseboundone}
Let~$\pi$ be a circuit of depth~$\Delta$ for distance-$R$ entanglement generation using $m$~qubits per repeater. 
Let $p\in (0,1)$ be arbitrary. Then there is a local stochastic noise model of strength~$p$ 
acting on the circuit such that the probability that 
$\pi$ produces a separable state is lower bounded by 
\begin{align}
  p_{\mathrm{fail}}\geq 1-e^{-R(4p^2/3)^{\Delta m}}\ .
  \end{align}
\end{theorem}

Reformulated, Theorem~\ref{thm:generalconverseboundone} implies that  any protocol~$\pi$ that can withstand arbitrary local stochastic noise of strength~$p$ on the initial state and still succeeds (i.e., produces the Bell pair~$\Phi$) with probability~$p_{\mathrm{succ}}$ must satisfy the relationship
\begin{align}
R&\leq (\textfrac{3}{(4p^2)})^{\Delta m} \log\left(1/p_\mathrm{succ}\right)\ .
\end{align}
In particular, any protocol that produces constant-fidelity entanglement for any error strength below some threshold can at most reach  a distance exponential in~$m$. 

The proof of Theorem~\ref{thm:generalconverseboundone}
 proceeds by an intermediate step, where we consider bipartite-entanglement-based protocols. To describe this in detail, let $\pi$ be a distance-$R$-entanglement generation protocol using~$m$ qubits per repeater which can be implemented by a circuit of depth~$\Delta$.  

We assume without loss of generality that the only two-qubit gates that~$\pi$  uses between neighboring sites are $\CNOT$-gates. 
Note that each two-qubit unitary can be realized by a simple circuit consisting of at most three~$\CNOT$ gates and (arbitrary) single-qubit gates~\cite{vidaldawson04}. 
Each such $\CNOT$-gate can be realized by a gate-teleportation based circuit as shown in Fig.~\ref{fig:gateteleportationcircuitcnot}. This circuit consumes a shared  resource Bell state~$\Phi$. 
We note that the $\CNOT$-implementing circuit uses only local gates (where locality is with respect to the bipartition, i.e., we allow two-qubit gates at the ``source'' and ``target'' sites), local measurements, and classical communication. In particular, this shows that the $\CNOT$~gate can be implemented using a Bell state~$\Phi$ and LOCC-operations only.

Applying this substitution to every $\CNOT$-gate of~$\pi$ which acts between neighboring repeaters, we obtain a bipartite-entanglement-based protocol~$\pi'$ which has the following properties:
\begin{enumerate}[(i)]
\item\label{it:firstpropertypiprimedef}
The protocol~$\pi'$ uses a resource state~$(\Phi^{\otimes  t})_{A^tB^t}$ with~$t\leq  \Delta m$ copies of~$\Phi$ shared between any neighboring pair~$(S_j,S_{j+1})$ of sites.  Here $t$~is the maximal number of two-qubit $\CNOT$ gates executed between two neighboring sites. 
\item\label{it:secondpropertypiprimedef}
The protocol~$\pi'$ is an LOCC-operation between $S_1:S_2:\cdots :S_R$.
\item\label{it:thirdpropertypiprimedef}
Local stochastic noise on~$\pi'$ is equivalent to local stochastic noise on~$\pi$. This follows from the fact that the substitution rule in  Fig.~\ref{fig:gateteleportationcircuitcnot} is local by  suitably commuting errors as in the proof of~\cite[Lemma~11]{bravyiQuantumAdvantageNoisy2020}. In fact, we only need the following special version: A local stochastic error~$E\sim\cN(p)$ of strength~$p$ acting on 
the $t$~qubits $A^t=A_1\cdots A_t$ (i.e., one half of the Bell states) in the resource state~$(\Phi^{\otimes  t})_{A^tB^t}$ followed by an error-free execution of~$\pi'$
is equivalent to executing~$\pi$ with local stochastic noise of strengh~$\sqrt{p}$.  This is because 
in the circuit of Fig.~\ref{fig:gateteleportationcircuitcnot}, 
a single-qubit~$Y$ error on qubit~$A'$ is equivalent to a two-qubit error~$Z_AX_B$  after the~$\CNOT$, whereas  single-qubit~$X$- and $Z$-errors on~$A'$ become  single-qubit errors~$X_B$ and~$Z_A$, respectively.
\end{enumerate}

\begin{figure}
\begin{center}
\begin{quantikz}
\lstick{$A$}  &\quad    &\qw                           &\ctrl{1} &\qw      &\gate{Z} &\qw        &\qw \\
\lstick{$A^{\prime}$} &\quad  & \makeebit[-20][black]{$\Phi_{A^\prime B^\prime}$} &\targ{}  &\meter{} &\cw      &\cwbend{2} \\
\lstick{$B^{\prime}$} &\quad  &                               &\ctrl{1} &\gate{H} &\meter{} \vcw{-2}            \\
\lstick{$B$}  &\quad    &\qw                           &\targ{}  &\qw      &\qw      &\gate{X} &\qw
\end{quantikz}
\end{center}
\caption{Gate-teleportation~\cite{gottesmanchuang99} based protocol for realizing a two-qubit $\CNOT_{AB}$~gate between two sites, see also~\cite{EisertLocalimplement,piveteauCircuitKnittingClassical2022}. The circuit using a Bell state~$\Phi_{A^\prime B^\prime}$ as a resource, as well as local operations and measurements.\label{fig:gateteleportationcircuitcnot}}
\end{figure}
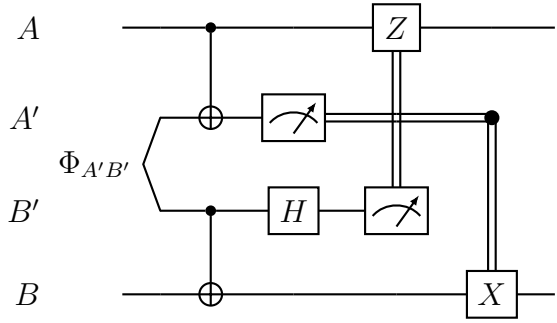

The following lemma gives an upper bound on such an  entanglement-based protocol~$\pi'$ for long-range entanglement generation. 
\begin{lemma}\label{lem:converseentanglementbased}
Consider $R$ sites $S_1,\ldots,S_R$ of the form 
\begin{align}
S_j&=\begin{cases}
A_1C_1\qquad &\textrm{ for }j=1\\
A_jB_jC_j &\textrm{ for } 1<j<R\\
B_{R}C_R&\textrm{ for }j=R\ ,
\end{cases}
\end{align}
where $A_j\cong B_j\cong (\mathbb{C}^2)^{\otimes t}$ for each~$j$. Let~$\pi'$ be a constant-depth adaptive quantum circuit  that acts on $S_1\cdots S_R$ and whose output state is on~$C_1C_R$.
Suppose~$\pi'$ satisfies properties~\eqref{it:firstpropertypiprimedef} and~\eqref{it:secondpropertypiprimedef}.

Assume that~$\pi'$ is an LOCC-operation with respect to the multipartition
\begin{align}
S_1:S_2:\cdots :S_R\ .
\end{align}
Suppose we apply~$\pi'$ to an initial state of the form
\begin{align}
\rho^{\mathrm{initial}}_{S_1\cdots S_R}:=\left(\bigotimes_{j=1}^{R-1}
(\proj{\Phi}^{\otimes t})_{A_jB_{j+1}}\right)\otimes
\bigotimes_{j=1}^R
\rho_{C_j}\ \label{eq:initialstaterhoerrormodelpsi}
\end{align}
where $\Phi\in \mathbb{C}^2\otimes \mathbb{C}^2$ is the two-qubit Bell state, and~$\rho_{C_1},\ldots,\rho_{C_R}$ are arbitrary. 
Then, for any $p\in (0,1)$,  there is a local stochastic error~$E\sim\cN(p)$ acting on the systems~$A^R=(A_1,\ldots,A_R)$ of the initial state~\eqref{eq:initialstaterhoerrormodelpsi}
 such that the final state~$\rho^{\mathrm{initial}}_{C_1C_R}$ is a separable state (with respect to the bipartition~$C_1:C_R$) with probability at least
\begin{align}
p_{\mathrm{fail}}&\geq 1-e^{-R\left(\textfrac{4p}{3}\right)^t}\  .
\end{align}
\end{lemma}

\begin{proof}
This statement immediately follows from the observation that a protocol~$\pi'$ can only convert preexisting entanglement in the initial state~$\rho^{\textrm{initial}}_{S_1\cdots S_R}$ to end up being on systems~$C_1:C_R$. Furthermore, the initial state  contains bipartite entanglement across sites only between
the pair of systems~$A_jB_{j+1}$, for $j=1,\ldots,R-1$. If this shared entanglement is lost due to noise then the state on~$C_1C_R$ remains separable throughout the execution of the  (ideal) protocol.

For concreteness, assume that the depolarizing noise channel
\begin{align}
\mathcal{E}_p(\rho) &=(1-p)\rho +p\tr(\rho) I/2
\end{align}
with $0 < p \leq 1$ is applied to each system~$A_j$ before application of the (ideal) protocol~$\pi'$. 
That is, we consider a noisy initial state of the form
\begin{align}
\rho^{\textrm{initial, noisy}}_{S_1\cdots S_R}:=\left(\bigotimes_{j=1}^{R-1}
\bigotimes_{k=1}^t
((\cE_p)_{A_j^{(k)}}\otimes \mathsf{id}_{B_{j+1}^{(k)}})
(\proj{\Phi}_{A_j^{(k)}B^{(k)}_{j+1}})\right)\otimes
\bigotimes_{j=1}^R
\rho_{C_j}\ .
\end{align}
For $p\in (0,1)$, we argue that this state can be written as a convex combination 
\begin{align}
\rho^{\textrm{initial, noisy}}_{S_1\cdots S_R}
&=p_0
\rho_{S_1\cdots S_R}+\sum_{j=1}^{R-1} p_j\rho^{(j)}_{S_1\cdots S_j:S_{j+1}\cdots S_R}
\end{align} 
of a certain  state~$\rho_{S_1\cdots S_R}$ 
and states~$\rho^{(j)}_{S_1\cdots S_j:S_{j+1}\cdots S_R}$ that are separable across  bipartitions of the form~$S_1\cdots S_j:S_{j+1}\cdots S_R$, for $j\in [R-1]$ where
\begin{align}
p_0 &\leq (1-p^{t})^R\\
&\leq e^{-p^t R}\ ,\label{eq:pjrminusonelowerbound}
\end{align}
where we used the inequality  $(1-x)^r\leq e^{-xr}$. 
 In particular, since each reduced state~$\rho^{(j)}_{S_1S_R}$ is separable, Eq.~\eqref{eq:pjrminusonelowerbound} and the fact that the (ideal) protocol~$\pi'$ is LOCC imply that the generated state is separable with probability at least~$1-e^{-p^tR}$. The claim follows from this because
 the depolarizing channel can be written as $\cE_p(\rho)= \left(1-\frac{3}{4} p\right) \rho + \frac{p}{4} X \rho X^{\dagger} + \frac{p}{4} Y \rho Y^{\dagger} + \frac{p}{4} Z \rho Z^{\dagger}$, i.e.,    is equivalent to 
applying each Pauli with probability~$p/4$. 
Thus applying~$\cE_p$ to each system~$A_j$ amounts to having a local stochastic error~$E\sim \cN(3p/4)$. (Substituting~$p$ by $4p/3$ gives the claim.)

Note that the depolarizing channel~$\cE_p$ is entanglement-breaking for any $p\geq 2/3$, see~\cite{ruskaiEB}. Eq.~\eqref{eq:pjrminusonelowerbound} is a consequence of the more trivial fact that for any $p\in (0,1)$, the channel~$\cE_p$
is a convex combination of the identity channel $\mathrm{id}$ and
the depolarizing channel~$\cE_{1}$ which is entanglement-breaking, that is, we have  
\begin{align}
  \mathcal{E}_p(\rho) = (1-p) \ \mathsf{id}(\rho) + p \ \mathcal{E}_{1}(\rho)\ .
\end{align}
With $p_0:=1-p$, $p_1:=p$ and the convention $\cE^{0}=\mathsf{id}$, this implies that for every $j\in [R-1]$
\begin{align}
\bigotimes_{k=1}^t
((\cE_p)_{A_j^{(k)}}\otimes \mathsf{id}_{B_j^{(k)}})
(\proj{\Phi}_{A_j^{(k)}B^{(k)}_{j+1}})
&=\sum_{x_1,\ldots,x_t\in \{0,1\}} \bigotimes_{k=1}^t p_{x_k}(\cE^{x_k}_{1}\otimes \mathsf{id})
(\proj{\Phi}_{A_j^{(k)}B^{(k)}_{j+1}})\\
&=\left(p^t \left( (\cE_{1}\otimes \mathsf{id}) (\proj{\Phi})\right)^{\otimes t}+\left(1-p^t\right)\sigma\right)_{A_jB_{j+1}}\ ,
\end{align}
where we collected all terms with $(x_1,\ldots,x_t)\neq 1^t$
in a state~$\sigma_{A_jB_{j+1}}$. 
Because~$\cE_{1}$ is entanglement-breaking,  the state~$\rho_{A_j:B_{j+1}}:=\left((\cE_{1}\otimes \mathsf{id}) (\proj{\Phi})\right)^{\otimes t}$ is separable across~$A_j:B_{j+1}$. 
We have  shown that 
\begin{align}
\rho^{\textrm{initial, noisy}}_{S_1\cdots S_R}&=\left(\bigotimes_{j=1}^{R-1}(
(1-q) \sigma_{A_jB_{j+1}}+q \rho_{A_j:B_{j+1}})
\right)\otimes
\left(
\bigotimes_{j=1}^R \rho_{C_j}
\right)\\
&=
(1-q)^{R-1} \left(\bigotimes_{j=1}^{R-1}\sigma_{A_jB_{j+1}}\right)
\otimes\left(\bigotimes_{j=1}^R \rho_{C_j}\right)
+\Sigma_{\mathrm{SEP}}
\end{align}
where 
\begin{align}
\Sigma_\mathrm{SEP}:=\sum_{(z_1,\ldots,z_{R-1})\in \{0,1\}^{R-1}\backslash \{0^{R-1}\}}
\bigotimes_{j=1}^{R-1}
(1-q)^{1-z_j}q^{z_j}\sigma^{1-z_j}_{A_jB_{j+1}}
\rho_{A_j:B_{j+1}}^{z_j}\otimes \left(
\bigotimes_{j=1}^R \rho_{C_j}
\right)\ 
\end{align}
and $q:=p^t$. 
Since every summand  in~$\Sigma_\mathrm{SEP}$ is separable with respect to the bipartition~$S_1\cdots S_j:S_{j+1}\cdots S_R$, the claim~\eqref{eq:pjrminusonelowerbound} follows. 
\end{proof}
Theorem~\ref{thm:generalconverseboundone} now follows by applying the reduction from a general protocol~$\pi$  to an entanglement-based protocol~$\pi'$.

\begin{proof}[Proof of Theorem~\ref{thm:generalconverseboundone}]
Suppose~$\pi$ is a depth-$\Delta$-circuit which 
generates distance-$R$-entanglement between two qubits using $m$~qubits per site. Without loss of generality (see the remark above), assume that all two-qubit gates between sites are $\CNOT$~gates. Let $t$~be the maximal number of nearest-neighbor $\CNOT$~gates used between two neighboring sites. Then $t\leq \Delta m$.

Let $\pi'$ be the entanglement-based protocol which uses the resource state~$\Phi^{\otimes t}$. 
Then, according to Lemma~\ref{lem:converseentanglementbased}, there is a local stochastic error~$E\sim\cN(p^2)$ acting on each half of the Bell states in the initial state~\eqref{eq:initialstaterhoerrormodelpsi} such that~$\pi'$ fails with probability
\begin{align}
    p_{\textrm{fail}}\geq 1-e^{-R(4p^2/3)^t}\geq 1-e^{-R(4p^2/3)^{\Delta m}}\ .
\end{align}
By property~\eqref{it:thirdpropertypiprimedef} of the entanglement-based protocol~$\pi'$, the local stochastic error~$E\sim\cN(p^2)$ on the resource state of~$\pi'$
is equivalent to (a certain) local stochastic noise of strength~$\cN(p)$ in the execution of~$\pi$. The claim follows from this.
\end{proof}

\subsection{A converse bound on cluster-state based schemes\label{sec:converseboundcluster}}
Corollary~\ref{cor:separabilityconverse} shows that 
any fault-tolerant low-latency scheme for entanglement generation  requires
\begin{align}
    m\geq \Omega(\log R)\label{eq:mlogR}
\end{align}
qubits per repeater.  For our cluster-state based scheme, the number~$m$ of qubits per repeater is given by~$m=\Theta(d^2)$, and our achievability result  requires~$d=\Omega(\log R)$. Thus our scheme works with 
\begin{align}
    m=\Theta(\log^2 R)\label{eq:mlogsquareR}
\end{align}
qubits per site.

Comparing~\eqref{eq:mlogR} with~\eqref{eq:mlogsquareR}, we observe a gap between our general no-go result and the achievability statement of  Corollary~\ref{cor:separabilityconverse}. In particular,~\eqref{eq:mlogR} leaves room for an improved achievability result, either by a better analysis of our scheme or by considering different schemes.

In this section, we show that, in fact, the
scaling~\eqref{eq:mlogsquareR} resulting for the cluster-state based scheme is optimal, i.e., our analysis is tight. We also find  that 
a better scaling cannot be obtained by modifying the classical decoding functions only: Any scheme which 
uses the cluster state $W\ket{0^{\cC}}$ with the lattice $\cC[d\times d \times R]$, applies the same measurement pattern and  relies on the same syndrome information (associated with observed boundaries on the decoding graph) 
requires at least~$m=\Omega(\log^2 R)$ qubits per repeater (i.e., within any plane with~$u_3$ fixed).

\begin{theorem}[Converse bound for cluster-state based schemes]\label{thm:mainconversebound}
  Let $0 < p \leq \frac{1}{4}$. Consider an (infinite) family of pairs~$(d,R)$ such that $d$ is even and $R$ satisfies
  \begin{align}
    R &\geq \left(\frac{1}{2\sqrt{p}}\right)^d\ .\label{eq:Depsilonpdcondition}
\end{align}
    Suppose we run Algorithm~\ref{alg:ftes} on the input state $E W\ket{0^{\cC}}$ in the cluster lattice $\cC[d \times d \times R]$ which is corrupted by i.i.d. single-qubit bit-flip noise with parameter~$p$, i.e., $E =\prod_{c\in \cC}X_c^{y_c}$ where~$\{y_c\}_{c\in\cC}$ are independent and identically distributed Bernoulli-$p$  random variables, i.e.,~$\Pr[y_c=1]=p$.  Then the  probability~$\nu_{(0,0)}$ of successful entanglement generation  satisfies
  \begin{align}
\nu_{(0,0)}\leq 3/4\qquad\textrm{ for all sufficiently large }d\ .
  \end{align}
The same statement applies to any scheme obtained from Algorithm on~\ref{alg:ftes} by replacing the function~$\mmatch$ by an (arbitrary) binary function.
\end{theorem}

\begin{proof}
We first argue that the conclusions of this Theorem hold for our protocol.   Recall the definition~\eqref{eq:defofEgl} of the operator~$E_{\mathrm{gl}}$ for any Pauli operator~$E$ on~$\cC$. 
  For the bit-flip noise model we consider here, there are no Pauli-$Z$ errors (i.e., $E^Z=0$), and thus
  \begin{align}
    E_{\mathrm{gl}} &= E^{X}|_{\Ecldec\cap\cA}=\prod_{c\in\Ecldec\cap\cA }X^{y_c}_c\ .
  \end{align}
  It follows that the two random variables $A_Z$ and $B_Z$ defined in \eqref{eq:defofAZBZAXBX} are given by
\begin{align}
    A_Z &=0 \label{eq:azconstanexpr}\\
    B_Z &= \ztwoinner{\supp(E^{X}|_{\Ecldec\cap\cA}) \oplus \mmatch_{\Tcldec}(\partial_{\Tcldec}\supp(E^{X}|_{\Ecldec\cap\cA}))}{\cLXcl}\ . \label{eq:bzexpressionm}
\end{align}
Eq.~\eqref{eq:nu00withAZBZAXBX}
  implies that
  \begin{align}
    1-\nu_{(0,0)}&\geq \Pr\left[(A_Z,B_Z)=(0,1)\right]\\
    &=\Pr\left[B_Z=1\right]\ ,
  \end{align}
  where we used~\eqref{eq:azconstanexpr}, i.e., the fact that $A_Z=0$ with probability~$1$. Our claim thus follows if we establish the lower bound
  \begin{align}
    \Pr\left[B_Z=1\right]\geq 1/4\ . \label{eq:lowerboundbztoprove}
    \end{align}
    Because of~\eqref{eq:bzexpressionm} and the definition of our noise model, we may express
  the probability of interest as
  \begin{align}
    \Pr\left[B_Z=1\right]=\Pr_Y\left[Y\in\cE^{(1)}\right]
  \end{align}
  where
  \begin{align}
    \cE^{(b)}:=\left\{
y \subseteq \Ecldec\ |\ 
\ztwoinner{(y \cap\cA) \oplus \mmatch_{\Tcldec}(\partial_{\Tcldec}(y \cap\cA))}{\cLXcl} = b \right\}\qquad\textrm{ for } b=0,1\ . \label{eq:cebdefinitionm}
  \end{align}
  (We will need $\cE^{(0)}$ below.) 
Here  $Y=\{Y_e\}_{e\in \Ecldec}$ are i.i.d. Bernoulli random variables with parameter~$p$. This is because
we have $Y=\supp(E^{X}|_{\Ecldec \cap \cA})$ by definition of the noise model.

We will prove~\eqref{eq:lowerboundbztoprove} by
showing that there is a subset~$\cE'\subset\cE$ that has probability at least
  \begin{align}
    \Pr_Y\left[Y\in \cE'\right]&\geq 1/4 \label{eq:lowerbundyeprimetoprove}
  \end{align}
  under this distribution.

  To construct~$\cE'$ let us call a subset~$y\subset\Ecldec$
  appropriate   if it contains exactly $d/2$~edges of any length-$d$, i.e., shortest-length path connecting the left and the right faces.
  We denote by~$\cE_{\mathsf{appr}}$ the set of appropriate subsets of~$\Ecldec$, and define
  \begin{align}
    \cE':=\cE^{(1)}\cap \cE_{\mathsf{appr}}\subset\cE^{(1)}
  \end{align}
  as the set of all~$y\in\cE$ that are appropriate. To state this formally,  let $e_1$ be the first canonical basis vector in~$\mathbb{R}^3$. For any $0\leq j\leq d-1$ and $1\leq k\leq \frac{R-1}{2}$, consider the two
  external vertices~$u_j^{k}$ and~$v_j^{k}$  of~$\Tcldec$ defined in \eqref{eq:tcldecleftlabel} and \eqref{eq:tcldecrightlabel}. We then consider the path~$P_{j,k}\subset \Ecldec$ obtained by traversing sequence of vertices
  \begin{align}
    u_j^k,u_j^k+2e_1,u_j^k+4e_1,\ldots,u_j^k+2(d-1)e_1,u_j^k+2de_1=v_j^k\ .
    \end{align}
  The path~$P_{j,k}$ is the shortest path (of length $d$) connecting two endpoints~$u_j^{k}$ and~$v_j^{k}$.
  A subset~$y\subset\Ecldec$ is then called 
  appropriate (i.e., $y\in\cE_{\mathsf{appr}}$) if there are $(j,k)$ with $0 \leq j \leq d-1$ and $1\leq k \leq \frac{R-1}{2}$ such that
  \begin{align}
|y\cap P_{j,k}|=d/2\ .\label{eq:weightconstraintpjky}
    \end{align}

  It remains to find a lower bound on
  \begin{align}
\Pr_Y\left[Y\in\cE^{(1)}\cap \cE_{\mathsf{appr}} \right]\ .\label{eq:toboundyceonequantity}
  \end{align}
  In order to do so, we use the following observation. 
  \begin{lemma}
    Consider the sets~$\cE^{(0)},\cE^{(1)}$ (cf.~\eqref{eq:cebdefinitionm}) and~$\cE_{\mathsf{appr}}$.         Let $b\in \{0,1\}$ be arbitrary.
    Then the following holds.
    \begin{enumerate}[(i)]
    \item
      For every~$y\in \cE^{(b)}\cap \cE_{\mathsf{appr}}$, there is an element~$\widehat{y}\in \cE^{(b\oplus 1)}\cap \cE_{\mathsf{appr}}$ such that
      $|y|=|\widehat{y}|$, i.e., $y$ and $\widehat{y}$ have the same Hamming weight.
    \item
      The map~$y\mapsto \widehat{y}$ is injective.
      \end{enumerate}
  \end{lemma}
  \begin{proof}
    We give the proof for~$b=0$ (the case $b=1$ is analogous). That is, consider~$y\in\cE^{(0)}\cap \cE_{\mathsf{appr}}$. Because~$y$ is appropriate, there
    is a -- not necessarily unique -- path~$P_{j,k}$ 
    such that~\eqref{eq:weightconstraintpjky} holds.
    In the following, we impose uniqueness by
    lexicographically ordering the set of pairs~$(j,k)$ with
    $0\leq j\leq d-1$ and $1\leq k\leq \frac{R-1}{2}$, and using the  first (lowest) such pair~$(j,k)$ with property~\eqref{eq:weightconstraintpjky}.   We then define
    \begin{align}
\widehat{y}:=y\oplus P_{j,k}
    \end{align}
    as the symmetric difference of~$y$ and $P_{j,k}$. It is easy to check that the map~$y\mapsto \widehat{y}$ is well-defined.
    
    Because~$|P_{j,k}|=d$,  Eq.~\eqref{eq:weightconstraintpjky} and the definition of~$\widehat{y}$ imply that $|\widehat{y}|=d/2$. In particular, we have $\widehat{y}\in \cE_{\mathsf{appr}}$.
With this property and the fact that the paths~$\{P_{j,k}\}_{j,k}$ are pairwise disjoint, it is easy to see that the map~$y\mapsto\widehat{y}$ has a left-inverse (obtained by applying the same construction again). Hence it is injective.

    It remains to show that~$\widehat{y}\in \cE^{(1)}$  for $y\in \cE^{(0)}\cap \cE_{\mathsf{appr}}$. 
    To this end, observe first that
  \begin{align}
    \partial_{\Tcldec}(\widehat{y}\cap\cA)&=\partial_{\Tcldec}(y \cap\cA)\ .\label{eq:identitypartialty}
  \end{align}
Eq.~\eqref{eq:identitypartialty} follows
  from the linearity of the boundary operator and the fact that the boundary of path is given by its endpoints. In our case, the endpoints of~$P_{j,k}$ are the external vertices~$\{u_j^k,v_j^k\}$ of $\Tcldec$, but these  do not contribute to the boundary~$\partial_{\Tcldec}P_{j,k}$  by definition of~$\partial_{\Tcldec}$.
  This establishes~\eqref{eq:identitypartialty}.

  We also use the following property of the path~$P_{j,k}$: We have 
  \begin{align}
    \ztwoinner{P_{j,k}\cap\cA}{\cLXcl} = 1 \ \label{eq:pjkecldeccapca}
  \end{align}
  since the subset~$\cLXcl$ lies in a plane  perpendicular to~$P_{j,k}$ with a single intersecting edge.   Combining~\eqref{eq:identitypartialty} and~\eqref{eq:pjkecldeccapca}, and using linearity, we have
  \begin{align}
    \ztwoinner{(\widehat{y}\cap\cA) \oplus \mmatch_{\Tcldec}(\partial_{\Tcldec}(\widehat{y}\cap\cA))}{\cLXcl}
    &= \ztwoinner{(y \cap\cA) \oplus \mmatch_{\Tcldec}(\partial_{\Tcldec}(y \cap\cA))}{\cLXcl}\oplus 1\\ \label{eq:yhatandyoplusone}
    &=0\oplus 1\\
    &=1
  \end{align}
  where we used the assumption that~$y\in \cE^{(0)}$ in the penultimate step. This shows that~$\widehat{y}\in \cE^{(1)}$ as claimed.
  \end{proof}

  Returning to the problem of bounding~\eqref{eq:toboundyceonequantity}, we observe that 
  \begin{align}
    \Pr_Y\left[Y\in\cE^{(1)}\cap \cE_{\mathsf{appr}} \right]
    &=\sum_{y\in \cE^{(1)}\cap \cE_{\mathsf{appr}}}P_Y(y)\\
    &=\sum_{y\in \cE^{(1)}\cap \cE_{\mathsf{appr}}}P_Y(\widehat{y})\\
    &\leq \sum_{y\in \cE^{(0)}\cap \cE_{\mathsf{appr}}}P_Y(y)\\
    &=\Pr_Y\left[Y\in\cE^{(0)}\cap \cE_{\mathsf{appr}} \right]\ ,
  \end{align}
  where we used that~$P_Y(y)$ only depends on the Hamming weight~$|y|$ of~$y$, and $|\widehat{y}|=|y|$. We also   used that $\widehat{y}\in\cE^{(0)}\cap \cE_{\mathsf{appr}}$ and the fact that the map~$y\mapsto \widehat{y}$ is injective to obtain the inequality. By symmetry, i.e., interchanging~$0$ and~$1$, we also have the converse inequality and thus the identity
  \begin{align}
    \Pr_Y\left[Y\in\cE^{(0)}\cap \cE_{\mathsf{appr}} \right]&=\Pr_Y\left[Y\in\cE^{(1)}\cap \cE_{\mathsf{appr}} \right]\ .
  \end{align}
  (In fact, it is possible to construct a Hamming-weight-preserving bijection between the two sets~$\cE^{(b)}\cap \cE_{\mathsf{appr}}$, $b=0,1$, but we do not need this here.)

  Because the sets~$\cE^{(0)}, \cE^{(1)}$ are disjoint, and the union is the set of all subsets of~$\Ecldec$, we conclude that
  \begin{align}
    \Pr_Y\left[Y\in\cE^{(1)}\cap \cE_{\mathsf{appr}} \right]
    &= \frac{1}{2}\left(
        \Pr_Y\left[Y\in\cE^{(0)}\cap \cE_{\mathsf{appr}} \right]+
        \Pr_Y\left[Y\in\cE^{(1)}\cap \cE_{\mathsf{appr}} \right]\right)\\
        &=\frac{1}{2}\Pr_Y\left[Y\in\cE_{\mathsf{appr}}\right]\ .\label{eq:yeappropr}
  \end{align}
We have 
  \begin{align}
    \Pr_Y\left[Y\in\cE_{\mathsf{appr}}\right]&=1-\Pr_Y\left[Y\not\in\cE_{\mathsf{appr}}\right]\\
    &=1-\Pr_Y\left[|Y\cap P_{j,k}|\neq\frac{d}{2}\textrm{ for all } (j,k)\right]\\
    &= 1 - \left(1 - \binom{d}{d/2} (2p)^{d/2}(1-2p)^{d/2}  \right)^{d\cdot \frac{R-1}{2}} \ .\label{eq:valueofprcEpauli}
  \end{align}
   In the last step, we used the fact that $Y$ consists of independent random Bernoulli variables, and the fact that the paths~$\{P_{j,k}\}_{j,k}$ are pairwise disjoint.
  
  By Stirling's formula~
  $\binom{2k}{k}=\left(1+o(1)\right)\frac{2^{2k}}{\sqrt{\pi k}}$ for $k\rightarrow\infty$ and the inequality $2p(1-2p) \geq p$ for $p \leq \frac{1}{4}$ we obtain 
  \begin{align}
    \binom{d}{d/2} (2p)^{d/2}(1-2p)^{d/2}
    &\geq \left(1+o(1)\right)
    \frac{2^d}{\sqrt{\pi d/2}} p^{d/2}\\
    &\geq 0.3 \fr{(2\sqrt{p})^d}{\sqrt{d}}\qquad\textrm{ for }\qquad d\textrm{ sufficiently large}\ .
  \end{align}
  We have 
  \begin{align}
      (1-x)^r &\leq e^{-rx}\qquad\textrm{ for all }\qquad x\in (0,1)\qquad\textrm{ and }\qquad r\in\mathbb{N}\ ,
  \end{align}
  hence
  \begin{align}
      \left(1 - \binom{d}{d/2} (2p)^{d/2}(1-2p)^{d/2}  \right)^{d\cdot \frac{R-1}{2}} 
      \leq \exp\left(- 0.15 \sqrt{d}\cdot (R-1)\cdot(2\sqrt{p})^d \right)\ .
  \end{align}
  We conclude with~\eqref{eq:valueofprcEpauli} and~\eqref{eq:yeappropr} that
  \begin{align}
    \Pr_Y\left[Y\in\cE^{(1)}\cap \cE_{\mathsf{appr}} \right]& \geq
    \frac{1}{2}\left(
    1-\exp\left(- 0.15 \sqrt{d} \cdot (R-1)\cdot(2\sqrt{p})^d\right) \right)\\
    &\geq 1/4
  \end{align}
  for sufficiently large~$d$ if
  \begin{align}
      \lim_{d\rightarrow\infty }\sqrt{d}\cdot (R-1) \cdot (2\sqrt{p})^d=\infty\ .\label{eq:Dsqrtpcondition}
  \end{align}
Condition~\eqref{eq:Dsqrtpcondition} is satisfied when~\eqref{eq:Depsilonpdcondition} holds.
  In particular, we obtain the claim~\eqref{eq:lowerbundyeprimetoprove} for sufficiently large~$d$.
  
  It remains to show that this argument also applies to protocols where the matching function~$\mmatch$ is replaced by a different function in Algorithm~\ref{alg:ftes}. This is because -- as discussed at the end of Section~\ref{sec:entanglementgenerationnoisy} --  Theorem~\ref{thm:successcondition} does not require $\mmatch$ to return a matching of a given input. Neither does~\eqref{eq:nu00withAZBZAXBX}, which is obtained from  Theorem~\ref{thm:successcondition}.

\end{proof}

\newpage
\appendix
\section{Products of stabilizer generators of the graph state\label{sec:appendixproductsofstabilizer}}
In this appendix, we give  proofs of 
Lemma~\ref{lem:productsofstabilizergeneratorsone} and Lemma~\ref{lem:logicalopsSXSZ} concerning products of stabilizer generators of the cluster state.

\begin{proof}[Proof of Lemma~\ref{lem:productsofstabilizergeneratorsone}]
  Let $\left(e_1,e_2,e_3\right)$ be the standard basis of $\mathbb{R}^3$. 
 \begin{enumerate}[(i)]
 \item Consider  an element~$u=(u_1,u_2,u_3)\in\Vcldecint$ with $u_3\not\in \{1,R\}$, i.e., $u \in \Veven$ with $\neigh(u) \subset \Eeven$.
 
 By definition, any $v\in \neigh(u)$ has the form 
  \begin{align}
  v=u+s_je_j\qquad\textrm{ with }\qquad j\in \{1,2,3\}\qquad\textrm{ and }\qquad s_j\in \{-1,1\}\ .\label{eq:uneighborform}
  \end{align}
   Moreover, for any such $v=u+s_je_j\in\neigh(u)$, any $w\in \neigh(v)$ takes the form 
   \begin{align}
   w=u+s_je_j+s_ke_k\qquad\textrm{ with  }\qquad k\in \{1,2,3\}\ , k\neq j\qquad\textrm{ and }\qquad s_k\in \{-1,1\}\ .\label{eq:updoubleneighborform}
   \end{align}
   This is because $\neigh(v)\subset \cC$ does not contain elements of the form~$(e,e,e)$, and thus $u\not\in\neigh(v)$.

   It follows from the definition~\eqref{eq:defofSu} of~$S^u$ for
$u\in\Vcldecint$ with $u_3\not\in \{1,R\}$ that
  \begin{align}
  S^u&=\prod_{v\in \neigh(u)} \left(Z_v\prod_{w\in \neigh(v)} X_w\right)\label{eq:operatorsu}\\
  &=\left(\prod_{v\in \neigh(u)}  Z_v\right)\Bigl(\prod_{
  \substack{(v,w):\\
  v\in\neigh(u)\\
  w\in\neigh(v)}} X_w\Bigr)\ .\label{eq:suneighuneighvprod}
  \end{align}
  Here we used that for any two distinct~$v,v'\in\neigh(u)$,  we have $v\not\in\neigh(v')$ (as follows immediately from~\eqref{eq:uneighborform}), implying that the product~\eqref{eq:operatorsu} does not contain a factor~$X_v$.
  
  Now suppose that $v,w$ are such that $v\in \neigh(u)$ and $w\in\neigh(v)$. Assume that~$v$ and $w$ are of the form~\eqref{eq:uneighborform} and~\eqref{eq:updoubleneighborform}, respectively, i.e.,
  \begin{align}
  v&=u+s_je_j\\
w&=u+s_j+s_ke_k\ .
  \end{align}
  Consider
  \begin{align}
  v'&:=u+s_ke_k\ .
    \end{align}
  Then it follows from~$u\in\Cprime$ and~$v,w\in \cC$ that $v'\in\cC$. In particular, we have 
  \begin{align}
  v'&\in\neigh(u)\\
  w&\in \neigh(v')\ .
  \end{align}
  This shows that associated with every pair~$(v,w)$ giving rise to a factor~$X_w$ in the product~\eqref{eq:suneighuneighvprod}, there is a pair~$(v',w)\neq (v,w)$ also contributing a factor~$X_w$. Thus all such factors cancel and we are left with
  \begin{align}
  S^u&=\prod_{v\in\neigh(u)}Z_v=\prod_{v \in \incident_{\Tcldec}(u)} Z_v\ \label{eq:suneightglu}
  \end{align}
  because $\neigh(u)=\incident_{\Tcldec}(u)$ by definition of~$\Tcldec$.   In particular, we have $\supp(S^u) = \incident_{\Tcldec}(u)$ as claimed.
  Let us argue that   for~$u\in\Vcldecint$ with $u_3\not\in \{1,R\}$ we have 
  \begin{align}
  \neigh(u)&=\incident_{\Tcldec}(u)\subset \cA\label{eq:tgluca}
  \end{align}
  (and thus in particular $\incident_{\Tcldec}(u)\cap \cB=\emptyset$). Observe that
  combining~\eqref{eq:suneightglu} with~\eqref{eq:tgluca} implies the claim~\eqref{eq:Suandincident} for $u \in \Vcldecint$ with $u_3 \notin \{1,R\}$. 
  To show~\eqref{eq:tgluca}, consider $v\in \neigh(u)$  of the form $v=u+s_je_j$. 
  Suppose that $v_3\in \{1,R\}$. Then this implies that either
  \begin{align}
u_3=2, j=3, s_j=-1\qquad\textrm{ or }\qquad u_3=R-1,j=3,s_j=1\ .
  \end{align}
  In both cases, it follows that $v\in \{(e,e,o)\}$ (since $u\in \{(e,e,e)\}$) and thus $v\in\cA$. 
  This implies~\eqref{eq:tgluca} since any $v\in \neigh(u)$ with $v_3\not\in\{1,R\}$ satisfies~$v\in\cA$.
\item Consider an element $u\in\Vcldecint$ with $u_3\in \{1,R\}$.
  We claim that
  \begin{equation}\label{eq:incidentuintoAandB}
  \begin{aligned}
    \{ u \} &= \incident_{\Tcldec}(u) \cap \cA \\
    \neigh(u) &= \incident_{\Tcldec}(u) \cap \cB \ .
  \end{aligned}
  \end{equation}
  Note that the claim~\eqref{eq:Suandincident} follows from~\eqref{eq:incidentuintoAandB} combined with the definition of $G_u$.
  Observe that $u$ is incident to an edge
  \begin{align}\label{eq:edgetowardsu3}
    \{u, v\} \qquad \textrm{ where } \qquad v =
    \begin{cases*}
      (u_1, u_2, 2) & \textrm{ if $u_3 = 1 $}\\
      (u_1, u_2, R-1) & \textrm{ if $u_3 = R$} \ .
    \end{cases*}
  \end{align}
  To describe the edges in the subgraph $\Tdec \times \{u_3\}$, let us define
  \begin{align}
    \Ctilde_{u_3} &= \{w \in \Ctilde \mid w_3 = u_3\} \\
    \cC_{u_3} &= \{w \in \cC \mid w_3 = u_3\} \ .
  \end{align}
  Then $\Ctilde_{u_3}$ and $\cC_{u_3}$ are in one-to-one correspondence with the set of two-dimensional sites and the set of locations of qubits introduced in Section~\ref{sec:surfacecodelattice}, respectively.
  For a site $w = (w_1, w_2, u_3) \in \Ctilde_{u_3}$, we denote by $\neigh_{u_3}(w)$ the set of nearest neighbors of $w$ in $\cC_{u_3}$. More explicitly,
  \begin{align}
    \neigh_{u_3}(w) = \left\{z \in \cC_{u_3} \mid \sum_{j=1}^2 \abs{w_j - z_j} = 1\right\} \ .
  \end{align}
  From the definition~\eqref{eq:tcldecinterior} of $\Vdecint$, it follows that $(u_1, u_2) \in \Vdecint$.
  Then it follows from the definition~\eqref{eq:tdecinterior} that
  \begin{align}
    \incident_{\Tdec \times \{u_3\}}(u) = \neigh_{u_3}(u) \ ,
  \end{align}
  where each sites in $\neigh_{u_3}(u)$ is the midpoint of an edge of the lattice $\Tdec \times \{u_3\}$.
  Furthermore, it is easy to see that
  \begin{align}
    \neigh_{u_3}(u) = \neigh(u) \ .
  \end{align}
  By construction of $\Tcldec$, the edge $\{u, v\}$ from~\eqref{eq:edgetowardsu3} and the edges in $\incident_{\Tdec \times \{u_3\}}$ form the complete set of the edges where $u$ is incident in $\Tcldec$. That is,
  \begin{align}\label{eq:incidentintoneighanditself}
    \incident_{\Tcldec}(u) = \neigh(u) \cup \{u\} .
  \end{align}
  The claim~\eqref{eq:incidentuintoAandB} follows from~\eqref{eq:incidentintoneighanditself} and the fact that $u \in \cA$ and $\neigh(u) \subset \cB$.
\end{enumerate}
  
  The cases for the dual decoding graph can be checked in analogous manner.
\end{proof}

\begin{proof}[Proof of Lemma~\ref{lem:logicalopsSXSZ}]
  We show the claim~\eqref{eq:sxdefinition}.
  Let $u = (u_1, u_2, u_3) \in \cLXcl \cap \cB$. Then $u_1 = 1$, and $u_2$, $u_3$ are even.
  The stabilizer generator $G_u$ is
  \begin{align}
    G_u = 
    \begin{cases*}
      X_{(1, u_2, u_3-1)}Z_{(1,u_2,u_3)}X_{(1, u_2+1,u_3)} X_{(1,u_2,u_3+1)} &\textrm{ if $u_2 = 0$}\\
      X_{(1, u_2, u_3-1)}Z_{(1,u_2,u_3)}X_{(1, u_2+1,u_3)} X_{(1, u_2-1,u_3)} X_{(u_1,u_2,u_3+1)} &\textrm{ if $2\leq u_2 \leq 2d-4 $}\\
      X_{(1, u_2, u_3-1)}Z_{(1,u_2,u_3)}X_{(1, u_2-1,u_3)} X_{(1,u_2,u_3+1)} &\textrm{ if $u_2 = 2d-2$} \ .
    \end{cases*}
  \end{align}
  Define the operators 
  \begin{equation}\label{eq:calGj}
    \begin{aligned}
    \cG_j :=
    \begin{cases*}
      \prod_{k=2,4,\dots,R-1} Z_{(1,0,k)}X_{(1,1,k)} &\textrm{ if $j = 0$} \\
      \prod_{k=2,4,\dots,R-1} Z_{(1,j,k)}X_{(1,j-1,k)}X_{(1,j+1,k)} &\textrm{ if $j = 2, 4, \dots, 2d-4$} \\
      \prod_{k=2,4,\dots,R-1} Z_{(1,r-1,k)}X_{(1,r-2,k)} &\textrm{ if $j = 2d-2$} \ .
    \end{cases*}
    \end{aligned}
  \end{equation}
  Then it is straightforward to check that
  \begin{align}
    \prod_{k=2,4, \dots, R-1} G_{(1,j,k)} = X_{(1,j,1)} \cG_j X_{(1,j,R)} \qquad \textrm{ for } \qquad j = 0, 2, \dots, 2d-2 \ .
  \end{align}
  We observe that for all even $j, j^{\prime}$ and  $u_3 \in \{1, R\}$, the site $(1, j, u_3)$ does not appear in any of the factors in the product $\cG_{j^{\prime}}$ in~\eqref{eq:calGj}.
  This implies that $X_{(1,j,u_3)}$ and $\cG_{j^{\prime}}$ commute, and thus
  \begin{align}
    S^X &= \prod_{u \in \cLXcl \cap \cX} G_u \\
    &= \prod_{j=0, 2, \dots, 2d-2} \left( \prod_{k=2, 4, \dots, R-1} G_{(1, j, k)} \right) \\
    &= \prod_{j=0, 2, \dots, 2d-2}  X_{(1,j,1)} \cG_j X_{(1,j,R)} \\
    &= \left(\prod_{j=0, 2, \dots, 2d-2}X_{(1,j,1)}X_{(1,j,R)} \right) \left( \prod_{j=0, 2, \dots, 2d-2} \cG_j \right) \\
    &= X_{q_1}X_{q_2} X(\cLXcl \cap \cB) \left( \prod_{j=0, 2, \dots, 2d-2} \cG_j \right) \ .\label{eq:SXintermediate}
  \end{align}
  It is easy to check that in $\prod_{j=0,2,\dots, 2d-2} \cG_j$  all single-qubit Pauli operators commute and that all Pauli-$X$ are cancelled.
  As a result,  we have
  \begin{align}
    \prod_{j=0, 2, \dots, 2d-2} \cG_j = \prod_{\substack{j=0, 2, \dots, 2d-2\\k=2,4,\dots,R-1}} Z_{(1,j,k)} = Z(\cLXcl \cap \cA) \ .\label{eq:prodofcGj}
  \end{align}
  The claim~\eqref{eq:sxdefinition} follows from~\eqref{eq:SXintermediate} and~\eqref{eq:prodofcGj}.
  The claim~\eqref{eq:szdefinition} can be checked in analogous manner.
\end{proof}

\section{An upper bound on the resilience function~
$\res_{\cL_Z}(p)$ \label{app:upperboundresiliencetdecdual}}

\begin{proof}[Proof of Lemma~\ref{lem:upperboundresiliencetdecdual} ]
   We label the external vertices of $\Tdecdual$ as
  \begin{align}
    \Vdecdualext = \{u_1, \dots, u_d, v_1, \dots, v_{d-1}\}
    \qquad \textrm{with} \qquad 
    u_j = (2j - 1, 2j - 1) \textrm{ and }v_k = (2k + 1, -1) \ ,
  \end{align}
  see Fig.~\ref{fig:Tdecdual}
  Clearly, a simple path $P\in \Zpathext(\Tdecdual)$ with external endpoints through internal vertices satisfies $\ztwoinner{P}{\cL_Z}=1$
  if and only if it starts at a vertex~$u_j$ (with $j\in \{1,\ldots,d\}$) and ends at a vertex~$v_k$ (with $k\in \{1,\ldots,d-1\}$) without visiting any other external vertex.
  For $j\in \{1,\ldots,d\}$ and $\ell\in \mathbb{N}$, let $\Delta(u_j,\ell)$ be the set of all such paths~$P\in \Zpathext(\Tdecdual)$ that start at $u_j$ and end at some vertex in the set~$\{v_k\}_{k=1}^{d-1}$.
  By the same argument as used for showing~\eqref{eq:deltaformofresLx}, we have 
\begin{align}\label{eq:rescLZclosedform}
  \res_{\cL_Z}(p)
 &=\sum_{j=1}^{d}\sum_{\ell=L^{\min}_j}^{L^{\max}_j} \binom{\ell}{\lceil \ell/2\rceil}\cdot
 |\Delta(u_j,\ell)|\cdot p^{\lceil \ell/2\rceil} \ .
\end{align}
Here $L^{\min}_j$ and $L^{\max}_j$ are the minimal and maximal lengths of a path~$P \in \Zpathext(\Tdecdual)$ starting at $u_j$ and ending in~$\{v_k\}_{k=1}^{d-1}$.

The set $\Delta(u_1,1)$ is empty and the set $\Delta(u_1, 2)$ consists of a single path, i.e., $|\Delta(u_1,1)|=0$ and $\abs{\Delta(u_1, 2)} = 1$. For $j\in \{2,\ldots,d-1\}$, we use that the graph~$\Tdecdual$ has vertices of degree at most~$4$, which implies that
\begin{align}
|\Delta(u_j,\ell)|\leq 3^{\ell-1}\qquad\textrm{ for all }\qquad \ell\in\mathbb{N}\ 
\end{align}
as in the proof of Lemma~\ref{lem:upperboundresiliencetdec}. 
Again using
\begin{align}
L_j^{\min}\geq j\ .
\end{align}
and the inequalities $p^{\lceil \ell/2\rceil}\leq p^{\ell/2}$ and $\binom{\ell}{\lceil\ell/2 \rceil}\leq 2^\ell$ for $\ell > 1$ in the expression~\eqref{eq:rescLZclosedform}, we have that
\begin{align}
  \res_{\cL_Z^{*}} &\leq \binom{2}{1} 1 \cdot p
                 + \sum_{\ell = 3}^{L_1^{\max}} 2^\ell  3^{\ell-1} p^{\ell/2}
                 + \sum_{j = 2}^{d}\sum_{\ell = j}^{L_j^{\max}} 2^\ell  3^{\ell-1} p^{\ell/2} \\
  &\leq 2p + \frac{1}{3}\sum_{\ell = 3}^\infty q^\ell + \frac{1}{3}\sum_{j = 2}^d \sum_{\ell = j}^\infty q^{\ell} \qquad \textrm{ where } \qquad q:=6 \sqrt{p}\\
  &= 2p + \frac{q^3}{3(1-q)} + \frac{q^2-q^{d-1}}{3(1-q)^2}. \label{eq:rescLZupperboundwithq}
\end{align}
Again, $q$ satisfies $q \leq 6 \sqrt{\frac{1}{144}} \leq \frac{1}{2}$.
For such $q$, the inequalities $\frac{q}{(1-q)} \leq 1$ and $\frac{1}{(1-q)^2} \leq 4$ hold and we conclude that
\begin{align}
  \res_{\cL_Z^{*}} 
  &\leq 2p + \frac{1}{3} \cdot (6\sqrt{p})^2 + \frac{4}{3} \cdot (6\sqrt{p})^2 \\
  &= 38p \ .
\end{align}
\end{proof}

\section{An upper bound on the resilience function~$\res_{\cLZcl}(p)$
\label{sec:resiliencefunctiondual}}
In this appendix, 
we derive an upper bound on the resilience function~$\res_{\cLZcl}(p)$. That is, we  give the 
\begin{proof}[Proof of Lemma~\ref{lem:rescLZclbound}]
  The proof is analogous to that of Lemma~\ref{lem:rescLXclbound}. As before, it is sufficient to consider the set~$\Zpathext(\Tcldecdual)$ in order to compute the resilience~$\res_{\cLZcl}(p)$.
  We partition the set of external vertices
  \begin{align}
    \Vcldecdualext = \Vcldecdualextbottom \cup \Vcldecdualexttop
  \end{align}
  into the set~$\Vcldecdualextbottom$ of vertices~$(u_1,u_2,u_3)$ forming the ``bottom'' boundary with $u_2=-1$ and  the set~$\Vcldecdualexttop$ remaining vertices which form the boundaries on the ``top'' ($u_2=r$), ``front'' ($u_1=1$) and ``back'' ($u_1=R$).
  We label the external vertices as
  \begin{align}
    \Vcldecdualextbottom &= \left\{v_i^{k} \mid 1 \leq i \leq d-1, k \in \{1,D\}\right\} \cup \left\{ v_i^{k} \mid 0 \leq i \leq d-1, 2 \leq k \leq \frac{R-1}{2} \right\} \\
    \Vcldecdualexttop    &= \left\{ u_{i,j}^{k} \mid 1 \leq i \leq j, 1 \leq j \leq d-1, k \in \{1,\frac{R-1}{2}\} \right\} \\
    & \quad \cup \left\{ u_i^{k} \mid 1 \leq i \leq d, 2 \leq k \leq \frac{R-1}{2} \right\} \cup \left\{ u_d^{1}, u_d^{\frac{R-1}{2}} \right\} \ ,
  \end{align}
 see Fig.~\ref{fig:internalexternalverticescl}.
  Written out, the external vertices are located at the following (dual) sites:
  \begin{align}
    v_i^{k} &= (2i+1, -1, 2k-1) \\
    u_{i,j}^{k} &= (2i-1, 2j-1, 2k-1) \\
    u_i^{k} &= (2i-1, 2d-1, 2k-1) \ .
  \end{align}

  We introduce some subsets of paths in~$\Zpathext(\Tcldecdual)$.
  Let us denote by~$\Delta(v, \ell)$ the set of all simple paths through internal vertices starting at~$v \in \Vcldecdualextbottom$ and ending at any vertex in~$\Vcldecdualexttop$. 
  Then, by the similar argument as for showing \eqref{eq:zpathexttcldecdecompose}, we have
  \begin{align}
    \left\{ P \in \Zpathext(\Tcldecdual) \mid \abs{P} = \ell \textrm{ and } \ztwoinner{P}{\cLZcl} = 1 \right\}
    = \bigcup_{v \in \Vcldecdualextbottom} \Delta(v, \ell) \qquad \textrm{ for } \qquad \ell \geq 1 \ ,
  \end{align}
  and the subsets of paths~$\left\{ \Delta(v, \ell) \right\}_{v \in \Vcldecdualextbottom, \ell \geq 1}$ are pairwise disjoint.
  Let $L_v^{\min}$ and $L_v^{\max}$ be the minimal and maximal lengths of paths in $\bigcup_{\ell = 1}^{\infty}\Delta(v, \ell)$ for~$v \in \Vcldecdualextbottom$,
  and let us define
  \begin{align}
    \Vextbottom{s} = \left\{ v \in \Vcldecdualextbottom \mid L_v^{\min} = s \right\} \ .
  \end{align}
  It is easy to check from the definition of~$\Tcldecdual$ that for each vertex~$v=(v_1,v_2,v_3) \in \Vcldecdualextbottom$ there exists a path of length at most~$d$ which starts from~$v$ and ends at some~$u \in \Vcldecdualexttop$ with $(v_1,v_3) = (u_1,u_3)$, and this guarantees that the set~$\Vextbottom{s}$ is empty for all $s > d$.
  Moreover, we observe that~$\Vextbottom{1}$ is empty.
  With an analogous argument as in the proof of  Lemma~\ref{lem:rescLXclbound}, we have
  \begin{align}\label{eq:rescLZclasBs}
    \res_{\cLZcl}(p) = \sum_{s=2}^d B_s \qquad \textrm{ where } \qquad B_s = \sum_{ v \in \Vextbottom{s} } \sum_{\ell = s}^{L_v^{\max}} \binom{\ell}{\ceil{\ell/2}} \cdot \abs{\Delta(v, \ell)} \cdot p^{\ceil{\ell/2}} \ .
  \end{align}
  Here we note that each vertex in the graph~$\Tcldecdual$ has degree at most~$6$, and this implies similarly as before that
  \begin{align}\label{eq:upperboundpaths3Ddual}
    \abs{\Delta(u,\ell)} \leq 5^{\ell} \qquad \textrm{ for } \qquad  u \in \Vcldecdualextbottom \ .
  \end{align}
  We first calculate an upper bound on~$B_d$.
  Due to the   bound
  \begin{align}
    \abs{\Vextbottom{d} } \leq \abs{\Vcldecdualextbottom} \leq d \cdot \frac{R+1}{2}
  \end{align}
  on the number of external vertices on the bottom, together with the bound~\eqref{eq:upperboundpaths3Ddual} and the fact that $\binom{\ell}{\ceil{\ell/2}} \leq 2^l$ and $p^{\ceil{\ell/2}} \leq p^{\ell / 2}$, we have
  \begin{align}
    B_d &\leq d \cdot \frac{R+1}{2} \sum_{\ell = d}^{\infty} 2^{\ell} \cdot 5^{\ell} \cdot p^{\ell / 2} \\
        &\leq d \cdot \frac{R+1}{2} \cdot  \frac{q^d}{1-q} \qquad \textrm{ where } \qquad q := 10 \sqrt{p} \ .\label{eq:boundBd}
  \end{align}
  Here the geometric series converges since $0 \leq q = 10 \sqrt{p}\leq \frac{1}{2}$

  It remains to calculate an upper bound of~$\sum_{s=2}^{d-1} B_s$. It is easy to see that 
  \begin{align}
    \abs{\Vextbottom{s}} \leq 2 \cdot s \qquad \textrm{ for } \qquad 2 \leq s \leq d-1 \ . 
  \end{align}
  By analogous reasoning as for showing~\eqref{eq:boundBd}, we have
  \begin{align}
    \sum_{s=2}^{d-1} B_s & \leq \sum_{s=2}^{d-1} \left( 2 \cdot s \sum_{\ell = s}^{\infty} 2^{\ell} \cdot 5^{\ell} \cdot p^{\ell / 2} \right) \\
                         &= \sum_{s=2}^{d-1} 2 \cdot s \cdot \frac{q^s}{1-q} \\
                         &\leq \frac{4q^2}{\left( 1-q \right)^2} + \frac{2q^3}{\left( 1-q \right)^3} \ . \label{eq:boundBmiddle}
  \end{align}
  Again, the geometric series and the arithmetico-geometric series converge since $0 \leq q \leq \frac{1}{2}$.
  Inserting~$q=10\sqrt{p}$,~\eqref{eq:boundBd} and~\eqref{eq:boundBmiddle} into~\eqref{eq:rescLZclasBs}, we obtain
  \begin{align}
    \res_{\cLZcl}(p) &\leq \frac{4q^2}{(1-q)^2} + \frac{2q^3}{(1-q)^3} + d \cdot \frac{R+1}{2} \cdot \frac{q^d}{1-q}  \\
    &\leq \left[ \frac{400}{(1-q)^2} + \frac{200q}{(1-q)^3} + 50d(R+1) \cdot \frac{(10\sqrt{p})^{d-2}}{1-q} \right] \cdot p  .\label{eq:rescLZclreciprocal}
  \end{align}
  The bound~\eqref{eq:rescLZclbound} follows from~\eqref{eq:rescLZclreciprocal} because $q \leq \frac{1}{2}$ and $\frac{1}{1-q} \leq 2$.
\end{proof}

\subsubsection*{Acknowledgements}
SC and RK gratefully acknowledge support by the European Research Council under grant agreement no. 101001976 (project EQUIPTNT).
Figures were produced using VESTA~\cite{vesta}. RK thanks Isaac Kim for discussions on surface codes. 
\bibliographystyle{plain}
\bibliography{ref}
\end{document}